\documentclass[11pt]{amsart} 
\usepackage{amscd,amssymb,amsxtra}
\usepackage[mathscr]{eucal}
\usepackage{comment}
\usepackage{color}
\usepackage{enumerate}
\makeatletter
\let\@secnumfont\bfseries
\def\section{\@startsection{section}{1}%
  \z@{4\linespacing\@plus\linespacing}{\linespacing}%
  {\bfseries\centering}}
\def\introsection{\@startsection{section}{1}%
  \z@{3\linespacing\@plus\linespacing}{\linespacing}%
  {\bfseries\centering}}
\def\subsection{\@startsection{subsection}{2}%
   \z@{1.25\linespacing\@plus.7\linespacing}{.5\linespacing}%
   {\normalfont\bfseries}}
\def\subsectionsinline{\def\subsection{\@startsection{subsection}{2}%
  \z@{1\linespacing\@plus.7\linespacing}{-.5em}%
  {\normalfont\bfseries}}}

\makeatother

\theoremstyle{definition}
\newtheorem{definition}[equation]{Definition}
\newtheorem{example}[equation]{Example}

\newtheorem{construction}[equation]{Construction}

\newtheorem*{definition*}{Definition}
\newtheorem*{example*}{Example}
\newtheorem*{problem*}{Problem}
\newtheorem*{exercise*}{Exercise}
\newtheorem*{question*}{Question}
\newtheorem*{construction*}{Construction}

\theoremstyle{remark}

\newtheorem{remark}[equation]{Remark}

\newtheorem*{note*}{Note}
\newtheorem*{notation*}{Notation}
\newtheorem*{remark*}{Remark}
\newtheorem*{data*}{Data}

\theoremstyle{plain}
\newtheorem{theorem}[equation]{Theorem}
\newtheorem{corollary}[equation]{Corollary}
\newtheorem{lemma}[equation]{Lemma}
\newtheorem{proposition}[equation]{Proposition}

\newtheorem{hypothesis}[equation]{Hypothesis}

\newtheorem*{theorem*}{Theorem}
\newtheorem*{corollary*}{Corollary}
\newtheorem*{lemma*}{Lemma}
\newtheorem*{proposition*}{Proposition}
\newtheorem*{conjecture*}{Conjecture}
\newtheorem*{claim*}{Claim}
\newtheorem*{proposal*}{Proposal}
\newtheorem*{conclusion*}{Conclusion}
\newtheorem*{hypothesis*}{Hypothesis}

\numberwithin{equation}{section}

\definecolor{refkey}{rgb}{0,.6,.4}

\renewcommand{\:}{\colon}

\DeclareMathOperator{\Aut}{Aut}
\newcommand{\CC}{{\mathbb C}}

\newcommand{\EE}{\mathbb E}

\DeclareMathOperator{\End}{End}

\newcommand{\HH}{{\mathbb H}}
\DeclareMathOperator{\Hom}{Hom}
\DeclareMathOperator{\id}{id}
\DeclareMathOperator{\Map}{Map}

\newcommand{\PP}{{\mathbb P}}
\DeclareMathOperator{\pt}{pt}

\newcommand{\RR}{{\mathbb R}}
\newcommand{\TT}{\mathbb T}
\DeclareMathOperator{\Spin}{Spin}

\newcommand{\ZZ}{{\mathbb Z}}

\newcommand{\chiup}{\raise.5ex\hbox{$\chi$}}
\newcommand{\cir}{S^1}

\newcommand{\inv}{^{-1}}
\newcommand{\mstrut}{^{\vphantom{1*\prime y\vee M}}}

\newcommand{\res}[1]{\negmedspace\bigm|\mstrut_{#1}}
\newcommand{\temsquare}{\raise3.5pt\hbox{\boxed{ }}}

\newcommand{\zmod}[1]{\ZZ/#1\ZZ}

\newcommand{\zt}{\zmod2}

\usepackage[all,2cell,dvips]{xy}\renewcommand{\cir}{\ensuremath{S^1}}
\usepackage{graphicx}
\usepackage[colorlinks,citecolor=refkey]{hyperref}

\DeclareMathOperator{\Ad}{Ad}
\DeclareMathOperator{\Aff}{Aff}
\DeclareMathOperator{\Cliff}{Cl}
\DeclareMathOperator{\Det}{Det}
\DeclareMathOperator{\Euc}{Euc}

\DeclareMathOperator{\Pin}{Pin}
\DeclareMathOperator{\Rep}{Rep}
\DeclareMathOperator{\Sp}{Sp}
\DeclareMathOperator{\Triv}{Triv}
\DeclareMathOperator{\Vect}{Vect}
\DeclareMathOperator{\domain}{domain}
\DeclareMathOperator{\ev}{ev}
\DeclareMathOperator{\image}{image}
\DeclareMathOperator{\pfaff}{pfaff}
\DeclareMathOperator{\sign}{sign}
\def\qi{i}
\def\qj{j}
\def\qk{k}
\newcommand{\AVI}{\Aut_k(V,I)}
\newcommand{\CMbC}{\CMb\otimes \CC}
\newcommand{\CMb}{\Cliff(\sM,b)}
\newcommand{\CTt}{CT~type}
\newcommand{\Clc}[1]{\Cliff^{\CC}_{#1}}
\newcommand{\Cl}[1]{\Cliff_{#1}}
\newcommand{\Eb}[1]{\sE_{[#1]}}
\newcommand{\GH}{\GS}
\newcommand{\GSF}{\mathscr{R}\TPF}
\newcommand{\GSI}{\mathscr{R}\TPI}
\newcommand{\GSZ}{\mathscr{R}\TPZ}
\newcommand{\GS}{\mathscr{R}\TP}
\newcommand{\Gptc}{(G,\phi ,\tau ,c)}
\newcommand{\Gpt}{(G,\phi ,\tau )}
\newcommand{\HDN}{\sH_{\textnormal{DN}}}
\newcommand{\HF}{\sH_{\textnormal{Fock}}}

\newcommand{\Kl}[2]{\widetilde{L}\mstrut _{[#1],#2}}
\newcommand{\Ll}[2]{L\mstrut _{[#1],#2}}
\newcommand{\Lv}{X\mstrut _\Pi }
\newcommand{\MM}{\mathbb{M}}
\newcommand{\Mb}{(\sM,b)}
\newcommand{\Mod}[1]{#1\textnormal{-Mod}}
\newcommand{\OMb}{O(\sM,b)}
\newcommand{\PH}{\PP\sH}
\newcommand{\QAut}{\Aut_{\textnormal{qtm}}}
\newcommand{\Qsymmetry}{QM~symmetry\ }
\newcommand{\RtpZ}{\RR/2\pi \ZZ}
\newcommand{\SH}{\mathscr{S}\mstrut _{\sH}}
\newcommand{\TPF}{\mathscr{T}\mathscr{P}\mstrut _F}
\newcommand{\TPI}{\mathscr{T}\mathscr{P}\mstrut _I}
\newcommand{\TPZ}{\mathscr{T}\mathscr{P}\mstrut _Z}
\newcommand{\TP}{\mathscr{T}\mathscr{P}}
\newcommand{\VXT}{s\!\Vect^\nu _{G''}(\tXL)}
\newcommand{\VpV}{V_{\textnormal{boost}} \oplus V_{\textnormal{trans}}}
\newcommand{\XL}{\Lv}
\newcommand{\Xt}{X^\tau }
\newcommand{\bC}{\overline{C}}

\newcommand{\bP}{\overline{P}}
\newcommand{\bS}{\overline{S}}
\newcommand{\bT}{\overline{T}}
\newcommand{\bVg}{\overline{V^g}}
\newcommand{\bW}{\overline{W}}
\newcommand{\bc}{\bar{c}}
\newcommand{\be}{\bar{e}}

\newcommand{\bl}{\bar\ell }
\newcommand{\bol}[1]{\bold{#1}}

\newcommand{\cc}[2]{{}_{\phantom{c}}^{\phi (#1)}\!#2}
\newcommand{\co}{_{\textnormal{c.o.}}}
\newcommand{\cpt}{\textnormal{cpt}}
\newcommand{\db}{\dot\beta }
\newcommand{\df}{\dot{f}}
\newcommand{\dm}{\dot\mu }
\newcommand{\fC}{\phi _{\sC}}
\newcommand{\fG}{\phi _G}
\newcommand{\fh}{\mathfrak{h}}
\newcommand{\fk}{\mathfrak{k}}
\newcommand{\fp}{\mathfrak{p}}

\newcommand{\go}{\frak{g}_1}
\newcommand{\gpd}{/\!/} 
\newcommand{\gp}{\mathcal{G}}

\newcommand{\hKp}[1]{{}_{\phantom{G}}^{#1}\kern-.26em K}
\newcommand{\hK}[2]{{}_{\phantom{G}}^{#1}\kern-.26em K^{#2}}
\newcommand{\hP}{\widehat{P}}
\newcommand{\hT}{\widehat{T}}
\newcommand{\half}{\frac 12}
\newcommand{\hf}{\hat{f}}
\newcommand{\hrK}[2]{{}_{\phantom{G}}^{#1}\kern-.26em \widetilde{K}^{#2}}
\newcommand{\hs}{\hat\sigma }
\newcommand{\idtGo}{\id_{\tGo}}
\newcommand{\mes}{\mu \mstrut _{\sH}}
\newcommand{\nEm}{\nabla^{\sEm}}
\newcommand{\nL}{L^\sigma  }
\newcommand{\nn}{\nu _0}
\newcommand{\oMb}{\mathfrak{o}(\sM,b)}
\newcommand{\oV}{\mathfrak{o}(V)}
\newcommand{\oo}{1}
\newcommand{\pmo}{\{\pm1\}}
\newcommand{\ptc}{(\phi ,\tau ,c)}
\newcommand{\ptzgces}{$\phi $-twisted $\zt$-graded extensions}
\newcommand{\ptzgce}{$\phi $-twisted $\zt$-graded extension}

\newcommand{\rep}[1]{\rho \mstrut _{#1}}
\newcommand{\sA}{\mathcal{A}}

\newcommand{\sC}{\mathscr{C}}

\newcommand{\sEm}{\mathcal{E}^-}
\newcommand{\sEp}{\mathcal{E}^+}
\newcommand{\sE}{\mathcal{E}}
\newcommand{\sF}{\mathscr{F}}
\newcommand{\sH}{\mathscr{H}}
\newcommand{\sK}{\mathcal{K}}
\newcommand{\sL}{\mathcal{L}}
\newcommand{\sM}{\mathscr{M}}
\newcommand{\sO}{\mathcal{O}}
\newcommand{\sS}{\mathscr{S}}

\newcommand{\sV}{s\!\Vect}
\newcommand{\sims}{\genfrac{}{}{0pt}{}{{}_{{}_{\textnormal{\large$\sim$}}}}{{}^{\textnormal{\tiny stable}}}}
\newcommand{\sqo}{\sqrt{-1}}
\newcommand{\tA}{A^\tau }
\newcommand{\tC}{{C}}
\newcommand{\tD}{\widetilde{D}}
\newcommand{\tGo}{G_1^{\tau_1}}
\newcommand{\tGt}{(G')^\tau }

\newcommand{\tG}{G^\tau }
\newcommand{\tKO}{\widetilde{KO}}

\newcommand{\tL}{L ^\tau }

\newcommand{\tP}{\Pi }

\newcommand{\tT}{{T}}
\newcommand{\tXL}{X_\Pi }
\newcommand{\tar}{\rho ^{\tau }}
\newcommand{\ta}{\tilde\alpha }

\newcommand{\tgo}{\go^{\tau_1}}
\newcommand{\tgp}{\widetilde{\gp}}
\newcommand{\tg}{\tilde{g}}
\newcommand{\tiL}{\widetilde{L}}
\newcommand{\tilG}{\widetilde{G}}
\newcommand{\tnL}{{L}^{\tilde{\sigma }}}
\newcommand{\tnu}{\tilde{\nu }}
\newcommand{\tord}{(\cir)^{\times d}}
\newcommand{\tp}{\tilde\phi }

\newcommand{\ts}{\tilde\sigma }

\newcommand{\ztd}{(\zt)^{\times d}}

\begin{document}

\abovedisplayskip18pt plus4.5pt minus9pt
\belowdisplayskip \abovedisplayskip
\abovedisplayshortskip0pt plus4.5pt
\belowdisplayshortskip10.5pt plus4.5pt minus6pt
\baselineskip=15 truept
\marginparwidth=55pt

\renewcommand{\labelenumi}{\textnormal{(\roman{enumi})}}



 \title[Twisted Equivariant Matter]{Twisted Equivariant Matter} 
 \author[D. S. Freed]{Daniel S.~Freed}
 \thanks{The work of D.S.F. is supported by the National Science Foundation
under grants DMS-0603964 and DMS-1207817.}
 \address{Department of Mathematics \\ University of Texas \\ 1 University
Station C1200\\ Austin, TX 78712-0257}
 \email{dafr@math.utexas.edu}

 \author[G. W. Moore]{Gregory W.~Moore}
 \thanks{The work of G.W.M. is supported by the DOE under grant
DE-FG02-96ER40959.  GM also gratefully acknowledges partial support from the
Institute for Advanced Study and the Ambrose Monell Foundation.}
 \address{NHETC and Department of Physics and Astronomy \\
Rutgers University \\ Piscataway, NJ 08855--0849}
 \email{gmoore@physics.rutgers.edu}
 \thanks{This material is also based upon work supported in part by the
National Science Foundation under Grant No. 1066293 and the hospitality of
the Aspen Center for Physics.  We also thank the Institute for Advanced Study
and the Simons Center for Geometry and Physics for providing support and
stimulating environments for discussions related to this paper.}

 \date{January 7, 2013}
 \begin{abstract}  
 We show how general principles of symmetry in quantum mechanics lead to
twisted notions of a group representation. This framework generalizes both
the classical 3-fold way of real/complex/quaternionic representations as well
as a corresponding 10-fold way which has appeared in condensed matter and
nuclear physics. We establish a foundation for discussing continuous families
of quantum systems. Having done so, topological phases of quantum systems can
be defined as deformation classes of continuous families of gapped
Hamiltonians.  For free particles there is an additional algebraic structure
on the deformation classes leading naturally to notions of \emph{twisted}
equivariant $K$-theory. In systems with a lattice of translational symmetries
we show that there is a \emph{canonical} twisting of the equivariant
$K$-theory of the Brillouin torus. We give precise mathematical definitions
of two invariants of the topological phases which have played an important
role in the study of topological insulators.  Twisted equivariant $K$-theory
provides a finer classification of topological insulators than has been
previously available.
  \end{abstract}

\maketitle

Increasingly sophisticated ideas from homotopy theory are being used to
elucidate issues in quantum field theory and string theory.  While relatively
elementary topological notions were applied long ago to nonrelativistic
condensed matter systems~\cite{M,TKNN}, it is only in the past few years that
more modern topological invariants have been appearing.  They are used to
classify phases of noninteracting systems of electrons~\cite{HM, H, K, KBMHL,
SRFL1, SRFL2, SRFL3, MB, FK, FHNQWW, HK, TeK, KM1, KM2, FKM, Ha, EMV, ETMV,
R1, R2, R3, QZ, FKi1, FKi2, SCR, TZM, Wen, LWCG}.  This paper is particularly
motivated and inspired by the beautiful results of A. Kitaev and of A. Ludwig
et.\ al.~\cite{K,SRFL1,SRFL2,SRFL3,FKi1}.  We show that fundamental
principles of quantum mechanics---with a Wignerian emphasis on
symmetry---suggest definitions of \emph{\Qsymmetry classes} and
\emph{topological phases} of general quantum mechanical systems which lead to
\emph{twistings of $K$-theory} and \emph{twisted $K$-theory}.  Some of our
arguments are quite general and apply in particular to interacting systems.
 
Our starting point is Wigner's theorem concerning symmetries of a quantum
mechanical system.  The pure states form the projective space~$\PH$ of a
complex Hilbert space~$\sH$, and quantum symmetries are the \emph{projective}
transformations which preserve transition probabilities.  Wigner's theorem
asserts that a quantum symmetry has a real \emph{linear} lift which is either
unitary or antiunitary.  Therefore, any group~$G$ of quantum symmetries comes
equipped with a homomorphism $\phi \:G\to\pmo$ which encodes unitarity
vs. antiunitarity, together with a group extension
  \begin{equation}\label{eq:144}
     1\longrightarrow \TT\longrightarrow \tG\longrightarrow G\longrightarrow
     1 
  \end{equation}
whose kernel~$\TT=\{e^{i\theta }:\theta \in \RR\}$ is the group of scalar
unitary transformations of~$\sH$, which act trivially on the space~$\PH$ of
pure quantum states.  We say \eqref{eq:144}~ is a \emph{$\phi $-twisted
 extension} (by~$\TT$) since complex scalars commute with unitary
transformations but complex conjugate when commuted past antiunitary
transformations.  The action of~$\tG$ on~$\sH$ is a ``twisted'' version of a
complex linear representation of~$G$; the twisting is encoded by the abstract
\emph{\Qsymmetry class}~$(G,\phi ,\tau )$.  If we artificially ask that
\eqref{eq:144}~be pulled back via~$\phi $, then there is a trichotomy of
\Qsymmetry classes for fixed~$G$, which roughly corresponds to the
Frobenius-Schur trichotomy of complex/real/quaternionic representations;
under a further splitting hypothesis $\sH$~has a real or quaternionic
structure if $\phi $~is onto.

Typically an abstract group~$G$ of quantum symmetries also acts on spacetime.
In particular, there is a homomorphism $t\:G\to\pmo$ which encodes whether a
symmetry preserves or reverses time-orientation.  If in addition we fix a
one-parameter group of time translations, so a one-parameter group of unitary
operators generated by a self-adjoint Hamiltonian~$H$, then a simple argument
(Lemma~\ref{thm:90}) shows that the operator corresponding to a
symmetry~$g\in G$ commutes or anticommutes with~$H$ according to the
product~$\phi (g)t(g)$.  In particular, if $H$~is bounded below and unbounded
above, then $\phi =t$.  That restriction on~$H$ does not necessarily hold in
finite quantum systems~($\dim\sH<\infty $), nor in noninteracting single
electron systems in the Dirac-Nambu picture~(\S\ref{sec:17}).  Therefore, we
study \emph{extended \Qsymmetry classes}~$\Gptc$ in which $c:=\phi
t\:G\to\pmo$~is allowed to be nontrivial.  If we now artificially restrict
the $\phi $-twisted extension to be a pullback via $(t,c)\:G\to\pmo\times
\pmo$, we recover a decachotomy ubiquitous in the condensed matter
literature; it reduces to the previous trichotomy for~$c\equiv 1$.  Under a
further splitting hypothesis the Hilbert space~$\sH$ is a module for one of
the ten Morita classes of real and complex Clifford algebras.  The refinement
of a \Qsymmetry class~$(G,\phi ,\tau )$ to an extended \Qsymmetry
class~$\Gptc$ only has meaning if we fix an evolution of the system.
 
In~\S\ref{sec:17} we discuss free fermion systems with finitely many degrees
of freedom.  Here we recover a result of
Altland-Heinzner-Huckleberry-Zirnbauer~\cite{AZ,HHZ} relating spaces of free
fermion Hamiltonians compatible with a given symmetry and classical symmetric
spaces of compact type.  We also define the Dirac-Nambu Hilbert space, which
motivates some of our later definitions.

We investigate deformation classes of quantum mechanical systems with fixed
extended \Qsymmetry type~$\Gptc$.  The first task is foundational: define a
continuous family of quantum mechanical systems.  Our general discussion of
this point in Appendix~\ref{sec:14} may have broader interest.  A system is
said to be {\it gapped\/} if the Hamiltonian is invertible, and two gapped
systems are said to be in the same \emph{topological phase} if there is a
continuous path leading from one to the other.  The set~$\TP\Gptc$ of
topological phases has an algebraic structure given by amalgamation of
quantum systems.  In standard quantum mechanics the Hilbert space of the
amalgam of quantum systems with Hilbert spaces~$\sH_1,\sH_2$ is the
\emph{tensor product}~$\sH_1\otimes \sH_2$.  However, in the Dirac-Nambu
picture Hilbert spaces combine using \emph{direct sum}~$\sH_1\oplus \sH_2$
instead.  This nonstandard amalgamation is the only specific feature we
abstract from noninteracting single fermion systems.  With this algebraic
structure $\TP\Gptc$~is a commutative monoid,\footnote{A \emph{monoid} is a
set with an associative composition law and a unit for composition.  In the
commutative case we write the composition as~`$+$' and write `$0$' for the
unit.  The commutative monoid is an \emph{abelian group} if additive inverses
exist.}  and it is natural to pass to an abelian group~$\GS\Gptc$ of
\emph{reduced topological phases}.  Two systems in the same topological phase
necessarily have the same reduced topological phase, but not \emph{vice
versa}, so $\GS$~is a cruder invariant than~$\TP$.  On the other hand, it is
more easily computable and is strong enough to distinguish interestingly
different systems.  We illustrate this by proving in~\S\ref{sec:16} that two
physically important topological invariants of certain band insulators---the
orbital magnetoelectric polarizability and the Kane-Mele invariant---factor
through~$\GS$, which leads to a proof that they are equal.

The passage from the commutative monoid~$\TP$ to the group~$\GS$ is
reminiscent of $K$-theory, and indeed our main results identify~$\GS$ with a
topological $K$-theory group.  For that we need several additional
hypotheses.  First, we assume that there is a gap in the spectrum of the
Hamiltonian~$H$, which we shift to be at zero energy.  Then after a homotopy
we may assume~$H^2=1$, i.e., that $H$~is the grading operator of a
$\zt$-grading on~$\sH$.  We also make certain finiteness and compactness
hypotheses on the Hilbert space, Hamiltonian, and symmetry group.  These have
discrete consequences---for example, representations of compact Lie groups
are rigid whereas representations of noncompact Lie groups often have
continuous deformations---which are part of the proofs that we obtain
$K$-theory groups.  As a warmup in~\S\ref{sec:11} we compute~$\GS$ for gapped
systems with finite dimensional Hilbert space.  While the proofs here are
mainly a matter of unwinding definitions, this case illustrates important
issues which arise in more intricate higher dimensional systems.
 
After these extensive preliminaries we turn to quantum mechanical systems
symmetric under a spacetime crystallographic group~$G$ with nontrivial
magnetic point group~$\hP$.  Thus $G$~is a group extension
  \begin{equation}\label{eq:145}
     1\longrightarrow \Pi \longrightarrow G\longrightarrow \hP\longrightarrow
     1 
  \end{equation}
with kernel a full lattice~$\Pi $ of translations.  The operators
corresponding to elements of~ $\Pi $ commute, so can be simultaneously
diagonalized, and this Fourier transform (using \emph{Bloch sums}) realizes
the Hilbert space as $L^2$~sections of an infinite rank vector bundle
$\sE\to\XL$ over the (Brillouin) torus~$\XL$ of characters of~$\Pi $.  The
group~$\hP$ acts on~$\XL$, but if \eqref{eq:145}~is not split that action
does not generally lift to~$\sE$.  Rather, the extension~\eqref{eq:145}
determines a central extension of the \emph{groupoid} $\XL\gpd \hP$, and it
is this central extension which lifts.  This central extension is known in a
different language in the condensed matter literature: when describing the
action of group elements of the point group on Bloch eigenfunctions there are
some phase ambiguities which are generally dealt with in an ad hoc fashion
(see, e.g.,~\cite{BP}).  These phases are encoded in the central extension of
the groupoid.  In general, a central extension of a groupoid is a special
kind of \emph{twisting} of its $K$-theory.  This canonical twisting
determined by a group extension was discovered in~\cite{FHT1,FHT2}.  More
generally, we consider an abstract symmetry group equipped with a
homomorphism to~\eqref{eq:145}, for example in systems with internal
symmetry, and it may come equipped with a nontrivial extended \Qsymmetry type
unconnected with the extension~\eqref{eq:145}.  In that case the \Qsymmetry
type and canonical twisting from~\eqref{eq:145} combine, and $\sE\to\XL$ is a
twisted bundle for the combined twisting.  The Fourier transform based
on~$\Pi $ brings us from twisted representations---a direct consequence of
Wigner's theorem and spacetime symmetry---to twisted equivariant vector
bundles that represent twisted $K$-theory classes:
 \bigskip
  \begin{equation*}
     \begin{aligned} \textnormal{group}&\longrightarrow \textnormal{groupoid}
      \\ \textnormal{extended \Qsymmetry class}&\longrightarrow
     \textnormal{twisting of 
      $K$-theory\phantom{MMMMMM}} \\ \textnormal{twisted
     representation}&\longrightarrow 
      \textnormal{twisted vector bundle} \\ \textnormal{classes of twisted
     virtual representations}&\longrightarrow \textnormal{twisted $K$-theory
     group} \\  
      \end{aligned} 
  \end{equation*}
The hypothesis of a gap in the Hamiltonian, specified by a Fermi level,
induces a decomposition $\sE=\sEp\oplus \sEm$.  Due to Fermi-Dirac statistics
these have an interpretation as unfilled bands and filled bands, also known
as conduction bands and valence bands, respectively.  We assume that
$\sEm$~has finite rank---there are only finitely many valence bands---and
consider separately the finite and infinite rank possibilities for~ $\sEp$.
Our main results (Theorem~\ref{thm:131} and Theorem~\ref{thm:136}) identify
the group of reduced topological phases with a twisted $K$-theory group.  The
nonequivariant specialization, which ignores the magnetic point group~$\hP$,
is well-known and expressed in terms of untwisted $K$-theory~\cite{K}; the
result for the full symmetry group is new.  (A special case is addressed
in~\cite{TZM}.)  In case $\sEp$~has infinite rank we pull off an Eilenberg
swindle to prove it does not carry any topological information.
 
Our results about topological phases are expressed in terms of monoids and
groups of equivalence classes.  These are sets of path components of
topological spaces---or homotopy types---which classify families of quantum
systems.  The arguments we give can be adapted to identify classifying spaces
of gapped quantum systems in terms of classifying spaces for $K$-theory.

A brief outline of the paper is as follows.  We begin in~\S\ref{sec:2} with
an exposition of Wigner's theorem and \Qsymmetry classes.  In~\S\ref{sec:1}
we review symmetries of an affine spacetime with Galilean structure.  We take
a pre-Cartesian non-formulaic approach to affine, Euclidean, and projective
geometry in these introductory sections; it brings out the symmetries more
clearly and is designed to convey the simplicity of the geometric ideas.
In~\S\ref{sec:10} we define extended \Qsymmetry classes and topological
phases.  Free fermions are the subject of~\S\ref{sec:17}, where we consider
them in the context of algebraic quantum mechanics.  We give a quick proof of
the main theorem in~\cite{HHZ}.  The definitions of topological phases and
reduced topological phases are the subject of~\S\ref{sec:18}.  Classification
results, including the decachotomy of special extended \Qsymmetry classes,
are proved in~\S\ref{sec:4}.  Section~\ref{sec:3} is a rapid introduction to
special twistings of $K$-theory and twisted $K$-theory, beginning with
twisted versions of complex representations of a group.  There we also
introduce groupoids, which simultaneously generalize groups and spaces.  At
this point we have enough to discuss the classification of systems with
finite dimensional state spaces in~\S\ref{sec:11}.  The canonical twisting of
equivariant $K$-theory from a group extension is worked out in~\S\ref{sec:5}.
The main theorems computing the group of phases of band insulators are stated
and proved in~\S\ref{sec:6}.  Section~\ref{sec:16} is a treatment of the
Kane-Mele invariant and the orbital magnetoelectric polarizability, which is
a Chern-Simons invariant.  We report on some computations.  There are several
appendices of background and supplementary material.  Appendix~\ref{sec:13}
is a rapid compendium of definitions related to group extensions.  As is
well-known, the ten special extended \Qsymmetry classes correspond to the ten
Morita classes of real and complex Clifford algebras.  We prove a precise
result in Appendix~\ref{sec:12}, which we use at the end of~\S\ref{sec:11}
and~\S\ref{sec:6} to recast special symmetries as a degree shift in
$K$-theory.  A beautiful paper of Dyson~\cite{Dy} also contains a 10-fold
way; we give a modern treatment in Appendix~\ref{sec:9}.  In
Appendix~\ref{sec:14} we give a careful definition of a \emph{continuous}
family of quantum mechanical systems with fixed \Qsymmetry type.  We also
give a geometric treatment of Bloch sums and the Berry connection.
Appendix~\ref{sec:19} contains a lemma which ensures that the definition of
twisted equivariant $K$-theory for the special twistings we encounter agrees
with the general definition.  An extended example, the 3-dimensional diamond
structure, is treated in Appendix~\ref{sec:15}.

While we do not discuss topological superconductors in detail, our
considerations apply to free fermions with Bogoliubov-de Gennes Hamiltonians.
We have nothing to say in this paper about the topology of electron systems
with no gap in the spectrum.  We stress that our discussion of~$\TP$ as a
commutative monoid and hence $\GS$~as an abelian group is limited to
noninteracting systems.  The generalization to interacting systems is an
interesting open problem; see~\cite{FKi1,FKi2} for interesting recent
results.
 
We thank Mike Freedman for explaining basics of topological insulators to us
a few years ago.  We are grateful to Nick Read for sharing an unpublished
manuscript on the subject.  We also thank Jacques Distler, Charlie Kane,
Alexei Kitaev, Joel Moore, Chetan Nayak, Karin Rabe, Andrew Rappe, Aaron
Royer, Lorenzo Sadun, Graeme Segal, David Vanderbilt, Edward Witten, and
Martin Zirnbauer for helpful discussions and correspondence.  We thank the
anonymous referees for their careful reading and many helpful comments.

{\small
\renewcommand{\tocsection}[3]{%
  \indentlabel{\ignorespaces#1 #2.\;\;}#3}
\renewcommand{\tocsubsection}[3]{%
  \indentlabel{\hskip 4em#3}}
\bigskip\bigskip
\begin{center}{\scshape Contents}\end{center}\bigskip
 
\contentsline {section}{\tocsection {}{1}{Quantum symmetries and twisted extensions}}{8}{section.1}
\contentsline {subsection}{\tocsubsection {}{}{Wigner's theorem}}{8}{section*.2}
\contentsline {subsection}{\tocsubsection {}{}{QM\nobreakspace symmetry\ classes and twisted representations}}{9}{section*.3}
\contentsline {subsection}{\tocsubsection {}{}{Symmetries and spacetime}}{11}{section*.4}
\contentsline {section}{\tocsection {}{2}{Symmetries of Galilean spacetime}}{12}{section.2}
\contentsline {subsection}{\tocsubsection {}{}{Affine spaces: global parallelism}}{12}{section*.5}
\contentsline {subsection}{\tocsubsection {}{}{Euclidean and Galilean structures}}{13}{section*.6}
\contentsline {subsection}{\tocsubsection {}{}{Tying abstract symmetry groups to spacetime}}{15}{section*.7}
\contentsline {subsection}{\tocsubsection {}{}{Crystals and crystallographic groups}}{16}{section*.8}
\contentsline {section}{\tocsection {}{3}{Extended QM\nobreakspace symmetry\ classes}}{18}{section.3}
\contentsline {subsection}{\tocsubsection {}{}{Quantum symmetries and time-reversal}}{18}{section*.9}
\contentsline {subsection}{\tocsubsection {}{}{Extended QM\nobreakspace symmetry\ classes and gapped systems}}{19}{section*.10}
\contentsline {section}{\tocsection {}{4}{Free fermions}}{21}{section.4}
\contentsline {subsection}{\tocsubsection {}{}{Symmetries in algebraic quantum mechanics}}{21}{section*.11}
\contentsline {subsection}{\tocsubsection {}{}{Free fermions}}{22}{section*.12}
\contentsline {subsection}{\tocsubsection {}{}{Free fermion Hamiltonians and classical symmetric spaces}}{23}{section*.13}
\contentsline {subsection}{\tocsubsection {}{}{The Dirac-Nambu space}}{25}{section*.14}
\contentsline {section}{\tocsection {}{5}{Topological phases}}{26}{section.5}
\contentsline {subsection}{\tocsubsection {}{}{Group completion}}{27}{section*.15}
\contentsline {subsection}{\tocsubsection {}{}{Quotienting by topological triviality}}{28}{section*.16}
\contentsline {subsection}{\tocsubsection {}{}{Reduced topological phases}}{29}{section*.17}
\contentsline {section}{\tocsection {}{6}{Special extended QM\nobreakspace symmetry\ classes}}{29}{section.6}
\contentsline {subsection}{\tocsubsection {}{}{A restricted set of extended QM\nobreakspace symmetry\ classes}}{29}{section*.18}
\contentsline {subsection}{\tocsubsection {}{}{Twisted central extensions and semidirect products}}{31}{section*.19}
\contentsline {subsection}{\tocsubsection {}{}{Proof of the classification}}{33}{section*.20}
\contentsline {subsection}{\tocsubsection {}{}{More QM\nobreakspace symmetry\ classes}}{33}{section*.21}
\contentsline {section}{\tocsection {}{7}{$K$-theory and some twistings}}{36}{section.7}
\contentsline {subsection}{\tocsubsection {}{}{Twisted representation rings}}{36}{section*.22}
\contentsline {subsection}{\tocsubsection {}{}{Extensions as line bundles}}{37}{section*.23}
\contentsline {subsection}{\tocsubsection {}{}{Groupoids, twistings, and twisted vector bundles}}{39}{section*.24}
\contentsline {subsection}{\tocsubsection {}{}{Twisted $K$-theory}}{42}{section*.25}
\contentsline {section}{\tocsection {}{8}{Gapped systems with a finite dimensional state space }}{43}{section.8}
\contentsline {section}{\tocsection {}{9}{Twistings from group extensions}}{45}{section.9}
\contentsline {subsection}{\tocsubsection {}{}{Warmup: direct products}}{45}{section*.26}
\contentsline {subsection}{\tocsubsection {}{}{Group extensions and twistings}}{46}{section*.27}
\contentsline {subsection}{\tocsubsection {}{}{A nontrivial example}}{50}{section*.28}
\contentsline {subsection}{\tocsubsection {}{}{A generalization}}{52}{section*.29}
\contentsline {section}{\tocsection {}{10}{Gapped topological insulators}}{53}{section.10}
\contentsline {subsection}{\tocsubsection {}{}{Periodic systems of electrons}}{53}{section*.30}
\contentsline {subsection}{\tocsubsection {}{}{Formal setup}}{55}{section*.31}
\contentsline {subsection}{\tocsubsection {}{}{Topological phases: Type\nobreakspace F}}{57}{section*.32}
\contentsline {subsection}{\tocsubsection {}{}{Topological phases: Type\nobreakspace I}}{58}{section*.33}
\contentsline {subsection}{\tocsubsection {}{}{Simplifying assumptions; more familiar $K$-theory groups}}{60}{section*.34}
\contentsline {section}{\tocsection {}{11}{Chern-Simons and Kane-Mele invariants}}{63}{section.11}
\contentsline {subsection}{\tocsubsection {}{}{Computations}}{64}{section*.35}
\contentsline {subsection}{\tocsubsection {}{}{Orbital magnetoelectric polarizability}}{66}{section*.36}
\contentsline {subsection}{\tocsubsection {}{}{The Kane-Mele invariant}}{68}{section*.37}
\contentsline {section}{\tocsection {Appendix}{A}{Group extensions}}{72}{appendix.A}
\contentsline {section}{\tocsection {Appendix}{B}{Clifford algebras and the 10-fold way}}{74}{appendix.B}
\contentsline {section}{\tocsection {Appendix}{C}{Dyson's 10-fold way}}{76}{appendix.C}
\contentsline {section}{\tocsection {Appendix}{D}{Continuous families of quantum systems}}{80}{appendix.D}
\contentsline {subsection}{\tocsubsection {}{}{Topologies on mapping spaces}}{80}{section*.38}
\contentsline {subsection}{\tocsubsection {}{}{Hilbert bundles and continuous families of quantum systems}}{81}{section*.39}
\contentsline {subsection}{\tocsubsection {}{}{Fourier transform (Bloch sums) and the Berry connection}}{84}{section*.40}
\contentsline {section}{\tocsection {Appendix}{E}{Twisted $K$-theory on orbifolds}}{86}{appendix.E}
\contentsline {section}{\tocsection {Appendix}{F}{Diamonds and dust}}{88}{appendix.F}
\contentsline {section}{References}{91}{equation.F.8}
}

   \section{Quantum symmetries and twisted extensions}\label{sec:2}

  \subsection*{Wigner's theorem}

The state space of a quantum system is the projective space~$\PH$ of a
complex separable Hilbert space~$\sH$.  In other words, the state of a
quantum system is a line of vectors.  (Sometimes the term `ray' is used in
place of `line', and one may require that the state be normalized to have
unit norm.)  If $\ell ,\ell '\in \PH$ are states, then the \emph{transition
probability} from~$\ell $ to $\ell '$ is defined as $|\langle \psi ,\psi '
\rangle|^2$, where $\psi ,\psi '\in \sH $ are unit norm vectors contained in
the lines~$\ell ,\ell '$, respectively, and $\langle -,- \rangle$~is the
inner product on~$\sH$.  The transition probability is a symmetric function
  \begin{equation}\label{eq:57}
     p\:\PH\times\PH\longrightarrow [0,1]. 
  \end{equation}
A \emph{projective quantum symmetry} of~$\PH$ is an invertible map
$\PH\to\PH$ which preserves the function~$p$.  A basic theorem of
Wigner~\cite{Wi}, \cite[\S2.A]{We} asserts that any projective quantum
symmetry lifts to either a unitary or antiunitary\footnote{A real linear map
$S\:\sH\to\sH$ is \emph{antiunitary} if it is antilinear ($S\lambda
=\bar\lambda S,\;\lambda \in \CC$) and if $\langle S\xi _1,S\xi _2
\rangle=\overline{\langle \xi _1,\xi _2 \rangle}$ for all~$\xi _1,\xi _2\in
\sH$.}  transformation $\sH\to\sH$, which we call a \emph{linear quantum
symmetry}.  Recall that both unitary and antiunitary transformations preserve
lengths and angles, but a unitary map is complex linear whereas an
antiunitary map is complex antilinear.

  \begin{remark}[]\label{thm:2}
 We have defined projective quantum symmetries in terms of the transition
probability structure on projective space, the symmetric
function~\eqref{eq:57}.  The projective space~$\PH$ has a more conventional
symmetric function which is induced from the inner product on~$\sH$: the
distance function~$d$ of the Fubini-Study metric.  It is not difficult to
compute that $p=\cos^2(d/2)$, from which the group of projective quantum
symmetries is precisely the isometry group of projective space.   Proofs of
Wigner's theorem based on this identification are given in~\cite{F1}.
  \end{remark}

We recast Wigner's theorem in the following terms.  Let $\QAut(\PH)$~denote
the group of projective quantum symmetries of~$\PH$ and $\QAut(\sH)$~the
group of linear quantum symmetries, that is, the group of all unitary and
antiunitary transformations of~$\sH$.  It is a group since the composition of
two antiunitary transformations is unitary.  In fact, $\QAut(\sH)$~ fits into
a group extension (see Definition~\ref{thm:54})
  \begin{equation}\label{eq:12}
      1\longrightarrow U(\sH)\longrightarrow
     \QAut(\sH)\xrightarrow{\;\;\phi_{\sH} \;\;}\pmo \longrightarrow 1 
  \end{equation}
with kernel the group~$U(\sH)$ of unitary operators: for $S\in \QAut(\sH)$ we
have~$\phi_{\sH} (S)=1$ if $S$~is unitary and $\phi_{\sH} (S)=-1$ if $S$~is
antiunitary.  Wigner's theorem asserts that there is a group extension
  \begin{equation}\label{eq:13}
     1\longrightarrow \TT\longrightarrow \QAut(\sH)\longrightarrow
     \QAut(\PH)\longrightarrow 1 
  \end{equation}
in which the kernel is the group~$\TT$ of scalings of~$\sH$ by complex
numbers of norm one.\footnote{We assume $\dim\sH>1$.}
In other words, every projective quantum symmetry lifts to a linear quantum
symmetry which is unique up to composition with scalar multiplication by a
unit norm complex number.  Note that $\TT$~is not central---for $\lambda \in
\TT$ and $S\in \QAut(\sH)$ we have
  \begin{equation}\label{eq:14}
     S\lambda = \begin{cases}  \lambda S,&\phi_{\sH} (S)=1;\\\bar\lambda
     S,&\phi_{\sH} (S)=-1.\end{cases} 
  \end{equation}
Hence~\eqref{eq:13} is not a central extension, but rather a slight twisting
of a central extension.  The twisting is expressed by the
homomorphism~$\phi_{\sH} $ in~\eqref{eq:12}, which descends to a homomorphism
(with the same name)
  \begin{equation}\label{eq:15}
     \phi_{\sH} \:\QAut(\PH)\longrightarrow \pmo
  \end{equation}
that encodes whether a quantum symmetry lifts to be unitary or antiunitary.

In the uniform (norm) topology $U(\sH)$~is connected---in fact, contractible
if $\sH$~is infinite dimensional---and $\QAut(\sH)$~has two components.  It
is more natural for us to topologize bounded operators with the compact-open
topology, as we discuss in Appendix~\ref{sec:14}.  In that
topology\footnote{We topologize these groups as subspaces of invertible
operators, but invertible operators have a topology finer than the subspace
topology inherited from the strong topology; see Definition~\ref{thm:107}.}
it is also true that $U(\sH)$~is contractible and $\QAut(\sH)$~has two
components, as we prove in Proposition~\ref{thm:108}.

  \subsection*{\Qsymmetry classes and twisted representations}

We abstract the structure which emerges from lifting projective quantum
symmetries to linear quantum symmetries. 

  \begin{definition}[]\label{thm:5}
 Let $G$~be a topological group and $\phi\:G\to\pmo$ a continuous
homomorphism.  A \emph{$\phi$-twisted extension} (by~$\TT$) of~$G$ is a group
extension
  \begin{equation}\label{eq:20}
     1\longrightarrow \TT\longrightarrow \tG\longrightarrow G\longrightarrow
     1 
  \end{equation}
which for all~$S\in \tG$ and~$\lambda \in \TT$ satisfies ~\eqref{eq:14} with
$\phi _{\sH}$ replaced by~$\phi $.  We denote the extension as~`$\tG\to G$',
or simply~`$\tau $'.
  \end{definition}

\noindent 
 If $\phi(g)=1$ for all~$g\in G$, then \eqref{eq:20}~is a central
extension, i.e., the kernel~$\TT$ commutes with every element of~$\tG$.  For
any~$\phi $ there is a trivial $\phi $-twisted extension, the semidirect
product $G\ltimes_\phi \TT$ (see Definition~\ref{thm:87}), where $G$~acts
on~$\TT$ by complex conjugation via~$\phi $.  We will meet nontrivial twisted
extensions below.

Now suppose $G$~is any topological group of quantum symmetries of a quantum
system whose state space is~$\PH$.  In other words, $G$~acts on~$\PH$ through
a continuous homomorphism
  \begin{equation}\label{eq:16}
     \rho \:G\longrightarrow \QAut(\PH).
  \end{equation}
Set $\phi=\phi _{\sH}\circ \rho \:G\to\pmo$ to be the composition of~$\rho $
with~\eqref{eq:15}.  There is a pullback $\phi$-twisted  extension 
of~$G$ (see Definition~\ref{thm:89}) which fits into the diagram
  \begin{equation}\label{eq:17}
     \xymatrix{1 \ar[r]& \TT \ar[r] \ar@{=}[d] & \tG \ar[d]^{\rho ^\tau
     }\ar[r] & G 
     \ar[r]\ar[d]^\rho & 1\\ 1 \ar[r] &\TT \ar[r] &\QAut(\sH) \ar[r]
     &\QAut(\PH)\ar[r] &1} 
  \end{equation}
Simply put, the group~$\tG$ consists of all linear quantum symmetries which
lift the projective quantum symmetries in~$G$.  Note that $\phi$~encodes
which operators in~$\tG $ act linearly and which act antilinearly.  Also,
$\TT\subset \tG$ acts on~$\sH$ by scalar multiplication.  We formalize these
properties.

  \begin{definition}[]\label{thm:15}
 Let $\tG$ be a $\phi $-twisted extension of~$G$, as in
Definition~\ref{thm:5}.  Then a homomorphism $\tar\:\tG\to\Aut_{\RR}(E)$ to
the group of real automorphisms of a complex vector space~$E$ is a
\emph{$(\phi ,\tau )$-twisted} representation of~$G$ if:
  \begin{enumerate}
 \item $\tar(g)$ is complex linear if $\phi (g)=+1$, $\tar(g)$ is complex
antilinear if $\phi (g)=-1$; and

 \item $\TT\subset \tG$ acts by complex scalar multiplication.
  \end{enumerate}
  \end{definition}

\noindent
 Note that if $G$---and hence also~$\tG$---is connected, then $G$~acts by
complex linear transformations.

  \begin{remark}[]\label{thm:33}
 A homomorphism~$\tar$ which satisfies~(i) is called a `co-representation' by
Wigner~\cite{Wi}.
  \end{remark}

Next we abstract the structure~\eqref{eq:17} induced on a group of projective
quantum symmetries.  The definition parallels the standard notions of
concrete group actions and abstract symmetry groups.

  \begin{definition}[]\label{thm:4}\ 
  \begin{enumerate}
 \item Let $\rho \:G\to\QAut(\PH)$ be a group of projective quantum
symmetries.  The \emph{\Qsymmetry type} of~$(G,\rho )$ is the isomorphism class
of the triple~$(G,\phi ,\tau)$ consisting of~$G$, the induced homomorphism
$\phi\:G\to\pmo$, and the induced $\phi$-twisted  extension~$\tG$
in~\eqref{eq:17}.
 
 \item A \emph{\Qsymmetry class}\footnote{Altland-Zirnbauer~\cite{AZ,HHZ} use
the term `symmetry classes' in a different way, referring to the
classification of Hamiltonians with given symmetry.  We make some comments on
that problem in~\S\ref{sec:17}.} is an isomorphism class of triples~$(G,\phi
,\tau)$ where $G$~is a topological group, $\phi \:G\to\pmo$ is a continuous
homomorphism, and $\tG$~is a $\phi $-twisted extension as in
Definition~\ref{thm:5}.
  \end{enumerate}
  \end{definition}

\noindent
 We often call a triple~$(G,\phi ,\tau )$ a \Qsymmetry class, though strictly
speaking the \Qsymmetry class is the equivalence class containing the triple.
Also, we say the \Qsymmetry class is \emph{based on~$G$}.  Recall from
Definition~\ref{thm:88} that two $\phi$-twisted  extensions~$G^{\tau
_1},G^{\tau _2}$ of~$G$ are isomorphic if there is an isomorphism of groups
$\varphi \:G^{\tau _1}\to G^{\tau _2}$ which fits into the commutative
diagram
  \begin{equation}\label{eq:97}
     \xymatrix@R-15pt{&&G^{\tau _1}\ar[dd]^{\varphi }\ar[dr]\\ 1 \ar[r]& \TT
     \ar[ur]\ar[dr] && G \ar[r] & 1\\ && G^{\tau _2}\ar[ur]} 
  \end{equation}
Two projective actions of~$G$ have the same \Qsymmetry type if both actions
induce the same unitary vs.~ antiunitary dichotomy on elements of~$G$ and if
the linear and antilinear lifts glue together into groups which are not only
abstractly isomorphic, but are isomorphic in a way that matches the
underlying projective symmetries.  Symmetry \emph{type} is an invariant of a
group of projective quantum symmetries.  A \Qsymmetry \emph{class} is an
abstract group of possible quantum symmetries, up to isomorphism.  The
classification of \Qsymmetry classes for fixed~$G$ can be expressed in terms
of group cohomology, but for topological and Lie groups one must use
continuous or smooth cohomology rather than the cohomology of discrete
groups.

There are special \Qsymmetry classes we term \emph{$\phi $-standard}.    

  \begin{definition}[]\label{thm:167}
  Let $G$~be a Lie group equipped with a homomorphism $\phi \:G\to\pmo$.  A
\Qsymmetry class~$(G,\phi ,\tau)$ is \emph{$\phi $-standard} if $\tG$~is
pulled back from a twisted extension of $\phi (G)\subset \pmo$.
  \end{definition}

\noindent
 Clearly a $\phi $-standard \Qsymmetry class is determined by a subgroup
$A\subset \pmo$ and a \Qsymmetry class based on~$A$ with $\phi
\:A\hookrightarrow \pmo$ the inclusion map.  There are 2~possible subgroups
and 3~possibilities in all.

  \begin{proposition}[3-fold way]\label{thm:34}
 \ 
  \begin{enumerate}
 \item There is a unique $\phi $-standard \Qsymmetry class based on the trivial
subgroup~$A=1$.  

 \item There are two $\phi $-standard \Qsymmetry classes based on $A=\pmo$.
We denote them generically by $\tA$.  Let $\tT\in \tA$ be a lift of~$-1\in
\pmo$; then the two \Qsymmetry classes are distinguished by~$\tT^2=\pm1$.
  \end{enumerate}
  \end{proposition}

  \begin{remark}[]\label{thm:78}
 In~(ii) the two groups~$\tA$ may be identified with the two
forms~$\Pin_{\pm2}$ of the pin group. 
  \end{remark}

  \begin{remark}[]\label{thm:140}
 If $\phi $~is surjective there may not be an element $\hT\in \tG$ mapping
to~$\bT$ with $\hT^2=\pm1$.  A stronger assumption than $\phi $-standard is
that $\tG\cong \tA\times G_0$, where $G_0$~is the kernel of~$\phi $.  If so,
the operator~$\hT=(T,1)$ in~(ii), where $T$~is a real or quaternionic
structure on~$\sH$ preserved by~$G_0$.
  \end{remark}

  \begin{remark}[]\label{thm:168}
 We prove the more general Proposition~\ref{thm:24} below so do not write out
a proof of the simple Proposition~\ref{thm:34} here.  But we do observe that
the key steps are contained in parts~(i) and~(ii) of Lemma~\ref{thm:11}. 
  \end{remark}

  \subsection*{Symmetries and spacetime}

In this section we have discussed abstract quantum systems and their
symmetries.  Typically the symmetry group of a quantum system has a
homomorphism to symmetries of spacetime, and so in the next section we review
those symmetries in the Galilean setting.  In~\S\ref{sec:10} we resume the
discussion of quantum symmetries and extend the notion of a \Qsymmetry class to
account for one particular aspect of the spacetime symmetry: preservation or
reversal of the arrow of time.

A nonrelativistic quantum system carries an action of (a subgroup of) the
Galilean group, which as we review in~\S\ref{sec:1} is a double cover of the
connected component of symmetries of a Galilean spacetime.  By~\eqref{eq:17}
this determines a central extension of the Galilean group.  There is in fact
a nontrivial central extension, and it encodes the total mass of the
system.\footnote{The mass central extension of the Galilean group can be
constructed as a limit of the Poincar\'e group in which the speed of light
tends to infinity; see~\cite[\S1.4]{DF1} for example.  We describe the Lie
algebra of the mass central extension in the proof of Lemma~\ref{thm:12}.} By
contrast a relativistic quantum system carries an action of the Poincar\'e
group---a double cover of the connected component of symmetries of a
Minkowski spacetime---and the induced central extension~\eqref{eq:17} is
split (see Definition~\ref{thm:86}), as is every central extension of the
Poincar\'e group.  In other words, for $G$~the Poincar\'e group any
homomorphism $\rho \:G\to\QAut(\PH)$ has a lift $\tilde\rho
\:G\to\QAut(\sH)$:
  \begin{equation}\label{eq:18}
     \xymatrix{&&& G\ar[d]^\rho \ar[dl]_{\tilde\rho }\\     
1\ar[r] &\TT\ar[r] &\QAut(\sH)\ar[r]& \QAut(\PH)\ar[r] & 1 } 
  \end{equation}
Nonrelativistic and relativistic quantum systems often have larger symmetry
groups~$G$ equipped with a homomorphism to Galilean or Poincar\'e group.  For
example, $G$~ may include parity-reversal and/or time-reversal symmetries.
There may also be other symmetries which commute with the Galilean or
Poincar\'e symmetries.  These are known as global symmetries, or sometimes as
internal symmetries.

   \section{Symmetries of Galilean spacetime}\label{sec:1}

  \subsection*{Affine spaces: global parallelism}

Just as Felix Klein's \emph{Erlangen program} defines geometry as the study
of properties invariant under a given symmetry group, so too models of
physical systems are often built around symmetry principles tied to physical
spacetime.  In this spirit nonrelativistic physics is the study of systems
invariant under the Galilean group.  It is useful to \emph{define} the
Galilean group as the group of symmetries of a \emph{Galilean
structure}~$\Gamma $ on spacetime, and one task in this subsection is to
define that structure precisely.  (In fact, we reserve the term `Galilean
group' for a double cover of the identity component of the group of
symmetries of~$\Gamma $.\footnote{One could imagine using the term without
the spin double cover, or to include other components, but in fact our usage
conforms to that adapted for the term `Poincar\'e
group'~\cite[Glossary]{Detal} }) The symmetry group in a condensed matter
system is often the subgroup that preserves a further \emph{crystal
structure}, as we explain.  We linger a bit over the elementary material here
to introduce ideas and definitions that we will use freely in more advanced
contexts later.
 
Space~$E$ and spacetime~$M$ are real finite dimensional \emph{affine} spaces.
An affine space carries a simply transitive action of a real \emph{vector}
space~$V$ of translations: given two points~$p_0,p_1$ in affine space there
is a unique vector~$\xi $ such that $p_1$~is the point reached by starting
at~$p_0$ and displacing by~$\xi $.  We write simply $p_1=p_0+\xi $.  The
action of~$V$ on~$E$ encodes the \emph{global parallelism} of affine space,
its characteristic property.  It is geometrically natural to distinguish
points and displacement vectors.  Displacements form a vector space: we can
add displacements.  However, it does not make sense to add geometric points
in a space: in a flat model of Earth what would be the sum of New York and
Boston?  (Note, however, that weighted averages of points in affine space are
well-defined: there is a definite point 2/3 of the way from New York to
Boston.)  Also, note that there is a distinguished zero displacement vector,
whereas to distinguish a point in the idealized affine space of the world
would be chauvinistic.  Less prosaically, translations act as symmetries of
affine space but not as symmetries of a vector space, since a symmetry of a
vector space must preserve the zero vector.  We denote the vector space of
displacements of space~$E$ as~$V$ and the vector space of displacements of
spacetime~$M$ as~$W$.  It is sometimes helpful to regard elements of~$V,W$ as
global vector fields on~$E,M$, respectively: affine spaces also have a global
\emph{infinitesimal} parallelism.
 
The symmetries of~$V$ are called \emph{linear} transformations and the
symmetries of~$E$ are called \emph{affine} transformations.  If we choose a
basis of~$V$, then we can identify the symmetry group~$\Aut(V)$ as the group
of invertible real square matrices of size~$d=\dim V=\dim E$.  An affine
transformation $f\:E\to E$ is a map which preserves the action of
displacements.  More precisely, if $p_0,p_1,q_0,q_1$ are points in~$E$ such
that the displacements~$p_1-p_0$ and $q_1-q_0$ are equal, then the
displacements $f(p_1)-f(p_0)$ and $f(q_1)-f(q_0)$ are also equal.  It follows
that $f$~induces a well-defined \emph{linear} map~$\df$ on displacements
defined by $\df(\xi )=f(p_1)-f(p_0)$ whenever $p_1=p_0+\xi $ for~$\xi \in V$.
The map~$\df$ is the differential of~$f$ and it is constant on~$E$.  (The
constancy of the derivative characterizes affine maps, and is defined using
the global parallelism.)  The affine maps~$f$ for which $\df$~is the identity
map on displacements are the translations $f(p)=p+\xi _0$ by a fixed
displacement vector~$\xi _0$.  This situation is summarized by the
group extension
  \begin{equation}\label{eq:1}
     1\longrightarrow V\longrightarrow \Aff(E)\longrightarrow
     \Aut(V)\longrightarrow 1 .
  \end{equation}
Here $1$~denotes the trivial group consisting only of an identity element;
$V$~is the abelian group of displacements under the vector addition;
$\Aff(E)$~is the group of affine transformations of~$E$; and $\Aut(V)$~the
group of linear transformations of~$V$.  The map $V\to\Aff(E)$ is the
inclusion of the translation subgroup and $\Aff(E)\to\Aut(V)$ is the
derivative homomorphism $f\mapsto\df$ described above.  Notice that a point
$p\in E$ determines a splitting of~\eqref{eq:1}: the image of~$\Aut(V)$
in~$\Aff(E)$ under the splitting is the subgroup of affine transformations
which fix~$p$.

  \subsection*{Euclidean and Galilean structures}

A geometric structure on the vector space~$V$ induces a translation-invariant
geometric structure on the affine space~$E$.  For example, the space~$E$
in a nonrelativistic system carries a Euclidean metric.  It is specified by a
positive definite inner (dot) product on~$V$, which then determines a
translation-invariant metric on~$E$.  The \emph{Euclidean
group}~$\Euc(E)$ is the group of isometries of~$E$ with its Euclidean
metric.  Since the metric is translation-invariant, $\Euc(E)$~is a subgroup
of~$\Aff(E)$ and it is a group extension
  \begin{equation}\label{eq:2}
     1\longrightarrow V\longrightarrow \Euc(E)\longrightarrow
     O(V)\longrightarrow 1 
  \end{equation}
where $O(V)$~is the orthogonal group of linear transformations which preserve
the inner product.  Relative to an orthonormal basis on~$V$ the orthogonal
group~$O(V)$ is the group of $d\times d$~orthogonal matrices.  Recall that
$O(V)$, hence~$\Euc(E)$, has two components according as a symmetry
preserves or reverses orientation.  In terms of matrices the components are
distinguished by the determinant, which takes values~$\pm1$.  The Euclidean
group contains the classical rigid motions: translations, rotations, and
reflections.

  \begin{definition}[]\label{thm:29}
 A \emph{Galilean structure}~$\Gamma $ on a spacetime~$M$ with translation
group~$W$ is:
  \begin{enumerate}
 \item a distinguished subspace~$V\subset W$ of codimension one with a positive
definite inner product;

 \item a positive definite inner product on the quotient line~$W/V$.
  \end{enumerate}
The pair~$(M,\Gamma )$ is called a \emph{Galilean spacetime}.
  \end{definition}

\noindent
 The subspace~$V$ determines a translation-invariant codimension one
foliation~$\sS$ by affine subspaces, and this gives a global notion of
\emph{simultaneity}: spacetime points on the same leaf occur at the same
time.  There is no distinguished complementary foliation defining worldlines
of spatial points at rest---Galilean observers can undergo boosts.  The
metric on~$W/V$ sets a scale for time; the clock for any observer moves at a
set rate.  There is no set \emph{arrow of time}, which would be an
orientation of~$W/V$, i.e., a positive and negative ``side'' of each
simultaneity slice (leaf of~$\sS$).

  \begin{definition}[]\label{thm:93}
 Let $(M,\Gamma )$ be a Galilean spacetime.  Then $\Aut(M,\Gamma )\subset
\Aff(M)$ is the group of affine transformations of~$M$ which preserve the
Galilean structure~$\Gamma $. 
  \end{definition}

\noindent
 The symmetry group~$\Aut(M,\Gamma )$ is a group extension
  \begin{equation}\label{eq:3}
     1 \longrightarrow W\longrightarrow \Aut(M,\Gamma )\longrightarrow
     \Aut(W,\Gamma )\longrightarrow 1 
  \end{equation}
where $\Aut(W,\Gamma )$~is the space of linear transformations on~$W$ which
preserve the structures in Definition~\ref{thm:29}: the subspace~$V$, its
metric, and the metric on the quotient~$W/V$.  As above, let $d=\dim V$; then
$\dim W=d+1$.  Choose a basis of~$W$ so that the first basis vector projects
onto a unit vector in~$W/V$ and the last $d$~basis vectors form an
orthonormal basis of~$V$.  Then elements of~$\Aut(W,\Gamma )$ are represented
by invertible $(d+1)\times (d+1)$ matrices of the triangular form
  \begin{equation}\label{eq:4}
     \begin{pmatrix} \pm1&0 \\ *&A 
     \end{pmatrix}
  \end{equation}
where $A$~is a $d\times d$ orthogonal matrix.  If $d>0$ the
group~$\Aut(W,\Gamma )$, hence also~$\Aut(M,\Gamma )$, has four components
according to~$\det A=\pm1$ and the sign of the upper left entry
of~\eqref{eq:4}.  More invariantly, an affine transformation which
preserves~$\Gamma $ also preserves the foliation~$\sS$ of simultaneous
events, and it might or might not reverse the orientation of each leaf
of~$\sS$ and might or might not reverse the orientation of the normal bundle
to~$\sS$: the arrow of time.  In other words, there are surjective (onto)
homomorphisms
  \begin{equation}\label{eq:10}
     t,p\:\Aut(M,\Gamma )\longrightarrow \pmo 
  \end{equation}
defined on transformations~$f\:M\to M$ such that $t(f)=\pm1$ according as
$f$~preserves/reverses the arrow of time, and $p(f)=\pm1$ according as
$f$~preserves/reverses the orientation of space.  The value of~$t$ on the
matrix in~\eqref{eq:4} is the upper left corner; the value of~$p$ is~$\det
A$.  Any transformation~$f$ with~$t(f)=-1$ is called \emph{time-reversing}.
For later purposes we single out the following subgroup with two components.

  \begin{definition}[]\label{thm:30}
  $\Aut^+(M,\Gamma )=p\inv (+1)\subset \Aut(M,\Gamma )$. 
  \end{definition}

\noindent
 In other words, $\Aut^+(M,\Gamma )$ is the subgroup of~$\Aut(M,\Gamma )$
consisting of transformations which are orientation-preserving on each leaf
of~$\sS$.  It contains time-reversing transformations.

The fundamental group of~$\Aut(M,\Gamma )$ is isomorphic to that of the
special orthogonal group~$SO(d)$ of space.  The latter has a nontrivial
double cover,\footnote{There are also double cover groups of the
two-component orthogonal group~$O(d)$, and these ``pin groups''~\cite{ABS}
can occur in systems with parity-reversing symmetries.}  the spin
group~$\Spin(d)$, which is simply connected if~$d\ge3$.  Note that
$\Spin(0)$~and $\Spin(1)$ are each cyclic of order two; $\Spin(2)$~is
isomorphic to the circle group, so is connected but not simply connected.

  \begin{definition}[]\label{thm:3}
 The \emph{Galilean group} of~$(M,\Gamma )$ is the spin double cover of the
identity component of~$\Aut(M,\Gamma )$.
  \end{definition}

\noindent
 Some authors might use `Galilean group' to refer to~$\Aut(M,\Gamma )$ or
some other group of symmetries of spacetime; as remarked earlier our choice
is analogous to the definition of `Poincar\'e group'
in~\cite[Glossary]{Detal}.  We can realize the Galilean group as the group of
symmetries of a geometric structure.  Namely, endow the vector space~$V$ with
an orientation and spin structure and also orient the quotient~$W/V$.  This
induces a translation-invariant orientation and spin structure on the
simultaneity foliation~$\sS$ as well as an orientation on the normal bundle
to~$\sS$, an arrow of time.  The Galilean group is the group of symmetries
of~$M$ which preserve these structures.  The Galilean group is connected
for~$d\ge2$ and has two components for~$d=0,1$.

  \subsection*{Tying abstract symmetry groups to spacetime}

  \begin{definition}[]\label{thm:94}
 Let $(M,\Gamma )$ be a Galilean spacetime.  A \emph{Galilean symmetry group}
is a triple~$(G,\gamma ,j)$ consisting of a Lie group~$G$, a homomorphism 
  \begin{equation}\label{eq:98}
     \gamma \:G\longrightarrow \Aut(M,\Gamma ), 
  \end{equation}
and a splitting 
  \begin{equation}\label{eq:99}
     j\:W\cap\gamma (G)\longrightarrow G 
  \end{equation}
of~$\gamma $ over the group of spacetime translations in the image of~$\gamma
$.  The image of~$j$ is required to be a normal subgroup of~$G$.
  \end{definition}

  \begin{remark}[]\label{thm:95}
 The homomorphism~$\gamma $ need not be surjective or injective.  The
surjectivity is violated if, for example, $G$~is the group of symmetries of a
system which has only a proper subgroup of Galilean symmetries, as with the
crystals below.  If there is extra internal structure on spacetime---say
extrinsic ``spins'' at each spacetime point or each point of a crystal in
spacetime---then $\ker\gamma $ may consist of internal symmetries which fix
the points of spacetime.  Any such internal structure should have a global
parallelism extending that of the affine space~$M$, which is why we
hypothesize the splitting~\eqref{eq:99} over spacetime translations.  On the
other hand, internal structures need not be homogeneous with respect to
rotations, reflections, and boosts, which is why we only hypothesize the
splitting over translations.  The quotient of~$G$ by the image of~$j$ maps to
the group~$\Aut(W,\Gamma )$ of linear symmetries.  We imagine these acting on
some model internal structure in the linear space~$W$, and it is the
existence of this quotient which justifies the normality assumption on the
image of~$j$.
  \end{remark}

  \begin{remark}[]\label{thm:127}
 There is a similar definition for Minkowski spacetimes. 
  \end{remark}

In some situations a time direction is distinguished, and this breaks the
boost symmetries: 

  \begin{definition}[]\label{thm:96}
 Let $(M,\Gamma )$ be a Galilean spacetime.  A \emph{time direction} is a
splitting of the exact sequence 
  \begin{equation}\label{eq:100}
     0 \longrightarrow V\longrightarrow W\longrightarrow W/V\longrightarrow
     0 
  \end{equation}
of spacetime translations. 
  \end{definition}

\noindent
 Let $U\subset W$ be the image of the splitting $W/V\hookrightarrow W$.  Then
$U$~defines a distinguished translation-invariant set of affine lines in~$M$
(each with tangent~$U$) which are transverse to the leaves of the
simultaneity foliation~$\sS$.  They are worldlines of distinguished
observers.  The inner product on~$W/V$ induces one on~$U$, and so two
distinguished unit-speed motions on each distinguished worldline.  These are
distinguished clocks in spacetime, though without a time-orientation the
clocks can run either forward or backward (and no observer can distinguish
forward from backward).

  \subsection*{Crystals and crystallographic groups}

We define subgroups of~$\Aut(M,\Gamma )$ of interest in condensed matter
physics, namely those which preserve a \emph{crystal}~$C$ in spacetime.

  \begin{definition}[]\label{thm:31}
 Let $(M,\Gamma )$ be a Galilean spacetime.  
  \begin{enumerate}
 \item A \emph{crystal}~$C$ is a subset of~$M$ such that (i)~the
subgroup~$\Pi \subset V$ of spatial translations which preserve~$C$ is a
lattice of full rank, and (ii)~there is a line~$U\subset W$ of spacetime
translations complementary to~$V$ such that translations in~$U$ preserve~$C$.

 \item Let $G(C)\subset \Aut(M,\Gamma )$ be the subgroup of transformations
which preserve a crystal~$C$.  The quotient~$G(C)/U$ is the \emph{spacetime
crystallographic group} associated to the crystal~$C$.

  \end{enumerate}
  \end{definition}

\noindent
 The subspace~$U$ is a time direction in the sense of
Definition~\ref{thm:96}, so $G(C)$~contains no Galilean boosts.  Note that
the matrices~\eqref{eq:4} which preserve the direct sum decomposition
$W=U\oplus V$ are block diagonal---the boost component~$*$ must be the zero
vector.  The spacetime crystallographic group is an extension
  \begin{equation}\label{eq:6}
     1\longrightarrow  \Pi \longrightarrow G (C)/U\longrightarrow
     \hP\longrightarrow 1 
  \end{equation}
whose translation subgroup consists of the lattice~$\Pi $ of discrete spatial
translations; the quotient is called the \emph{magnetic point group}
$\hP\subset \Aut(W,\Gamma )$.  The crystallographic group is called
\emph{symmorphic} if \eqref{eq:6}~splits and \emph{non-symmorphic} otherwise.
The homomorphism $t\:\Aut(W,\Gamma )\to\pmo$ in~\eqref{eq:10},
which maps a linear Galilean transformation onto its action on~$W/V$---which
is either the identity (time-preserving) or minus the identity
(time-reversing)---fits the magnetic point group into a group extension
  \begin{equation}\label{eq:7}
     1\longrightarrow P\longrightarrow \hP\xrightarrow{\;\; t\;\; }
     \pmo\longrightarrow 1 
  \end{equation}
whose kernel~$P$ is the \emph{point group}.  Both $P$~and $\hP$~are finite
groups.  The extension~\eqref{eq:7} is naturally split by the
involution~$\epsilon $ of~$W=U\oplus V$ which is~$-1$ on~$U$ and~$+1$ on~$V$.
Conjugation by~ $\epsilon $ reverses time translations and fixes spatial
translations.  

  \begin{example}[]\label{thm:97}
 Let $\MM^3=\RR\times \EE^2$ be the Cartesian product of the standard time line
and standard Euclidean 2-space with the evident Galilean structure.
Define~$C$ to be the Cartesian product of~$\RR$ with~$C_0\subset \EE^2$,
where 
  \begin{equation}\label{eq:101}
     C_0=\bigl\{(x^1,x^2):x^1,x^2\in \ZZ\bigr\}\;\cup\;\bigl\{ (x^1+\delta
     ,x^2+1/2): x^1,x^2\in \ZZ\bigr\} 
  \end{equation}
for~$0<\delta <1/2$.  Let $G\subset \Euc(\EE^2)$ be the group of Euclidean
symmetries which preserve~$C_0$.  It is the subgroup of the crystallographic
group~$G(C)/U$ of time-preserving transformations modulo time translations.
Then $G$~is a group extension
  \begin{equation}\label{eq:102}
     1\longrightarrow \Pi \longrightarrow G\longrightarrow \zt\times
     \zt\longrightarrow 1 
  \end{equation}
Here $\Pi \subset \RR^2$ is the full lattice of translations by vectors~$(n
^1,n ^2)$ with integer components~$n ^1,n ^2\in \ZZ$.  The quotient
group~$\zt\times \zt$ is generated by reflections~$g_1'',g_2''$ in the two
axes of the vector space~$\RR^2$.  Then $g_1''$~lifts to the \emph{glide
reflection}
  \begin{equation}\label{eq:103}
     g_1\:(x^1,x^2)\longmapsto (-x^1+\delta ,x^2+1/2) 
  \end{equation}
which interchanges the two sets in~\eqref{eq:101}, whereas $g_2''$~lifts to
the reflection 
  \begin{equation}\label{eq:104}
     g_2\:(x^1,x^2)\longmapsto (x^1,-x^2). 
  \end{equation}
Note that $(g_1)^2$~is translation by~$(0,1)$ whereas $(g_2)^2$~is the
identity.  Any other lift of~$g_1''$ is~$\tau g_1$ for some~$\tau \in \Pi $,
and $(\tau g_1)^2$~is again a nonzero translation.  This proves that
\eqref{eq:102}~is not split, and so $G(C)/U$~is non-symmorphic.   
 
For later use we note that the commutator $g_2\inv g_1\inv g\mstrut
_2g\mstrut _1$ is translation by~$(0,1)$.
  \end{example}

Finally, some condensed matter systems are described by specifying a subset
of $M\times V$ rather than a subset of~$M$.  Intuitively, this is a set of
points with spatial tangent vectors attached to each point.\footnote{More
precisely, we should allow either~$V$ or ${\textstyle\bigwedge} ^{d-1}V^*$
depending on whether we attach vectors or axial vectors.  Physically, these
represent localized displacements (as in ferroelectrics) or current loops and
atomic spins (as in ferromagnetics).  One could of course consider more
general order parameters.}

  \begin{definition}[]\label{thm:32}
  Let $(M,\Gamma )$ be a Galilean spacetime.  A \emph{spin crystal}~$C$ is a
subset of~$M\times V$ such that (i)~the subgroup~$\Pi \subset V$ of spatial
translations whose elements preserve~$C$ is a full lattice, and (ii)~there is
a line~$U\subset W$ of spacetime translations complementary to~$V$ which
preserve~$C$.
  \end{definition}

\noindent
 Since we can take the tangent vectors to be zero, Definition~\ref{thm:32}
generalizes Definition~\ref{thm:31}.  As before, the symmetry group~$G(C)$
contains the group~$U$ of time translations, the quotient~$G(C)/U$ is an
extension~\eqref{eq:6}, and we define the magnetic point group to be the
quotient~$\hP=G(C)/(U\times \Pi )$.  But for spin crystals the homomorphism
$t\:\hP\to\pmo$ need not be surjective---there might not be time-reversal
symmetries---and even if $t$~is surjective the extension~\eqref{eq:7} might
not split, as the following example illustrates.

  \begin{example}[]\label{thm:142}
 We use the standard Galilean spacetime~$\MM^3=\RR\times \EE^2$ of
Example~\ref{thm:97}.  Fix $0<\delta _1\not= \delta _2<1/2$ and a nonzero
vector~$\xi \in V=\RR^2$, and set $\ZZ^2=\{(x^1,x^2):x^1,x^2\in \ZZ\}\subset
\EE^2$.  Let
  \begin{equation}\label{eq:153}
     \begin{aligned} 
   C_0 =   
           \quad\; &\bigl((\phantom{-}\delta _1,\phantom{-}\delta _2)+\ZZ^2
           \bigr)\times \{\xi \} 
         \;\cup\;  \bigl((-\delta _2,\phantom{-}\delta _1)+\ZZ^2 \bigr)\times
           \{-\xi \}\\ 
         \;\cup\;   &\bigl((-\delta _1,-\delta _2)+\ZZ^2 \bigr)\times
         \{\xi \} 
         \;\cup\;  \bigl((\phantom{-}\delta _2,-\delta _1)+\ZZ^2 \bigr)\times \{-\xi \}
      \end{aligned} 
  \end{equation}
which is a subset of $\EE^2\times \RR^2$.  The spin crystal is $C=\RR\times
C_0\subset \MM^3\times \RR^2$.  Then $\hP$~is cyclic of order~4 with
generator the linear transformation of~$W=\RR\oplus \RR^2$ which is $-1$~on
the first summand and a $\pi /2$~rotation on the second.  The homomorphism
$t\:\hP\cong \zmod4\to\zt$ is a nonsplit surjection.
  \end{example}

  \begin{remark}[]\label{thm:1}
 As a crystal breaks symmetries which mix time and space, it can equally be
embedded in a relativistic (Minkowski) spacetime and its symmetry group~$G$
embedded in the symmetry group of relativistic quantum mechanics.
  \end{remark}

   \section{Extended \Qsymmetry classes}\label{sec:10}

  \subsection*{Quantum symmetries and time-reversal}

Resuming the discussion in~\S\ref{sec:2}, suppose $G$~is a Galilean symmetry
group (Definition~\ref{thm:94}) on a Galilean spacetime~$(M,\Gamma )$, and
let $G$~act as projective quantum symmetries on a Hilbert space~$\sH$.  Let
$\rho \:G\to\QAut(\PH)$ denote the action of~$G$ , and let $(G,\phi ,\tau )$
be the induced \Qsymmetry type (Definition~\ref{thm:4}).  Recall that the
homomorphism
  \begin{equation}\label{eq:65}
     \phi \:G\longrightarrow \pmo
  \end{equation}
tracks which symmetries are implemented linearly and which antilinearly.
Composing the homomorphism~\eqref{eq:98} with ~\eqref{eq:10} we obtain a
homomorphism
  \begin{equation}\label{eq:60}
     t\:G\longrightarrow \pmo
  \end{equation}
which tracks time-reversal: $t(g)=1$~if a symmetry~$g$ preserves the arrow of
time and $t(g)=-1$ if $g$~reverses the arrow of time.  \emph{A priori} the
homomorphisms~$\phi ,t$ are independent, but there is a standard argument
which is often invoked to identify them.  We review it now.

A one-parameter subgroup $\RR\to\tG$ induces via~$\rho ^\tau $
in~\eqref{eq:17} a one-parameter group~$s\mapsto e^{is\sO /\hbar}$ of unitary
operators on~$\sH$, where $\sO $~is a self-adjoint operator.  (The
parameter~$s$ typically does not have dimensions of time.)  In particular,
suppose $(M,\Gamma )$ has a distinguished time direction
(Definition~\ref{thm:96}); the Galilean symmetry group~$G$ includes the
distinguished time translation subgroup~$U\subset G$, which comes with a
fixed isomorphism $\RR\cong U$; and assume that $\tG\to G$ is split
over~$U\subset G$.  Then time translation gives a one-parameter
group~$t\mapsto e^{-itH/\hbar}$ of unitary transformations whose self-adjoint
generator~ $H$ is the \emph{Hamiltonian} of the distinguished time
evolution.\footnote{The sign convention~\cite{DF2} is that the self-adjoint
Hamiltonian~ $H$ corresponds to \emph{minus} infinitesimal time translation.}
We assume oriented time directions in the sequel, and so have systems with a
well-defined Hamiltonian.

  \begin{lemma}[]\label{thm:90}
 \  
  \begin{enumerate}
 \item For all $g\in \tG$ we have 
 \begin{equation}\label{eq:62}
     \tar (g)H=\phi (g)t(g)H\tar (g).
  \end{equation} 

 \item If $H$~ is bounded below and unbounded above, then $\phi =t$.

  \end{enumerate} 
   \end{lemma}

  \begin{proof}
 The distinguished time translations commute or anticommute with~$g\in \tG$,
according to~$t(g)$, and so
  \begin{equation}\label{eq:147}
     \tar(g)e^{-itH/\hbar} = e^{-t(g)itH/\hbar}\tar(g).
  \end{equation}
Therefore, $\tar(g)iH = t(g)iH\tar(g)$, from which \eqref{eq:62}~follows.
The conclusion in~(ii) is immediate since if $\phi (g)t(g)=-1$, then
$\tar(g)$~flips the spectrum of~$H$.
  \end{proof}

\noindent
 Because of this lemma it is usually assumed that time-reversing symmetries
are antiunitary and time-preserving symmetries are unitary.  However, the
hypotheses do not apply to some Dirac Hamiltonians in condensed matter
systems, and so we allow $\phi $~and $t$~to be independent.  

Define
  \begin{equation}\label{eq:61}
     c=\phi t\:G\longrightarrow \pmo .
  \end{equation}
This homomorphism tracks whether a symmetry commutes or anticommutes
with~$H$.  A symmetry which commutes with~$H$ is called
\emph{Hamiltonian-preserving}; a symmetry which anticommutes with~$H$ is
called \emph{Hamiltonian-reversing}.  Of course, any two of~$t,c,\phi $
determine the third.

  \subsection*{Extended \Qsymmetry classes and gapped systems}

We now augment Definition~\ref{thm:4}(ii) of a \Qsymmetry class and
Definition~\ref{thm:15} of a twisted representation to include the
homomorphism~\eqref{eq:61} which tracks whether symmetries commute or
anticommute with the Hamiltonian.

  \begin{definition}[]\label{thm:84}
 \  
  \begin{enumerate}
 \item An \emph{extended \Qsymmetry class} is an isomorphism class of
quadruples~$(G,\phi ,\tau,c)$ where $G$~is a Lie group, $\phi
\:G\to\pmo$ is a continuous homomorphism, $\tG$~is a $\phi $-twisted 
extension, and $c\:G\to\pmo$ is a continuous homomorphism.

 \item A homomorphism $\tar\:\tG\to\Aut_{\RR}(E)$ to real automorphisms of a
$\zt$-graded complex vector space~$E$ is a \emph{$(\phi ,\tau ,c)$-twisted
representation of~$G$} if it satisfies the conditions of
Definition~\ref{thm:15} and in addition $\tar(g)$~is even if $c(g)=1$ and odd
if~$c(g)=-1$.
  \end{enumerate} 
   \end{definition}

\noindent 
 As before, we often call a quartet~$(G,\phi ,\tau ,c)$ an extended \Qsymmetry
class, though strictly speaking the \Qsymmetry class only remembers the
isomorphism class of the $\phi $-twisted  extension $\tG\to G$.
Lemma~\ref{thm:90}(i) specializes to the condition in
Definition~\ref{thm:84}(ii) if we take $H$ to be the grading operator on~$E$.

We are now ready to abstract some basic properties of gapped free fermion
systems.

  \begin{definition}[]\label{thm:25}
 A \emph{gapped system with extended \Qsymmetry class~$(G,\phi ,\tau ,c)$} is
a triple~$(\sH,H,\tar )$ consisting of a complex separable $\zt$-graded
Hilbert space~$\sH$, a self-adjoint operator~$H$ acting on~$\sH$, and a
$(\phi ,\tau ,c)$-twisted representation $\tar \:\tG\to\QAut(\sH)$ such that
  \begin{enumerate}
 \item the commutation relation~\eqref{eq:62} holds: 
  \begin{equation}\label{eq:149}
     \tar (g)H=c(g)H\tar (g);
  \end{equation}

 \item the Hamiltonian~$H$ is invertible with bounded inverse. 
  \end{enumerate} 
  \end{definition}

\noindent 
 If $H$~is an unbounded self-adjoint operator, then (ii)~means that 0~is in
its resolvent set, i.e., not in its spectrum.

  \begin{remark}[]\label{thm:98}
 It is probably better to include the Hamiltonian as part of~$G$ by assuming
that $G$~is a Galilean symmetry group with a distinguished oriented time
direction~$U\subset G$, which then comes with a fixed isomorphism $\RR\cong
U$, and further assuming that $\tG\to G$ is split over~$U\subset G$.  Then
the Hamiltonian~$H$ is computed from the 1-parameter group~$e^{-itH/\hbar}$
which is part of the representation~$\tar$.  Instead, we exclude time
translations from our symmetry groups, as in Definition~\ref{thm:31}(ii) of
crystallographic groups, and so keep the Hamiltonian distinct from the
symmetry group.  Nonetheless, for the definitions in Appendix~\ref{sec:14} it
is useful to incorporate the Hamiltonian into the symmetry, and we do so by
constructing a \Qsymmetry class~$(\tilG,\tp,\tilde\tau )$ from an extended
\Qsymmetry class as follows.  Let $\tilG=G\ltimes_c\RR$ where $G$~acts
on~$\RR$ via the homomorphism $c\:G\to\pmo$.  (See Definition~\ref{thm:87}
for semidirect products.)  Similarly, set $\tilG^\tau =\tG\ltimes \RR$ where
$\tG$~acts via the composition $\tG\to G\xrightarrow{c}\pmo$.  Finally,
define $\tp\:\tilG^\tau \to\tG\to G\xrightarrow{\phi }\pmo$.  Then a gapped
system with extended \Qsymmetry class~$(G,\phi ,\tau ,c)$ is precisely a
$(\tp,\tilde\tau )$-twisted representation of~$\tilG$ such that the
self-adjoint generator~$H$ of the unitary 1-parameter group $\RR\subset
\tilG^\tau \to U(\sH)$ does not contain~0 in its spectrum.
  \end{remark}

  \begin{remark}[]\label{thm:92}
 We will add a finiteness hypothesis on the Hilbert space~$\sH$ for
particular classes of gapped systems, depending on the dimension of the
system and the symmetry group.  These will be spelled out in~\S\ref{sec:11}
and~\S\ref{sec:6}.  The representation~$\tar $ obeys conditions which are
specified in Definition~\ref{thm:25}.  The second hypothesis in
Definition~\ref{thm:25}, that the Hamiltonian be invertible, is specific to
``gapped'' systems of free fermions.  We do \emph{not} assume any boundedness
of the Hamiltonian, such as appear in Lemma~\ref{thm:90}(ii).  The assumption
that 0~is not in the spectrum is a bit arbitrary.  We can shift the
Hamiltonian by a constant without changing the physics; our assumption is
really that there is a spectral gap.  Often one number, the \emph{Fermi
energy level}, is singled out in that gap.  We take that Fermi level to be~0.
We have nothing to say in this paper about systems without a gap.
  \end{remark}

   \section{Free fermions}\label{sec:17}

  \subsection*{Symmetries in algebraic quantum mechanics}

Quantum mechanics is usually described, as in~\S\ref{sec:2}, in terms of a
Hilbert space of states.  Here we focus instead on the algebra of
observables\footnote{The second author objects to this standard terminology
since the observables, which are self-adjoint, do not form an associative
algebra (though they do form a Jordan algebra).}  in an abstract form.  For a
recent exposition of algebraic quantum mechanics, see~\cite[\S2]{FR}.

Let~$\sA$ be a \emph{complex topological\footnote{In the examples we consider
$\sA$~is finite dimensional, so there is a unique topology compatible with
the vector space structure.} $*$-algebra}.  Thus $\sA$~is a complex
topological vector space equipped with a continuous antilinear involution
$*\:\sA\to\sA$ which satisfies $(ab)^* = b^*a^*$ for all~$a,b\in \sA$.  An
element~$a\in \sA$ is \emph{real} (\emph{self-adjoint}) if $a^*=a$.  A
self-adjoint element is an abstract observable; a \emph{Hamiltonian} is a
distinguished observable which is used to generate motion.

  \begin{remark}[]\label{thm:177}
 The abstract system described by~$\sA$ can be made concrete by the choice of
an \emph{$\sA$-module}, a complex separable Hilbert space~$\sH$ and a
continuous $*$-homomorphism $\sA\to\End(\sH)$ into bounded\footnote{One can
also allow unbounded operators.  In our examples $\sA$~and $\sH$~are finite
dimensional and all operators are bounded.} operators on~$\sH$.  A
\emph{state} on~$\sA$ is a positive linear functional $\sA\to\CC$, and at
least for $C^*$-algebras a state determines an $\sA$-module via the
GNS~construction.
  \end{remark}

For systems with fermions it is natural to assume that the operator
algebra~$\sA$ is $\zt$-graded.  In that case the Koszul sign rule demands
that the $*$~involution satisfy
  \begin{equation}\label{eq:180}
     (ab)^* = (-1)^{|a|\,|b|} b^*a^* 
  \end{equation}
for homogeneous elements~$a,b\in \sA$.  Here $|a|\in \{0,1\}$ is the parity
of~$a$.   

  \begin{definition}[]\label{thm:171}
 \ 

  \begin{enumerate}
 \item Let $G$~be a topological group with $\zt$-grading $\phi \:G\to\pmo$,
and let $\sA$~be a $\zt$-graded topological $*$-algebra.  An \emph{action}
of~$(G,\phi )$ on~$\sA$ is a homomorphism
  \begin{equation}\label{eq:181}
     \alpha \:G\longrightarrow \Aut_{\RR}(\sA) 
  \end{equation}
such that for all~$g\in G$ the real linear map~ $\alpha (g)$ is complex
linear or antilinear according as $\phi (g)=+1$ or~$-1$ and $\alpha(g)$~
preserves the $*$-structure: 
  \begin{equation}\label{eq:217}
     \alpha (g)(a^*) = \bigl[\alpha (g)(a)\bigr]^*
  \end{equation}
for all~$a\in \sA$.

\item  Let  $H\in  \sA$  be  a  Hamiltonian.   The  action~$\alpha  $  is  a
\emph{symmetry}  of~$(\sA,H)$ if  there exists  a  homomorphism $c\:G\to\pmo$
such that
  \begin{equation}\label{eq:182}
     \alpha (g)(H) = c(g)H,\qquad g\in G. 
  \end{equation}

  \end{enumerate}
  \end{definition}

\noindent
 The homomorphism~$c$ is the operator algebra
analog of the homomorphism~$c$ in~\eqref{eq:149}.

  \begin{remark}[]\label{thm:172}
 If $\sH$ is an \emph{irreducible} $\sA$-module preserved by the
automorphism~$\alpha (g)$ of~$\sA$, then that automorphism is realized
on~$\sH$ by a line of complex linear or antilinear automorphisms.  This leads
to a $\phi $-twisted extension of~$G$ and the (easy) analog of Wigner's
theorem in this algebraic context.  As before, the $\zt$-grading~$\phi $ also
encodes the complex linearity vs. antilinearity of the real linear operators
on~$\sH$.
  \end{remark}

  \subsection*{Free fermions}

Let\footnote{The notation `$\sM$' derives from `mode space'.} $\sM$~be a
finite dimensional real vector space with positive definite inner
product~$b$.  It encodes a system of finitely many fermions for which the
operator algebra is 
  \begin{equation}\label{eq:183}
     \sA = \CMbC. 
  \end{equation}
Recall that the \emph{Clifford algebra} $\CMb$ is the free unital real
algebra generated by~$\sM$ subject to the canonical (anti)commutation
relations
  \begin{equation}\label{eq:184}
     \xi _1\xi _2 + \xi _2\xi _1 = 2b(\xi _1,\xi _2),\qquad \xi _1,\xi _2\in
     \sM. 
  \end{equation}
(See Appendix~\ref{sec:12} for a review of Clifford algebras.)  Define the
$*$-structure on~\eqref{eq:183} by\footnote{From the point of view of
canonical quantization it is more natural~\cite[p.~679]{Detal} to work
with~$\hat\xi =(\sqo\,\hbar)^{1/2}\xi $, in which case the $*$-operation
becomes complex conjugation.}
  \begin{equation}\label{eq:185}
     \xi ^* = \sqo\,\xi ,\qquad \xi \in \sM. 
  \end{equation}
Then, for example, from~\eqref{eq:180} we deduce 
  \begin{equation}\label{eq:186}
     (\xi _1\xi _2)^* = -\xi _2^*\xi _1^* = \xi _2\xi _1\qquad
     \textnormal{for all $\xi _1,\xi _2\in\sM$}. 
  \end{equation}
 
Let $\OMb$~denote the orthogonal group of~$\Mb$ and $\oMb$~its Lie algebra.
There is an embedding $\oMb\otimes \CC\subset \sA$. 

  \begin{definition}[]\label{thm:173}
 A \emph{free fermion Hamiltonian}\footnote{More fundamentally, a free
fermion Hamiltonian is an even element in the second (quadratic) filtered
subspace of the Clifford algebra, so has the form $H^{ij}e_ie_j+c$ for
$H^{ij}\in \CC$~arbitrary and $c\in \CC$.  Self-adjointness then leads to
the same definition we give in the text, except we omit the constant~$c$
(which must be real).} is a self-adjoint element of $\oMb\otimes \CC\subset
\sA$.
  \end{definition}

\noindent
 Let $e_1,\dots ,e_n$ be an orthonormal basis of~$\sM$.  An element of
$\oMb\otimes \CC\subset \sA$ has the form~$H^{ij}e_ie_j$ for some~$H^{ij}\in
\CC$ which satisfy $H^{ij}=-H^{ji}$.  The self-adjointness implies
$H^{ij}=\sqo A^{ij}$ for $A^{ij}\in \RR$, and necessarily $A^{ij} = -A^{ji}$.

  \begin{construction}[]\label{thm:174}
 Let $\Mb$ be a finite dimensional real inner product space and $(G,\phi )$ a
$\zt$-graded topological group.  An {action} of~$G$ on the free fermion
system based on~$\Mb$ is determined by a continuous homomorphism $\beta
\:G\to \OMb$.  Extend the induced homomorphism $G\to\Aut_{\RR}\bigl(\CMb
\bigr)$ to a homomorphism~\eqref{eq:181} into real automorphisms of
$\sA=\CMbC$ with $\alpha (g),\;g\in G,$~complex linear or antilinear
according as $\phi (g)=+1$ or~$-1$.  Recall that an orthogonal transformation
$P\in \OMb$ acts as an automorphism of the real Clifford algebra~$\CMb$ by
  \begin{equation}\label{eq:187}
     P\cdot (\xi _1\cdots \xi _k) = P\xi _1\cdots P\xi _k,\qquad \xi _i\in
     \sM. 
  \end{equation}
Let $\xi _1,\dots ,\xi _k\in \sM,\; x,y\in \RR$.  Then $g\in G$ acts on~$\sA$ by
  \begin{equation}\label{eq:212}
     \alpha (g)\bigl(\xi _1\cdots\xi _k\otimes (x+\sqo\,y) \bigr) =
     \begin{cases} 
     \beta
     (g)\xi _1\cdots \beta (g)\xi _k\otimes (x + \sqo\,y),&\phi (g)=1;\\\beta
     (g)\xi _1\cdots \beta (g)\xi _k\otimes (\sqo)^k(x - \sqo\,y),&\phi
     (g)=-1.\end{cases} 
  \end{equation}
The extra factor of~$(\sqo)^k$ is explained by footnote~15, in which $\hat\xi
$~belongs to the subspace of~$\sA$ fixed by~$*$.  The action~\eqref{eq:212}
satisfies~\eqref{eq:217}.

The extension to real automorphisms of~$\sA$ provides the data of
Definition~\ref{thm:171}. 
  \end{construction}

  \subsection*{Free fermion Hamiltonians and classical symmetric spaces}

We reprove a theorem of Altland-Heinzner-Huckleberry-Zirnbauer~\cite{AZ,HHZ}.

  \begin{theorem}[]\label{thm:175}
 Suppose the $\zt$-graded group~$(G,\phi )$ acts on the free fermion system
based on the inner product space~$\Mb$, as in Construction~\ref{thm:174}.  Fix
$c\:G\to\pmo$.  Let $\fp$ denote the real vector space of free fermion
Hamiltonians which are invariant under the action~$\alpha $, as
in~\eqref{eq:182}.  Then $\exp(\fp)$~is a classical symmetric space of
compact type. 
  \end{theorem}

\noindent
 Recall that a free fermion Hamiltonian has the form $H=\sqo\,A^{ij}e_ie_j$
for $A^{ij}\in \RR$ satisfying $A^{ij}=-A^{ji}$.  Set~$t=\phi c$.  Then
from~\eqref{eq:182} we deduce 
  \begin{equation}\label{eq:188}
     \alpha (g)(A^{ij}e_ie_j) = t(g)(A^{ij}e_ie_j),\qquad g\in G. 
  \end{equation}
Note $A^{ij}e_ie_j\in \CMb$ lies in the \emph{real} Clifford algebra.
Identify~$A^{ij}e_ie_j$ with an element~$A\in \oMb$.  From the
definition~\eqref{eq:187} of~$\alpha $ we deduce 
  \begin{equation}\label{eq:189}
     \fp = \{A\in \oMb:\Ad_{\beta (g)}(A) = t(g)A\textnormal{ for all $g\in
     G$}\}. 
  \end{equation}
We may have $\fp=0$, in which case $\exp(\fp)$~is a single point.

  \begin{remark}[]\label{thm:176}
 Note that $\exp(\fp) = \{e^{tA}:t\in \RR,\;A\in \fp\}= \{e^{-\sqo\, tH}:t\in
\RR,\, \sqo\,H\in \fp\}$, which may be construed as the space of unitary
evolutions of the free fermion Hamiltonians.  Here we exponentiate elements
of~$\oMb$ to real orthogonal transformations of~$\sM$.  We have not chosen an
$\sA$-module---here a fermionic Fock space---which would be the usual arena
for unitary evolution of the free fermion.
  \end{remark}

  \begin{proof}
 Define 
  \begin{equation}\label{eq:190}
     \fk = \{A\in \oMb:\Ad_{\beta (g)}(A) = A\textnormal{ for all $g\in
     G$}\}. 
  \end{equation}
Use the embedding 
  \begin{equation}\label{eq:192}
     \begin{aligned} \fk\oplus \fp &\longrightarrow \oMb\oplus \oMb \\
     (A_1,A_2) 
      &\longmapsto (A_1+A_2,A_1-A_2 )\end{aligned} 
  \end{equation}
to induce a Lie algebra structure on~$\fh = \fk\oplus \fp$, and then an easy
check shows
  \begin{equation}\label{eq:191}
     [\fk,\fk]\subset \fk,\qquad [\fk,\fp]\subset \fp,\qquad [\fp,\fp]\subset
     \fk. 
  \end{equation}
Since the Killing form on $\oMb\oplus \oMb$ is negative definite, the same is
true for~ $\fh$ and~$\fk$, whence both are Lie algebras of compact type.  It
follows that $\exp(\fp)$~is a compact symmetric space.

To see that $\fh$~and $\fk$~ are classical Lie algebras---that is, matrix
algebras over~$\RR$, $\CC$, or~$\HH$---we use an extended form of
Theorem~\ref{thm:151} in Appendix~\ref{sec:9}.  That theorem asserts that the
commutant of an \emph{irreducible} real representation is a division
algebra~$D$ over~$\RR$, so is isomorphic to~$\RR$, $\CC$, or~$\HH$.  An
\emph{isotypical} representation of~$G$ has the form $U\otimes V$, where
$G$~acts irreducibly on the real vector space~$V$ and trivially on the real
vector space~$U$.  Then the commutant is $\End(U)\otimes D$, which is a
matrix algebra over~$\RR$, $\CC$, or~$\HH$.  Finally, the commutant of an
arbitrary representation of~$G$ is a sum of matrix algebras over~$\RR$,
$\CC$, or~$\HH$.  Now~$\fk$ is by definition~\eqref{eq:190} the commutant of
a representation of~$G$ on~$\sM$, whence is classical.  To prove that
$\fh$~is classical, consider the representation of~$G$ on~$\sM\oplus \sM$
given by
  \begin{equation}\label{eq:193}
     \begin{alignedat}{2} g&\longmapsto \begin{pmatrix} {\beta (g)}  &
      0\\[1em]0&{\beta (g)}  \end{pmatrix},\qquad &t(g)&=+1, \\[2em]
     g&\longmapsto 
      \begin{pmatrix} 0&{\beta (g)} \\[1em]{\beta (g)}  &
     0 \end{pmatrix},\qquad 
      &t(g)&=-1.\end{alignedat} 
  \end{equation}
The intersection of the commutant with the subalgebra of block diagonal
matrices is~$\fh$, embedded in~$\End(\sM)\oplus \End(\sM)$ by~\eqref{eq:192},
so $\fh$~is also classical.
  \end{proof}

  \begin{remark}[]\label{thm:187}
 There are 10~sequences of symmetric spaces based on simple groups, and all
of them can be realized by choosing~$G$ to be an appropriate Pin group using
the construction of~\cite[\S24]{Mi}.  See~\cite{Mo} for further details.
  \end{remark}

  \begin{remark}[]\label{thm:189}
 There is a parallel discussion for free bosons in which the bilinear
form~$b$ is symplectic and the Clifford algebra is replaced by the
Heisenberg-Weyl algebra.  See~\cite{Mo} for further details.
  \end{remark}

  \subsection*{The Dirac-Nambu space}

We continue with the free fermion system based on~$\Mb$ with symmetry
group~$(G,\phi )$.  Define the \emph{Dirac-Nambu space}
  \begin{equation}\label{eq:194}
     \HDN := \sM\otimes \CC
  \end{equation}
While $\HDN$~is a complex vector space, we do not (yet) endow it with a
hermitian metric.  Define
  \begin{equation}\label{eq:210}
     \rho \:G\longrightarrow \Aut_{\RR}(\HDN) 
  \end{equation}
by letting $\rho (g)$~be the complex linear or antilinear extension
of~$\beta (g)$ according as $\phi (g)=+1$ or~$-1$.  In symbols,
  \begin{equation}\label{eq:211}
     \rho (g)\bigl(\xi \otimes (x+\sqo\,y)\bigr) = \beta (g)(\xi )\otimes
     (x+\phi (g)\sqo\,y),\qquad g\in G,\quad \xi \in \sM,\quad x,y\in \RR.
  \end{equation}

We next choose a complex structure~$I$ on~$\sM$, which in geometric
quantization plays the role of a complex polarization.  Recall that $I$~is a
real isometry $I\:\sM\to\sM$ which satisfies~$I^2=-\id_{\sM}$.  We require a
compatibility condition with the symmetry group~$G$, namely, $I\in \fp$.
This means in particular that each~$g\in G$ either preserves or reverses~$I$.
Such a complex structure may or may not exist.  Assuming it does, extend~$I$
to a complex linear endomorphism of~$\HDN$.  Then $\sqo\,I\:\HDN\to\HDN$
squares to the identity, so induces a decomposition
  \begin{equation}\label{eq:213}
     \HDN \cong V\oplus \overline{V} 
  \end{equation}
into its~$-1$ and~$+1$-eigenspaces.  Here $V$~denotes the complex vector
space~$\sM$ with complex structure~$I$.  Since we assume~$I\in \fp$, the map
$\beta (g)\in \OMb$ commutes or anticommutes with the complex structure~$I$
according as $t(g)=+1$ or~$-1$.  Then since
  \begin{equation}\label{eq:214}
     \rho (g)\sqo\,I = c (g)\sqo\,I\rho (g), 
  \end{equation}
the automorphism~$\rho (g)$ in~\eqref{eq:211} preserves or reverses the
decomposition ~\eqref{eq:213} according as $c(g)=+1$ or~$-1$.  The complex
structure determines a hermitian structure on~$V$:
  \begin{equation}\label{eq:195}
     \langle \xi _1,\xi _2 \rangle := b(\xi _1,\xi _2) + \sqo\,b(I\xi _1,\xi
     _2),\qquad \xi _1,\xi _2\in V.
  \end{equation}
As usual, the hermitian metric induces a linear
isomorphism~$\overline{V}\cong V^*$, a hermitian metric on that space, and so
too a hermitian metric on~$\HDN$.  

We extend Hamiltonians to act linearly and self-adjointly on~$\HDN$.  We do
\emph{not} ask that a Hamiltonian be compatible with the complex structure,
lest we rule out interesting examples such as Bogoliubov-de Gennes
Hamiltonians.  Nonetheless, for an \emph{invertible} Hamiltonian~$H=
\sqo\,A^{ij}e_ie_j\in \sA$ invariant under the action~\eqref{eq:212}
of~$(G,\phi )$, there is a natural associated complex structure~$I_A$.  For in
this case the associated $A\in \oMb$ is also invertible, and we use the
spectral calculus to rescale the eigenvalues and define $I_A=A/|A|$.  Since
$A\in \fp$, it follows that $I_A\in \fp$.

  \begin{remark}[]\label{thm:186}
 The Dirac-Nambu space is \emph{not} an $\sA$-module in the sense of
Remark~\ref{thm:177}, so does not fit the standard paradigm of a realization
of an abstract quantum system.  Rather, given the complex structure~$I$ one
usually forms the \emph{fermionic Fock space}~$\HF={\textstyle\bigwedge}
^{\bullet }V$, which \emph{is} an $\sA$-module.  The degree in the exterior
algebra is the particle number, so $V\subset \HF$ may be identified with the
1-particle subspace: lines in~$V$ represent 1-particle states.  The action
of~$G$ on~$\sA$ is implemented as a \emph{projective} action of~$G$ on~$\HF$.
Namely, the entire orthogonal group~$\OMb$ acts projectively on~$\HF$ as the
(s)pin representation, and there is a pullback central extension of~$G$ along
$\beta \:G\to\OMb$.  Contrast this projective action with the nonprojective
action~\eqref{eq:210} of~$G$ on the Dirac-Nambu space~$\HDN$.  The symmetries
in~$G$ are very special among orthogonal symmetries as they either preserve
or reverse the complex structure~$I$.  
  \end{remark}

Despite the considerations in Remark~\ref{thm:186}, we use the Dirac-Nambu
space~$\HDN$ as the Hilbert space of a free fermion system.  In the old Dirac
picture of particles and holes, lines in~$\overline{V}\subset \HDN$ represent
``1-hole'' states.  Thus we regard~$\HDN$ as implementing this particle-hole
picture, special to free fermions.  From this point of view elements~$g\in G$
with~$c(g)=+1$ are \emph{particle-hole preserving} whereas elements~$g\in G$
with~$c(g)=-1$ are \emph{particle-hole reversing}.

   \section{Topological phases}\label{sec:18}

Our interest is in the topological properties of gapped systems, so we need
to specify a notion of topological equivalence.

  \begin{definition}\label{thm:26}
 Gapped systems with extended \Qsymmetry class~$(G,\phi ,\tau ,c)$ are in the
same \emph{topological phase} if they are homotopic.  We denote the set of
equivalence classes by~$\TP(G,\phi ,\tau ,c)$.
  \end{definition}

\noindent
 This definition requires explanation and perhaps motivation.  First, the
explanations.  An isomorphism of gapped systems $(\sH_0,H_0,\tar _0)$ and
$(\sH_1,H_1,\tar _1)$ is an isomorphism $\sH_0\to\sH_1$ which intertwines the
Hamiltonians~$H_i$ and the representations~$\tar _i$, $i=0,1$.  Next, two
systems are homotopic if there exists a continuously varying family of gapped
systems parametrized by the 1-simplex $\Delta ^1=[0,1]$ whose restriction
to~$\{0\}$ is isomorphic to~$(\sH_0,H_0,\tar _0)$ and whose restriction
to~$\{1\}$ is isomorphic to $(\sH_1,H_1,\tar _1)$.  In Appendix~\ref{sec:14}
we give a precise definition of a continuous family of quantum systems; in
particular Definition~\ref{thm:113} defines a continuous family of gapped
systems with given extended \Qsymmetry class.  Isomorphism and homotopy are
completely standard equivalence relations in many contexts in mathematics
and, increasingly, in physics.  The key is to define the notion of a
\emph{continuous} deformation---a slightly technical, but nonetheless
crucial, task undertaken in Appendix~\ref{sec:14}.
 
The set $\TP(G,\phi ,\tau ,c)$ has an algebraic structure given by
amalgamation of quantum systems.  Usually in quantum mechanics if
$\sH_1,\sH_2$~are Hilbert spaces of two systems, then the Hilbert space of
the composite quantum system is the (graded) \emph{tensor
product}~$\sH_1\otimes \sH_2$.  This applies to the Fock spaces of free
electron systems.  Since the tensor product of Fock spaces corresponds to the
\emph{direct sum} of \emph{Dirac-Nambu} Hilbert spaces~\eqref{eq:194}, and we
have in mind the Dirac-Nambu picture, we use direct sum to combine
systems. Direct sum is compatible with the notion of topological equivalence
in Definition~\ref{thm:26}, and so $\TP(G,\phi ,\tau ,c)$~ is a commutative
\emph{monoid}, that is, a set with an associative, commutative composition
law with an identity element.  The composition is written additively and the
identity element is represented by the zero Hilbert space.

  \begin{remark}[]\label{thm:120}
 We emphasize that in this paper we do not use any detailed structure of free
systems of fermions.  This use of direct sum rather than tensor product to
combine quantum systems is the concrete manifestation of the Dirac-Nambu
picture; in a general system only tensor products exist.
  \end{remark}

Definition~\ref{thm:26} is quite general and seems to be what the physics
demands.  But it may be difficult to compute, and so we define a weaker
invariant which is easier to compute and which is in addition an abelian
group.  We obtain it from the commutative monoid~$\TP(G,\phi ,\tau ,c)$ by
group completion or by imposing a further ``topological triviality''
relation: see Definition~\ref{thm:102}, Definition~\ref{thm:132}, and
Definition~\ref{thm:135}.  We consider each of these in turn.

  \subsection*{Group completion}

  \begin{construction}[Grothendieck group completion]\label{thm:68}
 Let $M$~be a commutative monoid.  Its \emph{group completion} is the
quotient~$A$ of~$M\times M$ by the equivalence relation
  \begin{equation}\label{eq:77}
     (m'_1,m''_1)\sim(m'_2,m''_2) \quad \textnormal{if\quad
     $m'_1+m''_2=m''_1+m'_2$.} 
  \end{equation}
Then $A$~is an abelian group.  There is a homomorphism $M\to A$ given by
$m\mapsto(m,0)$.  The group completion~$A$ of~$M$ satisfies a universal
property, which characterizes it up to unique isomorphism: if $B$~is an
abelian group and $\theta \:M\to B$ a homomorphism of commutative monoids in
the sense that $\theta (0)=0$ and $\theta (m_1+m_2)=\theta (m_1)+\theta
(m_2)$ for all~$m_1,m_2\in M$, then there is a unique homomorphism of abelian
groups $\bar\theta \:A\to B$ such that the diagram
  \begin{equation}\label{eq:160}
     \xymatrix{ M\ar[rr]^\theta\ar[dr] &&B \\ &
      A\ar@{-->}_{\bar\theta }[ur]} 
  \end{equation}
commutes.
  \end{construction}

\noindent
 Intuitively, $(m',m'')$ represents~$m'-m''$.   

  \begin{example}[]\label{thm:122}
 There are two elementary, but illustrative, examples.  First, the group
completion of the commutative monoid~$M=\ZZ^{\ge0}$ of nonnegative integers
is the abelian group of all integers.  Note in this case that $\ZZ^{\ge0}$~is
also a semi-ring---there is an associative multiplication with identity which
distributes over addition---and so the group completion~$\ZZ$ is a ring.
Next, suppose $G$~is a compact Lie group.  Let $\Rep(G)$ denote the set of
equivalence classes of \emph{finite dimensional} complex representations
of~$G$.  Weyl's ``unitary trick'' implies that every representation is
completely reducible---a direct sum of irreducible representations---so
$\Rep(G)$~can also be defined as the free commutative monoid generated by
isomorphism classes of irreducible complex representations of~$G$.  The group
completion of~$\Rep(G)$ is the \emph{representation ring}~$K_G$ of~$G$.
  \end{example}

  \begin{example}[]\label{thm:60}
 For $G=SU_2$ there is an irreducible representation for each nonnegative
integer, which labels its dimension, so as a commutative monoid $\Rep(G)$~is
the free commutative monoid on~$\ZZ^{\ge0}$ whose elements are finite linear
combinations of symbols~$\bol0,\bol1,\bol2,\dots $ with positive integer
coefficients.  The structure of the \emph{ring}~$K_G$ is much simpler: it is
the polynomial ring~$\ZZ[t]$ on a single generator~$t=\bol2$, which
represents the 2-dimensional irreducible representation.  By Clebsch-Gordon
other irreducible representations are polynomials in~$t$.  For example,
$\bol3=t^2-1$.  The representation ring of a \emph{connected, simply
connected} compact Lie group is always a polynomial ring, but for a general
compact Lie group there are relations.  For example, the representation ring
of the cyclic group of order~2 is isomorphic to $\ZZ[\epsilon ]/(\epsilon
^2-1)$.
  \end{example}

  \begin{remark}[]\label{thm:165}
 The multiplication in the representation ring~$K_G$ does not play a role in
our considerations.
  \end{remark}

  \begin{remark}[]\label{thm:125}
 We illustrate why a finiteness hypothesis is crucial to obtain a nontrivial
abelian group after completing.  Resuming Example~\ref{thm:122} consider the
trivial group~$G=\{1\}$, so $\Rep(G)\cong \ZZ^{\ge0}$ is the set of
isomorphism classes of finite dimensional vector spaces and the group
completion is the integers.  If we relax the finite dimensionality, so allow
infinite dimensional vector spaces, then $\ZZ^{\ge0}$~is replaced by the
commutative monoid $\ZZ^{\ge0}\cup\{\infty \}$.  In this monoid $n+\infty =\infty
$ for any nonnegative integer~$n$.  The group completion of
$\ZZ^{\ge0}\cup\{\infty \}$ is the trivial abelian group.  This is because
for any pair of nonnegative integers~$(n',n'')$, we have
  \begin{equation}\label{eq:127}
     (n',n'')\sim(\infty ,\infty )\sim(0,0) 
  \end{equation}
according to the equivalence relation~\eqref{eq:77} which defines the group
completion.  The maneuver~\eqref{eq:127} is called the \emph{Eilenberg
swindle}.  We use a more sophisticated variation in the proof of
Theorem~\ref{thm:136}. 
  \end{remark}

  \subsection*{Quotienting by topological triviality}

There is another construction of the representation ring which illustrates
the idea of ``topological triviality''.  We start with the set~$\Rep_s(G)$ of
\emph{$\zt$-graded } finite dimensional complex representations up to
isomorphism.  (The `s'~stands for `super', which is a synonym for
`$\zt$-graded'.)  A $\zt$-graded representation is a $\zt$-graded complex
vector space~$W=W^0\oplus W^1$ with a homomorphism $\rho \:G\to\Aut(W)$ to
automorphisms which separately preserve~$W^0$ and~$W^1$.  So it is the
\emph{direct sum} of two representations, but is to be thought of as the
\emph{formal difference} of the representations.  The role of the grading is
to encode the formal minus sign.  Direct sum of super representations
endows~$\Rep_s(G)$ with the structure of a commutative monoid.  Then define
the submonoid~$\Triv_s(G)$ to be equivalence classes of $\zt$-graded
representations~$\rho \:G\to\Aut(W)$ for which there exists an odd
automorphism $P\:W\to W$ such that
  \begin{equation}\label{eq:126}
     P\rho (g)=\rho (g)P,\qquad g\in G. 
  \end{equation}

  \begin{lemma}[]\label{thm:170}
  The quotient ~$\Rep_s(G)/\Triv_s(G)$ is an abelian group. 
  \end{lemma}

\noindent
 The additive inverse of~$W^0\oplus W^1$ is the parity-reversed
representation $W^1\oplus W^0$, which we define more precisely in the proof.
 
  \begin{proof}
 Let $\Pi $~denote the complex line~$\CC$ in odd degree and $\pi \in \Pi $
the (odd) element~$1\in \CC$.  There is a canonical isomorphism $\Pi \otimes
\Pi \cong \CC$ under which $\pi \otimes \pi =1$.  If $W$~is a finite
dimensional representation of~$G$, we define its parity-reversal to be the
tensor product $\Pi \otimes W$, with $G$~acting trivially on~$\Pi $.  Then
$W\,\oplus \,(\Pi \otimes W)$ lies in~$\Triv_s(G)$: for $x\in W\,\oplus
\,(\Pi \otimes W)$ set $P(x) = \pi\otimes x$.  Thus $\Pi \otimes W$ is an
inverse to~$W$ in~$\Rep_s(G)/\Triv_s(G)$.
  \end{proof}

The map $W^0\oplus W^1\mapsto (W^0,W^1)$ from $\zt$-graded representations to
the group completion of the monoid of equivalence classes of ungraded
representations is an isomorphism, after modding out by~$\Triv_s(G)$.  In
other words, there is an isomorphism 
  \begin{equation}\label{eq:218}
     \Rep_s(G)/\Triv_s(G)\cong K_G. 
  \end{equation}

  \subsection*{Reduced topological phases}

We call the abelian group obtained from the commutative monoid in
Definition~\ref{thm:26} by these processes---group completion or quotienting
by ``topological triviality''---the group of \emph{reduced topological
phases}~$\GH(G,\phi ,\tau ,c)$.  The choice of process depends on the
particular situation and will be specified in each case.  Our main results
identify~$\GS(G,\phi ,\tau ,c)$ with a topological $K$-theory group in three
different situations: see Theorem~\ref{thm:40}, Theorem~\ref{thm:131}, and
Theorem~\ref{thm:136}.

One immediate observation points the way to $K$-theory. 

  \begin{lemma}[]\label{thm:27}
 In every topological equivalence class of gapped systems there is a
Hamiltonian~$H$ with $H^2=1$. 
  \end{lemma}

  \begin{proof}
 Let $(\sH,H_0,\tar )$ be a gapped system.  Define the 1-parameter family of
functions
  \begin{equation}\label{eq:63}
     \begin{aligned} f_t\:\RR\setminus \{0\}&\longrightarrow \RR\setminus
      \{0\} \\ \lambda &\longmapsto (1-t)\lambda +t\lambda /|\lambda
      |\end{aligned} 
  \end{equation}
where $t\in [0,1]$.  The spectral theorem gives a continuous 1-parameter
family of Hamiltonians~$H_t=f_t(H_0)$ and $H=H_1$ squares to the identity.
  \end{proof}

\noindent
 $H$~is sometimes called the \emph{spectral flattening} of~$H_0$.  It induces
a $\zt$-grading on~$\sH$ whose homogeneous subspaces are the
$\pm1$-eigenspaces of~$H$.  The $\zt$-grading is a hallmark of $K$-theory,
but to obtain nontrivial results we must impose finiteness conditions.  We
will specify these in~\S\ref{sec:11} and~ \S\ref{sec:6}; see
Hypothesis~\ref{thm:38}, Hypothesis~\ref{thm:39}, and
Hypothesis~\ref{thm:56}.

   \section{Special extended \Qsymmetry classes}\label{sec:4}

  \subsection*{A restricted set of extended \Qsymmetry classes}

Let $\sC=\pmo\times \pmo$ and
  \begin{equation}\label{eq:58}
     \sC=\pmo\times \pmo\xrightarrow{\;\;\fC\;\;}\pmo 
  \end{equation}
be multiplication.  The element~$\bT=(-1,1)$ represents time-reversal and
$\bC=(1,-1)$ represents Hamiltonian-reversal.

  \begin{definition}[]\label{thm:23}
 Let $G$~be a Lie group equipped with a homomorphism $\psi=(t,c)
\:G\to\sC$.  An extended \Qsymmetry class~$(G,\phi ,\tau,c)$ is \emph{$\psi
$-standard} if $\phi =\fC\circ \psi =tc$ and $\tG$~is pulled back from a
$\fC$-twisted extension of $\psi (G)\subset \sC$.
  \end{definition}

\noindent
 Thus $\psi (g)=\bigl(t(g),c(g) \bigr)$ tells whether $g\in G$ is implemented
as time-reversing and/or Hamiltonian-reversing, and $g$~is implemented
unitarily or antiunitarily according to~$\phi (g)$: it is antiunitary if it
is either time-reversing or Hamiltonian-reversing but not both, and it is
unitary if it is neither or both.  Furthermore, in a $\ptc$-twisted
representation (Definition~\ref{thm:84}(ii)) of~$G$ the subgroup of~$\tG$
lying over $\ker\psi \subset G$ is split---there is no nontrivial central
extension---but the hypothesis that the extension is pulled back is stronger.

  \begin{remark}[]\label{thm:35}
 The notion of a $\psi $-standard extended \Qsymmetry class is
artificial\footnote{There is a physical viewpoint which makes the assumption
of $\phi$-standard and $\psi$-standard symmetries somewhat more natural. In
systems with low symmetry---such as systems with disorder---it might be that
the \emph{only} symmetries present are a time-reversal symmetry and---in the
case of free fermions--a Hamiltonian-reversing symmetry.} and only given here
to reproduce the 10-fold way which permeates the literature on topological
classes of free fermion systems.  While the restriction $\phi =\fC\circ \psi
=tc$ is physically reasonable, we do not see a rationale to require that the
twisted central extension be a pullback.  So in general there are more
extended \Qsymmetry classes which are not pullbacks.  We give an example
below in Remark~\ref{thm:79}.
  \end{remark}

Returning now to free fermion systems, it is clear from
Definition~\ref{thm:23} that a $\psi $-standard extended \Qsymmetry class is
determined by a subgroup~$A\subset \sC$ and an extended \Qsymmetry class
based on~$A$ and the restriction of~$\fC$ to~$A$.  There are 5~subgroups and
a total of 10~possibilities.

  \begin{proposition}[10-fold way]\label{thm:24}
 \ 
  \begin{enumerate}
 \item There is a unique $\psi$-standard extended \Qsymmetry class based on the trivial
subgroup~$A=1$.  

 \item There is a unique $\psi$-standard extended \Qsymmetry class based on the diagonal
subgroup $A=\pmo\subset \pmo\times \pmo$.

 \item There are two $\psi$-standard extended \Qsymmetry classes based on
$A=\pmo\times 1$.  Let $\tT\in \tA$ be a lift of~$\bT$; then the two
\Qsymmetry classes are distinguished by~$\tT^2=\pm1$.

 \item Similarly, there are two \Qsymmetry classes based on $A=1\times \pmo$,
distinguished by~$\tC^2=\pm1$, where $C\in \tA$ is a lift of~$\bC$.

 \item There are four $\psi$-standard extended \Qsymmetry classes based on
$A=\sC$, distinguished by $\tT^2=\pm1$ and $\tC^2=\pm1$.
  \end{enumerate}

  \end{proposition}

\noindent
 The presence or non-presence of~$\bT,\bC,\bT\bC$ and the squares
of~$\tT,\tC$ are the standard invariants of the 10-fold way.  We term the
pair ~$(A,\tA )$ of the subgroup and its central extension a `\CTt'.
Proposition~\ref{thm:34} and Proposition~\ref{thm:24} have quick and easy
proofs via group cohomology computations, once the problem of group
extensions is translated to a cohomology question.\footnote{Untwisted central
extensions of~$A$ up to isomorphism form an abelian group which is identified
with the second group cohomology~$H^2(A;\TT)$ with $A$~acting trivially on
the coefficients.  More generally, a homomorphism $\fC\:A\to\pmo$ determines
an action of~$A$ on~$\TT$ (if $\fC(a)=-1$, then $a$~acts by $\lambda
\mapsto\lambda \inv $), and $\fC$-twisted extensions up to isomorphism are
identified with~$H^2(A;\widetilde\TT)$ with twisted coefficients.  These
cohomology groups are easy to compute for the small abelian groups~$A$ in
Proposition~\ref{thm:24}.}  Instead we adopt a more direct, though
unnecessarily lengthy, approach.

  \begin{remark}[]\label{thm:79}
 As we have said, we find the $\psi$-standard and $\phi $-standard
conditions---which are generally assumed in the literature---to be
artificial. It is certainly easy to give physically reasonable examples where
they fail. For a non $\phi $-standard example consider $G=\pmo \times \pmo$
with $\phi \: G \to \pmo$ defined as the product.  Then we know from
Proposition~\ref{thm:24}(v) that there are 4 non-isomorphic $\phi $-twisted
 extensions $G^\tau$ of $G$, but there can be only two nontrivial pullbacks
of $\phi$-twisted extensions of $\pmo$ with $\phi =\id$. Concretely, if in
$G^\tau$ we have $T^2 = C^2 = 1$ or $T^2=C^2=-1$ then the extension is a
pullback and hence is $\phi $-standard, but if $T^2 = 1$, $C^2=-1$ or $T^2 =
-1 $, $C^2=+1$ then the extension is not a pullback.  Similarly, one can make
non $\psi$-standard extended \Qsymmetry classes by considering $G = \pmo
\times \pmo \times \pmo$ with $\psi\: G \to \sC$ given by multiplying the
first two and the last two factors of $G$. There are $8$ $\phi$-twisted
extensions of $G$ (defined by the signs of the squares of the lifts of the
three generators of $G$) but there are only $4$ pullbacks of extensions of~
$\sC$.
  \end{remark}

  \subsection*{Twisted central extensions and semidirect products}

As a preliminary we develop some facts about twisted extensions of semidirect
products, despite the fact that to prove Proposition~\ref{thm:34} and
Proposition~\ref{thm:24} we only need to consider direct products; we will
encounter semidirect products later.  Here we show how to build up
complicated $\phi$-twisted extensions out of simpler extensions.

Let $G',G''$ be topological groups and $\alpha \:G''\to\Aut(G')$ a continuous
homomorphism.  Recall the semidirect product $G=G''\ltimes _\alpha G'$
defined in Definition~\ref{thm:87}.  Now suppose
  \begin{equation}\label{eq:23}
     1\longrightarrow \TT\longrightarrow H'\xrightarrow{\;\;\pi
     '\;\;}G'\longrightarrow 1 
  \end{equation}
is a \emph{central} extension, $\phi \:G''\to\pmo$ is a homomorphism and  
  \begin{equation}\label{eq:24}
     1\longrightarrow \TT\longrightarrow H''\xrightarrow{\;\;\pi
     ''\;\;}G''\longrightarrow 1 
  \end{equation}
a $\phi$ -twisted  extension, and 
  \begin{equation}\label{eq:25}
     \ta\:G''\longrightarrow \Aut(H') 
  \end{equation}
is a continuous homomorphism such that for~$g''\in G''$ and~$\lambda \in
\TT\subset H'$,  
  \begin{equation}\label{eq:26}
     \ta(g'')(\lambda )=\begin{cases} \lambda ,&\phi (g'')=+1;\\\bar\lambda
     ,&\phi (g'')=-1;\end{cases} 
  \end{equation}
and for~$g''\in G''$, the diagram
  \begin{equation}\label{eq:27}
 \xymatrix{H'\ar[r]^{\pi '}\ar[d]_{\ta(g'')}&G'\ar[d]^{\alpha (g'')}\\
H'\ar[r]^{\pi '}&G'} 
  \end{equation}
commutes.  Let $\phi _G\:G\to\pmo$ be the composition $G\to
G''\xrightarrow{\phi }\pmo$.

  \begin{construction}[]\label{thm:8}
 The data \eqref{eq:23}, \eqref{eq:24}, \eqref{eq:25} determine a $\phi
_G$-twisted  extension 
  \begin{equation}\label{eq:28}
     1\longrightarrow \TT\longrightarrow \tG\xrightarrow{\;\;\pi
     \;\;}G\longrightarrow 1 
  \end{equation}
which restricts on~$G'\subset G$ to~\eqref{eq:23} and on~$G''\subset G$
to~\eqref{eq:24}.   
  \end{construction}

  \begin{proof}
 We use a variation of the semidirect product construction~\eqref{eq:22}.
Namely, define the set underlying~$\tG$ as the quotient of $H''\times H'$ by
the diagonal embedding of~$\TT$: identify $(h'',h')\sim(\lambda h'',\lambda
h')$ for all~$h''\in H''$, $h'\in H'$,~$\lambda\in \TT$.  Multiplication is
defined so that $H''\subset \tG$~and $H'\subset \tG$~are subgroups, and
  \begin{equation}\label{eq:29}
     h''\cdot h'=\ta(\pi ''h'')(h')\cdot h'',\qquad h''\in H'',\quad h'\in
     H'. 
  \end{equation}
A routine check shows that \eqref{eq:29}~is well-defined and descends to the
quotient by the diagonal~$\TT$. 
  \end{proof}

  \begin{lemma}[]\label{thm:9}
 Let $G=G''\ltimes _\alpha G'$ be a semidirect product.  Suppose
$\fG\:G\to\pmo$ is a homomorphism which is pulled back from $\phi
\:G''\to\pmo$.  Then any $\fG$-twisted  extension~\eqref{eq:20} is
obtained via Construction~\ref{thm:8}.
  \end{lemma}

  \begin{proof}
 Given a $\fG$-twisted  extension~\eqref{eq:20}, define the central
extension ~\eqref{eq:23} of~$G'$ as the restriction of~\eqref{eq:20}
over~$G'\subset G$ and the $\phi$ -twisted  extension~ \eqref{eq:24}
of~$G''$ as the restriction of~\eqref{eq:20} over~$G''\subset G$.  For
$h''\in H''\subset \tG$ and $h'\in H'\subset \tG$ define
  \begin{equation}\label{eq:31}
     \ta(h'')(h') = (h'')h'(h'')\inv . 
  \end{equation}
Straightforward checks show: (i)\ $\ta$~depends only on~$\pi (h'')$; (ii)\
$\ta(h'')$ is an automorphism of~$H'$ which satisfies~\eqref{eq:26}
and~\eqref{eq:27}; and (iii)\ $\ta \:H''\to\Aut(H')$ is a homomorphism.
  \end{proof}

  \begin{lemma}[]\label{thm:10}
  Let $\fG\:G\to\pmo$ be a homomorphism which is pulled back from $\phi
\:G''\to\pmo$ and fix a central extension~\eqref{eq:23} of $G'\subset G$.
Then isomorphism classes of $\fG$-twisted  extensions of~$G$ which
restrict on~$G'$ to~\eqref{eq:23} are in 1:1~correspondence with
pairs~$(\ta,\epsilon '')$, where $\ta$~is a continuous homomorphism as
in~\eqref{eq:25} and $\epsilon ''$~is the isomorphism class of a $\phi
$-twisted  extension as in~\eqref{eq:24}.
  \end{lemma}

\noindent
 An isomorphism $\varphi \:G^{\tau _1}\to G^{\tau _2}$ of $\phi _G$-twisted
 extensions---as in~\eqref{eq:19}---is here required to be the
identity map on the common subgroup~$H'$ of~$G^{\tau _1}$ and~$G^{\tau _2}$.

  \begin{proof}
 To see that $\ta$~is an invariant, apply an isomorphism $\varphi \:G^{\tau
_1}\to G^{\tau _2}$ to~\eqref{eq:31}.  Restrict~$\varphi $ over~$G''$ to see
that the isomorphism class of the restriction of~$\tG$ over~$G''$ is an
invariant.  Conversely, by Construction~\ref{thm:8} every pair~$(\ta,\epsilon
'')$ gives a $\fG$-twisted  extension. 
  \end{proof}

  \subsection*{Proof of the classification}

Proposition~\ref{thm:34} and Proposition~\ref{thm:24} follow immediately from
the next result.

  \begin{lemma}[]\label{thm:11}
  \ \begin{enumerate}
 \item Every central extension 
  \begin{equation}\label{eq:32}
     1\longrightarrow \TT\longrightarrow H\longrightarrow \pmo\longrightarrow
     1 
  \end{equation}
is split, that is, is isomorphic to the trivial central extension. 

 \item Let $\phi \:\pmo\to\pmo$ be the identity map.  Then there are two
isomorphism classes of $\phi $-twisted  extensions~\eqref{eq:32}
distinguished by whether the lift of~$-1\in \pmo$ to~$H$ squares to~$+1$
or~$-1$.  

 \item Let $G''=G'=\pmo$ with generators~$g'',g'$ and let $\phi \:G''\times
G'\to\pmo$ be projection onto the first factor.  Then there are four
isomorphism classes of $\phi $-twisted  extensions
  \begin{equation}\label{eq:33}
     1\longrightarrow \TT\longrightarrow H\longrightarrow G''\times
     G'\longrightarrow 1 
  \end{equation}
Choose a lift~$h'$ of~$g'$ so that $(h')^2=+1$, and let $h''$~be any lift
of~$g''$.  Then the four isomorphism classes of extensions are distinguished
by the independent possibilities $(h'')^2=\pm1$ and~$(h'h'')^2 =\pm1$.

  \end{enumerate}
  \end{lemma}

  \begin{proof}
 For~(i), if $h$~is any lift of~$-1\in \pmo$, then $h^2=\mu \in \TT$ for
some~$\mu $, so replacing~$h$ by~$\lambda h$ for $\lambda ^2=\mu \inv $ we
have~$(\lambda h)^2=+1$.  This splits~\eqref{eq:32}. 
 
For~(ii), if $h$~is any lift of~$-1\in \pmo$, and $h^2=\mu \in \TT$, then
$(\lambda h)^2=\lambda \bar\lambda h^2=\mu $ for all~$\lambda \in \TT$.
Since $h=(h^2)h(h^2)\inv =\mu h\mu \inv =\mu ^2h$, we conclude~$\mu ^2=+1$
and so $h^2=\pm1$.
 
For~(iii), observe first that the restriction~$H'\subset H$ of the extension
to~$G'\subset G'\times G''$ is split, by~(i), and since $\phi $~is trivial
on~$G'$ we have $H'\cong \pmo\times \TT$.  Now apply Lemma~\ref{thm:10}.
By~(ii) there are two possibilities for the isomorphism class~$\epsilon ''$
of the extension over~$G''$.  There are also two homomorphisms
$\ta\:G''\to\Aut(H')$ which satisfy~\eqref{eq:26} and~\eqref{eq:27}: namely,
$\beta =\ta(-1)\in \Aut(H')$ is complex conjugation on~$\TT\subset H'$ and
either $\beta (h')=h'$ or~$\beta (h')=-h'$, where $-h'$~is the product
of~$-1\in \TT\subset H'$ and~$h'\in H'$.
  \end{proof}

  \subsection*{More \Qsymmetry classes}

As a further illustration of these techniques we prove that certain twisted
extensions, though not $\psi $-standard (Definition~\ref{thm:23}), are
nonetheless uniquely determined by simple data.  In the following examples we
incorporate the mass central extension of the Galilean group, which is
physically relevant, but which does not fit into the 10-fold scheme of
Proposition~\ref{thm:24}.  Let $(M,\Gamma )$ be a Galilean spacetime.  Set
  \begin{equation}\label{eq:21}
     \begin{aligned} G_1 &= \textnormal{identity component of }\Aut(M,\Gamma
     )\\ G_2 &= G_1\times 
      \pmo \\ G_3 &= {\Aut}^+(M,\Gamma ) \\ G_4&=G_3\times
      \pmo\end{aligned} 
  \end{equation}
Here $G_3$~is the group of parity-preserving symmetries of~$(M,\Gamma )$; see
Definition~\ref{thm:30}.  Notice that $G_1$~is \emph{not} the Galilean group;
the Galilean group is a double cover of~$G_1$.  Each of these groups is
naturally equipped with a homomorphism~$t$ to~$\pmo$ which tracks
time-reversal: it is trivial for~$G_1,G_2$ and surjective for~$G_3,G_4$.  Let
the homomorphism~ $c$ to~$\pmo$ to be trivial on~$G_1,G_3$ and either trivial
or projection onto the second factor on~$G_2,G_4$.  (These two cases
correspond to the two possibilities contemplated in Lemma~\ref{thm:12}
below.)  Thus for each~$G_i$ we have a homomorphism $\psi =(t,c)\:G_i\to\sC$.
Recall that there is a universal mass central extension 
  \begin{equation}\label{eq:37}
     1\longrightarrow \TT\longrightarrow G_1^\tau \longrightarrow
     G_1\longrightarrow 1 
  \end{equation}
of the identity component of~$\Aut(M,\Gamma )$.\footnote{It may be more
natural to take the kernel of~\eqref{eq:37} to be the universal cover~$\RR$
of~$\TT$, but an action $G_1\to\QAut(\PH)$ induces a central extension
by~$\TT$, not~$\RR$.}

  \begin{proposition}[]\label{thm:36}
 Let $(G,\phi ,\tau)$ be a \Qsymmetry class with $G=G_i$ for some~$i=1,2,3,4$
which satisfies (i)~$\phi =\fC\circ \psi $ and (ii)~$\tG$ restricts to the
mass central extension~\eqref{eq:37} over~$G_1\subset G$.  Then the \Qsymmetry
class is determined by a twisted  extension of the quotient~$G/G_1$.
  \end{proposition}

  \begin{remark}[]\label{thm:37}
 The hypothesis~(ii) is a bit unnatural.  For example, we should surely
replace~$G_1$ by its Galilean double cover, in which case there are more
possibilities if $d=0,1$.  Also, we could have a larger group of symmetries
which commutes with the geometric spacetime symmetries, and that group may
have nontrivial extensions. 
  \end{remark}

As each of the groups ~$G=G_1,G_2,G_3,G_4$ is a semidirect product of~$G_1$
and a finite group isomorphic to~$G/G_1$, and the extension on the
subgroup~$G_1$ has been fixed, we can apply Lemma~\ref{thm:9} to conclude
that Proposition~\ref{thm:36} follows if we prove that the homomorphism~$\ta$
in~\eqref{eq:25} is unique.  Recall from~\eqref{eq:21} that $G_1$~is the
component of the identity of symmetries of a Galilean spacetime~$(M,\Gamma
)$.  Let
  \begin{equation}\label{eq:34}
     1\longrightarrow \TT\longrightarrow \tGo\xrightarrow{\;\;\pi
     _1\;\;}G_1\longrightarrow 1 
  \end{equation}
be the mass central extension. 

  \begin{lemma}[]\label{thm:12}
 
  \ \begin{enumerate}
 \item The constant map with value~$\idtGo$ is the unique homomorphism
  \begin{equation}\label{eq:35}
     \ta\:\pmo\longrightarrow \Aut(\tGo) 
  \end{equation}
such that $\ta(-1)(\lambda )=\lambda $ for~$\lambda \in \TT$ and
$\ta(-1)$~covers the identity automorphism of~$G_1$.

 \item There is a unique homomorphism~\eqref{eq:35} such that
$\ta(-1)(\lambda )=\bar\lambda $ for~$\lambda \in \TT$ and $\ta(-1)$~covers
the automorphism of~$G_1$ obtained by conjugation by a time-reversal symmetry.

  \end{enumerate}
  \end{lemma}

  \begin{proof}
 Set $\beta =\ta(-1)$.  Then for~$h\in \tGo$ with $\pi _1h=g$ we have $\beta
(h)=\mu (g)h$ for some character $\mu \:G_1\to\TT$.  Since $G_1$~is
connected, $\mu $~is determined by its infinitesimal character
$\dm\:\mathfrak{g}_1\to\RR$.  Following the notation in~\S\ref{sec:1} we
write the Lie algebra~$\go$ of~$G_1$ as an extension
  \begin{equation}\label{eq:36}
     0\longrightarrow \mathfrak{o}(V)\ltimes(\VpV)\longrightarrow
     \go\longrightarrow W/V\longrightarrow 0,
  \end{equation}
where the orthogonal algebra~$\mathfrak{o}(V)$ acts diagonally by the
standard representation on~$\VpV$ (infinitesimal boosts plus infinitesimal
spatial translations).\footnote{The quotient~$W/V$ could be written
$W/V_{\textnormal{trans}}$ in the current notation.}  Let~ $d=\dim V$.
If~$d\ge3$, then $\mathfrak{o}(V)$~is semisimple and commutators in
$\go'=\oV\ltimes(\VpV)$ span~$\go'$, whence the infinitesimal character~$\dm$
vanishes on~$\go'$.  Therefore, $\dm$~ factors through a linear
map~$W/V\to\RR$.  Exponentiating~\eqref{eq:36} we obtain a presentation
of~$G_1$ as a group extension whose kernel is the normal subgroup of~$G_1$
consisting of automorphisms of~$(M,\Gamma )$ which preserve the leaves of the
simultaneity foliation~$\sS$.  From the Lie algebra argument we conclude that
$\mu $~ factors through the quotient group, which is isomorphic to the
translation group~$\RR$.  But $\beta \circ \beta =\id$ and only the trivial
character of~$\RR$ squares to the trivial character.  This proves~$\beta
=\idtGo$.  Now the subgroup~$V\oplus V$ of~$\go$ is an ideal with quotient
isomorphic to~$\mathfrak{o}(V)\oplus W/V$, and if~$d=2$ the latter is
abelian.  The group~$G_1$ then also has characters pulled back from those on
the group~$SO(V)\cong SO(2)$.  There are no nontrivial characters which
square to the identity character, so once more~$\beta =\idtGo$.  For~$d=1$
the group~$G_1$ is the group of translations and boosts, and any character
factors through the quotient by the normal subgroup~$\exp(\VpV)$ of spatial
translations and boosts (similar to the situation for~$d\ge3$).  That
quotient is isomorphic to~$\RR$, so as before we conclude~$\beta =\idtGo$.
Finally, for~$d=0$ the group~$G_1$ consists only of time translations, and
again~$\beta =\idtGo$.  This proves~(i).
 
For~(ii) we first observe that since $\tGo$~is connected, the automorphism
$\beta =\ta(-1)\:\tGo\to\tGo$ is determined by its differential
$\db\:\tgo\to\tgo$ on the Lie algebra.  Now a time-reversal symmetry
determines a splitting $W=U\oplus V$, where $U$~is the $(-1)$-eigenspace of
the induced linear action on~$W$.  Then 
  \begin{equation}\label{eq:38}
     \tgo \cong \RR\cdot c \oplus \oV\ltimes(\VpV)\oplus U, 
  \end{equation}
where $c$~is central; there is a nontrivial commutator 
  \begin{equation}\label{eq:39}
     [\xi ,\eta ]=\langle \xi ,\eta \rangle\,c,\qquad \xi \in
     V_{\textnormal{boost}},\quad \eta \in V_{\textnormal{trans}},
  \end{equation}
between an infinitesimal boost and an infinitesimal spatial translation---the
angle brackets denote the Euclidean inner product on~$V$; and there is a
nonzero commutator between a (unit)\footnote{To write the formulas we
orient~$U$, and this arrow of time determines a sign for the action of
boosts.  Reversing the orientations does not change the Lie algebra.}
infinitesimal time translation~$\tau \in U$ and an infinitesimal boost~$\xi
\in V_{\textnormal{boost}}$: the bracket $[\tau ,\xi ]$ is~$\xi \in
V_{\textnormal{trans}}$.  The infinitesimal time-reversal is the automorphism
of~$\go$ which is the identity on~$\oV$ and~$V_{\textnormal{trans}}$, and is
minus the identity on~$V_{\textnormal{boost}}$ and~$U$.  The unique lift to
an automorphism~$\db$ of~$\tgo$ sends $c\mapsto -c$.  We must check that
$\db$~exponentiates to an automorphism $\beta \:\tGo\to\tGo$ of order two.
First, $\db$~exponentiates to an automorphism of the simply connected cover
of~$\tGo$ which fixes the central element in the spin double cover of~$SO(V)$
and inverts the kernel of the infinite cyclic cover $\RR\to\TT$ of the
centers, whence it drops to an automorphism~$\beta $ of~$\tGo$.  The
differential of~$\beta \circ \beta $ is $\db\circ \db$, which is the identity
on~$\tgo$, whence $\beta \circ \beta $~is the identity on~$\tGo$.
  \end{proof}

   \section{$K$-theory and some twistings}\label{sec:3}

There are several approaches to topological $K$-theory.  It was introduced in
the late 1950s by Atiyah and Hirzebruch~\cite{AH} to formulate a topological
analog of Grothendieck's Riemann-Roch theorem for sheaves.  The basic object
is a complex vector bundle over a topological space~$X$.  Somewhat later
Atiyah and Segal developed an equivariant version~\cite{Se} for a compact Lie
group~$G$ acting on~$X$.  Our main interest is in twistings of $K$-theory and
twisted $K$-theory.  There are many approaches and so many references, for
example \cite{DK,Mu,FHT1,AS,BCMMS,TXL}; the paper~\cite{GSW} includes the
twists related to complex conjugation which we encounter here.  The twistings
we encounter are very special, and we give an appropriately tailored
exposition.  These twistings are a generalization of extended \Qsymmetry
classes of groups (Definition~\ref{thm:84}(i)) to \emph{groupoids}.
Therefore, we begin with the case of groups and define the twisted
$K$-theory, which is a twisted form of the representation ring
(Example~\ref{thm:122}).  The construction in Lemma~\ref{thm:170} is relevant
here.  We then recast extended \Qsymmetry classes in terms of line bundles
and generalize to groupoids.  Twisted $K$-theory in this situation may still
be defined in terms of finite rank vector bundles, which we prove in
Appendix~\ref{sec:19}.  (For more general twistings one must use bundles of
infinite rank equipped with Fredholm operators~\cite{FHT1}.)

  \subsection*{Twisted representation rings}

Let $G$~be a compact Lie group.  Recall the two constructions of the
representation ring~$K_G$ in~\S\ref{sec:18}.  First, if $\Rep(G)$~is the set
of isomorphism classes of finite dimensional complex representations of~$G$,
which is a commutative monoid under direct sum, then its group completion
is~$K_G$.  Second, let $\Rep_s(G)$~be the commutative monoid of isomorphism
classes of $\zt$-graded complex representations of ~$G$.  Define the
submonoid~$\Triv_s(G)$ of $\zt$-graded representations~$E$ which admit an odd
automorphism~$P\:E\to E$ which satisfies~\eqref{eq:126}.  Then $K_G\cong
\Rep_s(G)/\Triv_s(G)$; see~\eqref{eq:218}.

As we explained in~\S\ref{sec:2} quantum mechanics enhances the notion of an
ordinary complex representation of a group in two ways: (i)~there are both
linear and antilinear transformations, and (ii)~a twisted  extension
of the symmetry group acts on the vector space.  In~\S\ref{sec:10} we
encountered a third enhancement: (iii)~the representation space is
$\zt$-graded and there are both even and odd automorphisms.  We formalize the
resulting ``twisted'' version of~$K_G$ as follows.

  \begin{definition}[]\label{thm:61}
 Let $G$~be a compact Lie group and $(G,\phi ,\tau ,c)$ an extended \Qsymmetry class
based on~$G$ (Definition~\ref{thm:84}(i)).
  \begin{enumerate}
 \item Let $\Rep_s^{\ptc}(G)$ denote the monoid of isomorphism classes of
finite dimensional $\ptc$-twisted representations
(Definition~\ref{thm:84}(ii)). 

 \item Let the submonoid~$\Triv_s^{\ptc}(G)$ consist of $\ptc$-twisted
representations $\tar\:\tG\to\Aut_{\RR}(E)$ such that there exists an odd
$\CC$-linear automorphism $P\:E\to E$ which satisfies $P\tar(g)=c(g)\tar(g)P$
for all~$g\in \tG$.

 \item Define $\hK\phi {\tau ,c}_G=\Rep_s^{\ptc}(G)/\Triv_s^{\ptc}(G)$.
  \end{enumerate}
  \end{definition}

The proof of Lemma~\ref{thm:170} carries over \emph{verbatim} to prove the
following. 

  \begin{lemma}[]\label{thm:128}
 The commutative monoid $\hK\phi {\tau ,c}_G$ is a group. 
  \end{lemma}

  \begin{remark}[]\label{thm:67}
 We always use the ``Koszul sign rule'' for tensor products of $\zt$-graded
vector spaces.  Informally, it asserts that whenever commuting two
homogeneous elements~$x,y$ one picks up a minus sign if both $x$~and $y$~are
odd.  See~\cite[Chapter~1]{DM} for the formal implementation as the symmetry
in the symmetric monoidal category of super vector spaces, as well as a
discussion of parity-reversal.
  \end{remark}

  \begin{remark}[]\label{thm:178}
 It would be reasonable to call an extended \Qsymmetry class~$(G,\phi ,\tau
,c)$ a `$\phi $-twisted $\zt$-graded extension': the homomorphism~$c$ encodes
the $\zt$-grading and $\tau $~is the extension.  This point of view is
developed in the next subsection; see especially Remark~\ref{thm:75}. 
  \end{remark}

  \begin{remark}[]\label{thm:62}
 We emphasize that $\hK\phi {\tau ,c}_G$~is an abelian group which does not
have a ring structure in general.  Rather, it is a module over~$\hKp\phi _G$,
the Grothendieck ring of representations of~$G$ with the same
homomorphism~$\phi $ but trivial $\phi $-twisted extension~$\tau $ and
trivial homomorphism~$c$.  Note that the tensor product of linear maps is
linear and the tensor product of antilinear maps is antilinear, whereas the
tensor product of odd maps is even and the tensor product of scalar
multiplication by~$\lambda _1$ and scalar multiplication by~$\lambda _2$ is
scalar multiplication by~$\lambda _1\lambda _2$.  So the tensor product of
two representations with the same~$\phi $ gives another representation with
the same~$\phi $, whereas the homomorphism~$c $ and the central
extension~$\tau $ change under tensor product.  This distinguished role
of~$\phi $ is reflected in the notation.
  \end{remark}

  \subsection*{Extensions as line bundles}

It is convenient (and standard) to recast the $\phi $-twisted extension as a
hermitian line bundle $\tL\to G$ with a ``composition law''.  Namely, the
extension
  \begin{equation}\label{eq:78}
     1\longrightarrow \TT\longrightarrow \tG\longrightarrow  G\longrightarrow 1 
  \end{equation}
is a principal $\TT$-bundle over~$G$, and $\tL$~is the associated hermitian
line bundle $\tL=(\tG\times \CC)/\TT$, where $\lambda \in \TT$ acts on the
right as $(\tg,z)\cdot \lambda =(\tg\lambda ,\lambda \inv z)$ for $\tg\in
\tG$ and~$z\in \CC$.  The hermitian structure is induced from the standard
hermitian structure on~$\CC$, which is preserved by multiplication
by~$\lambda \in \TT$.  Assume first that $\phi \equiv 1$ so that
\eqref{eq:78}~is a central extension.  We leave the reader to use the group
law in~$\tG$ to construct isometries
  \begin{equation}\label{eq:79}
     \lambda _{g_2,g_1}\:\tL_{g_2}\otimes \tL_{g_1}\longrightarrow \tL_{g_2g_1} ,
  \end{equation}
where $\tL_g$~is the fiber of~$\tL\to G$ over~$g\in G$, and to use the
associativity property of~$\tG$ to prove that the diagram
  \begin{equation}\label{eq:80}
     \xymatrix@C=50pt{\tL_{g_3}\otimes \tL_{g_2}\otimes \tL_{g_1}\ar[r]^{\lambda
     _{g_3,g_2}\otimes \id}\ar[d] _{\id\otimes \lambda _{g_2,g_1}}& 
     \tL_{g_3g_2}\otimes \tL_{g_1}\ar[d]^{\lambda _{g_3g_2,g_1}}\\ \tL_{g_3}\otimes
     \tL_{g_2g_1}\ar[r]^{\lambda _{g_3,g_2g_1}}&\tL_{g_3g_2g_1}} 
  \end{equation}
commutes.  This diagram can be turned into a cocycle condition; see
Remark~\ref{thm:124} below.

For a general~$\phi $ we need to bring in complex conjugation.  Recall that
if $W$~is a complex vector space, then the \emph{complex conjugate vector
space}~$\bW$ has the same underlying real vector space as~$W$ but the complex
multiplication is conjugated.  Thus $\ell \in W$ is the same element $\bl\in
\bW$ in the set~$W=\bW$, and if~$\lambda \in \CC$ then $\bar\lambda \cdot
\bl=\overline{\lambda \cdot \ell } $.  (The first scalar multiplication is in
the complex vector space~$\bW$, the second in~$W$.)  For $\phi \in \pmo$
define
  \begin{equation}\label{eq:81}
     {}^\phi W =\begin{cases} W,&\phi =+1;\\\overline{W},&\phi
      =-1.\end{cases} 
  \end{equation}
An antilinear map $W'\to W$ between complex vector spaces is equivalently a
linear map $\overline{W'}\to W$.  For arbitrary~$\phi \:G\to\pmo$
replace~\eqref{eq:79} with a $\CC$-linear isometry
  \begin{equation}\label{eq:82}
     \lambda _{g_2,g_1}\: \cc{g_1}\tL_{g_2}\otimes  \tL_{g_1} \longrightarrow
     \tL_{g_2g_1},
  \end{equation}
and replace~\eqref{eq:80} with 
  \begin{equation}\label{eq:83}
     \xymatrix@C=70pt{\cc {g_2g_1}\tL_{g_3}\otimes  
     \cc{g_1}\tL_{g_2}\otimes   \tL_{g_1}\ar[r]^(.55){\id\otimes
     \lambda _{g_2,g_1}}\ar[d]_{\lambda _{g_3,g_2}\otimes \id} &
     \cc{g_2g_1}\tL_{g_3}\otimes   \tL_{g_2 g_1}\ar[d]^{\lambda
     _{g_3,g_2 g_1}} \\ \cc{g_1}\tL_{g_3 g_2}\otimes  
     \tL_{g_1}\ar[r] ^(.55){\lambda _{g_3, g_2,g_1}}& \tL_{g_3 g_2 g_1}} 
  \end{equation}
which can be interpreted as a twisted cocycle condition.   

  \begin{remark}[]\label{thm:124} To see a cocycle~$\alpha $ explicitly,
choose a unit norm vector~$\ell _g\in \tL_g$ for each~$g\in G$.  Then
for~$g_1,g_2\in G$ define $\alpha (g_2,g_1)\in \TT$ by
  \begin{equation}\label{eq:128}
     \lambda _{g_2,g_1}(\cc{g_1}\ell _{g_2},\ell _{g_1}) = \alpha (g_2,g_1)\ell
     _{g_2g_1}, 
  \end{equation}
where as in~\eqref{eq:81} the left superscript~$\phi $ controls complex
conjugation.  The commutative diagram~\eqref{eq:83} translates into the
twisted cocycle condition 
  \begin{equation}\label{eq:129}
     \alpha (g_2,g_1)\alpha (g_3,g_2g_1) = \cc{g_1}\alpha (g_3,g_2)\alpha
     (g_3g_2,g_1),\qquad      g_1,g_2,g_3\in G. 
  \end{equation}
A different choice $\ell _g\mapsto b(g)\ell _g$, $b(g)\in \TT$,
shifts~$\alpha$ by a twisted coboundary.
  \end{remark}

In preparation for twisted $K$-theory let us note that a
$(\phi,\tau,c)$-twisted representation of $G$ is for each~$g\in G$ a linear
map
  \begin{equation}\label{eq:84}
     \tar (g)\:\cc{g}{(\tL_g\otimes W)}\longrightarrow  W 
  \end{equation}
such that the diagram 
  \begin{equation}\label{eq:85}
     \xymatrix@C=70pt{\cc {g_2g_1}\bigl(\cc{g_1}\tL_{g_2}\otimes \tL_{g_1}\otimes
     W\bigr)\ar[r]^(.55){ \lambda _{g_2,g_1}\otimes \id}\ar[d]_{\id\otimes \tar
     (g_1)} & \cc{g_2g_1}\bigl(\tL_{g_2g_1}\otimes W\bigr)\ar[d]^{\tar  (g_2
     g_1)} \\ 
     \cc{g_2}(\tL_{g_2}\otimes W)\ar[r] ^(.55){\tar (g_2)}& W}  
  \end{equation}
commutes and the linear map~$\tar (g)$ in~\eqref{eq:84} is even or odd
according to~$c(g)$.

  \begin{remark}[]\label{thm:75}
 It is convenient---and for the addition of twistings necessary---to combine
the homomorphism~$c\:G\to\pmo$ and the twisted  extension~$\tau $ into
a \emph{$\zt$-graded} line bundle~$\tL\to G$: the line~$L_g$ is even
if~$c(g)=+1$ and odd if $c(g)=-1$.  The Koszul sign rule
(Remark~\ref{thm:67}) must be used to permute tensor products of $\zt$-graded
lines.  
  \end{remark}

  \subsection*{Groupoids, twistings, and twisted vector bundles}

We now generalize from groups to \emph{groupoids}.  The elements of a group
may be pictured as arrows which begin and end at the same (abstract) point,
and the group law is composition of arrows.  A groupoid has a similar
picture, but now there are many possible starting and ending points for the
arrows and composition is limited to a sequence of arrows with the ending
point of one arrow equal to the starting point of the next.  Thus a groupoid
consists of a set~ $\gp_0$ of points and $\gp_1$~of arrows.  The arrow
$(x_0\xrightarrow{\gamma }x_1)$ has source~$x_0$ and target~$x_1$.  There is
a partially defined associative composition law on arrows: for $\gamma
_1,\gamma _2\in \gp_1$ the composition $\gamma _2\circ \gamma _1\in \gp_1$ is
defined only if the target for~$\gamma _1$ equals the source for~$\gamma _2$.
There is an identity arrow for each state.  Each arrow is
invertible---two-sided inverse arrows exist.  As we discuss below,
$\gp_0,\gp_1$~may be topological spaces, or even manifolds, in which case we
impose continuity (smoothness) on the structural maps which define the
groupoid structure.  The structural maps (except for composition) appear in
the diagram
  \begin{equation}\label{eq:90}
     \xymatrix@1{\gp_0\;\ar@{-->}[r]&\;\gp_1\;\ar@<1ex>[l]^{p_1}
     \ar@<-1ex>[l]_{p_0}}  
  \end{equation}
The dashed map assigns the identity morphism to each object, the map~$p_1$
takes each arrow to its source (=domain), and the map~$p_0$ takes each arrow
to its target (=codomain).  See~\cite[Appendix]{FHT1} for a formal definition
and development of groupoids and the associated $K$-theory.

We give several examples to help the reader warm up to this notion. 

  \begin{example}[]\label{thm:63}
 As mentioned above, a group~$G$ may be viewed as the groupoid~$\gp$ with
$\gp_0=\pt$ and $\gp_1=G$.  There is a unique object and any two arrows are
composable.
  \end{example}

  \begin{example}[]\label{thm:64}
 A set~$X$ is a groupoid~$\gp$ with only identity arrows, that is,
$\gp_0=\gp_1=X$.   
  \end{example}

  \begin{example}[]\label{thm:65}
 Let $X$~be a set and $G$~a discrete group acting on~$X$.  We construct the
quotient groupoid~$\gp=X\gpd G$ with $\gp_0=X$ and $\gp_1=X\times G$.  The
arrow $\gamma =(x,g)\in X\times G$ has domain~$x$ and codomain~$g\cdot x$,
where $g\cdot x$~is the result of acting~$g$ on~$x$.  Composition of arrows
is defined by the action.  For each~$x$ the arrows which map~$x$ to itself
form a group, called the stabilizer group of~$x$; it is a subgroup of~$G$.
The orbit of an element in~$X$ also has a natural interpretation in terms of
the groupoid~$\gp$.  We use a topological version of~$X\gpd G$ in which
$X$~is a nice topological space (in fact, a smooth torus of some dimension)
and $G$~is a compact Lie group.  Then the $K$-theory of~$X\gpd G$ , as
defined below, is identical to the $G$-equivariant $K$-theory of~$X$.
  \end{example}

A homomorphism of groupoids $\varphi \:\gp'\to\gp$ is a pair of maps $\varphi
_0\:(\gp')_0\to\gp_0 $ and $\varphi _1\:(\gp')_1\to\gp_1$ which preserves
compositions.  For example, for any~$\gp$ there is an inclusion of the
groupoid~$\gp'$ of identity arrows, which is a set or space as in
Example~\ref{thm:64}: $(\gp')_0=(\gp')_1=\gp_0$.
 
Next we generalize representations of groups to groupoids.  We begin with
discrete groupoids.  A complex representation of a group~$G$, viewed in terms
of the associated groupoid~$\gp$ defined in Example~\ref{thm:63}, assigns a
complex vector space~$W$ to the unique point in~$\gp_0$ and a complex linear
map to each arrow in~$\gp_1$ such that the composite of two arrows maps to
the composite of the linear maps.  The generalization to groupoids is
straightforward. 

  \begin{definition}[]\label{thm:66}
 Let $\gp$~be a groupoid.  A \emph{complex vector bundle over~$\gp$} assigns
a complex vector space~$W_x$ to each object~$x\in \gp_0$ and a complex linear
map $\rho (\gamma )\:W_x\to W_{x'}$ to each arrow~$\gamma \in \gp_1$ with
domain~$x$ and codomain~$x'$.  The assignment~$\rho $ is a homomorphism in
the sense that $\rho (\gamma _2\circ \gamma _1)=\rho (\gamma _2)\circ \rho
(\gamma _1)$ whenever $\gamma _2\circ \gamma _1$~is defined.  If $\gp$~is a
topological groupoid, then the vector spaces are required to fit together into
a vector bundle~$W\to \gp_0$ and the linear maps~$\rho (\gamma )$ to form a
continuous map $p_1^*W\to p_0^*W$ over~$\gp_1$.
  \end{definition}

\noindent 
 Recall that $p_0,p_1$~are defined in~\eqref{eq:90}.  If $\gp=X\gpd G$ is the
quotient of a set by a group action, then a vector bundle over~$\gp$ is a
$G$-equivariant vector bundle $W\to X$.

We now introduce certain extensions of a groupoid, which are special twistings
of its $K$-theory.

  \begin{definition}[]\label{thm:85}
 Let $\gp$ be a groupoid.
  \begin{enumerate}
 \item A \emph{central extension} of~$\gp$ is a hermitian line bundle
$\tL\to\gp_1$ together with maps~$\lambda _{\gamma _2,\gamma _1}$ as
in~\eqref{eq:79} which satisfy the associativity (cocycle)
constraint~\eqref{eq:80}.

 \item Let $\phi \:\gp_1\to\pmo$ be a homomorphism.  A \emph{$\phi $-twisted
 extension} of~$\gp$ is a hermitian line bundle $\tL\to\gp_1$ together
with maps~$\lambda _{\gamma _2,\gamma _1}$ as in~\eqref{eq:82} which satisfy
the associativity (cocycle) constraint~\eqref{eq:83}.

 \item A \emph{\ptzgce} of~$\gp$ is a triple~$\nu =\ptc$ consisting of
homomorphisms $\phi ,c\:\gp_1\to\pmo$ and a $\phi $-twisted 
extension~ $\tau $.  

 \item If $\nu =\ptc$ is a \ptzgce, then a \emph{$\nu$-twisted vector
bundle}~ $W$ over~$\gp$ assigns a $\zt$-graded complex vector space~$W_x$ to
each object~$x\in \gp_0$ and a \emph{linear} map
  \begin{equation}\label{eq:86}
     \tar (\gamma) \:\cc{\gamma }{(\tL_\gamma \otimes W_{x_0})}\longrightarrow
     W_{x_1} 
  \end{equation}
to each arrow $(x_0\xrightarrow{\gamma }x_1)\in \gp_1$.  The map~$\tar
(\gamma) $ is even or odd according to~$c(\gamma )$.  If $\gp$~is a
topological groupoid, then we demand that $W$~be a vector bundle and
\eqref{eq:86} define a continuous map $\cc\gamma (\tL\otimes p_1^*W)\to
p_0^*W$ over~$\gp_1$.  An analog of the commutative diagram~\eqref{eq:85}
expresses the compatibility between~$\lambda $ and~$\rho $.

  \end{enumerate}
  \end{definition}

\noindent 
 In~(i) $\gamma _1,\gamma _2,\gamma _3\in \gp_1$ are composable arrows, and
in~\eqref{eq:79}, \eqref{eq:80}, \eqref{eq:82}, \eqref{eq:83},
and~\eqref{eq:85} we replace~`$g$' with~`$\gamma $'.  Note that central
extensions and $\phi $-twisted  extensions are special cases of
\ptzgces.

  \begin{remark}[]\label{thm:69}
 We emphasize that $W\to\gp_0$ is an ordinary $\zt$-graded complex vector
bundle.  The twisting is in the action of~$\gp_1$, which is the same twisting
as occurs in Definition~\ref{thm:61}.  The only new idea is that of multiple
objects, as parametrized by~$\gp_0$.
  \end{remark}

  \begin{remark}[]\label{thm:119}
 Let $\gp=X\gpd G$ be the quotient groupoid of the action of a group~$G$ on a
space~$X$.  Suppose $\phi ,c\:G\to\pmo$ are homomorphisms.  Then there is an
associated \ptzgce~ $\nu (\phi ,c)$ of~$\gp$.  Namely, $\gp_1=X\times G$ and
we define homomorphisms $\gp_1\to\pmo$ by composition with projection
$X\times G\to G$.  The central extension $\tL\to\gp_1$ is taken to be
trivial.  We use this \ptzgce\ in Theorem~\ref{thm:51} below.
  \end{remark}

  \begin{definition}[]\label{thm:117}
 Let $\gp$~be a groupoid and $\nu =\ptc$ and $\nu '=(\phi ',\tau ',c')$
\ptzgces.  Then if $\phi =\phi '$ and~$c=c'$, an \emph{isomorphism of
\ptzgces} $\nu \to\nu '$ is an isomorphism $\tL\to L^{\tau '}$ of line
bundles over~$\gp_1$ which is compatible with the structure
maps~\eqref{eq:82}.
  \end{definition}

\noindent 
 Isomorphism of \ptzgces\ is an equivalence relation, and so there is a set
of equivalence classes of \ptzgces. 

  \begin{remark}[]\label{thm:100}
 We can then ask for a classification of \ptzgces\ up to isomorphism.  If the
groupoid~$\gp$ is sufficiently nice, for example the groupoid associated to
the action of a compact Lie group~$G$ on a nice topological space~$X$ as in
Definition~\ref{thm:71} below, then the set of isomorphism classes of
\ptzgces\ is noncanonically isomorphic to 
  \begin{equation}\label{eq:197}
     H^1_G(X;\zt)_{\textnormal{rel}}\times H^3_G(X;\ZZ_\phi )_{\textnormal{rel}} ,
  \end{equation}
where
  \begin{equation}\label{eq:196}
     H^\bullet_G(X;\ZZ_\phi )_{\textnormal{rel}} = \ker\bigl(H^\bullet_G(X;\ZZ_\phi
     )\longrightarrow H^\bullet(X;\ZZ) \bigr). 
  \end{equation}
The subscript~`$G$' denotes equivariant cohomology, and `$\ZZ_\phi $'~denotes
the local system (twisted coefficients) defined by~$\phi $.
See~\cite[\S2.2.1]{FHT1} for more details and a proof (in case $\phi $~is
trivial).  If $G$~is a finite group and $X$~is compact, then
\eqref{eq:197}~is a finite set.
  \end{remark}

  \begin{remark}[]\label{thm:76}
 There are various operations on \ptzgces.  If $\varphi \:\gp'\to\gp$ is a
homomorphism of groupoids, and $\nu =\ptc$ is a \ptzgce\ of~$\gp$, then there
is a pullback \ptzgce~$\varphi ^*\nu $ of~$\gp'$.  Also, if $\gp$~is a
groupoid and $\nu _1=(\phi ,\tau _1,c_1)$, $\nu _2=(\phi ,\tau _2,c_2)$ are
\ptzgces\ of~$\gp$ with the same homomorphism~$\phi $, then there is a new
\ptzgce~$\nu =\nu _1+\nu _2$ of~$\gp$, also with the same~$\phi $.  To define
it combine~$c_i,\tau _i$ into a $\zt$-graded line bundle $L^{\tau
_i}\to\gp_1$ and set $L=L^{\tau _1}\otimes L^{\tau _2}$.  The Koszul sign
rule is used to define~\eqref{eq:82}.  The addition gives an abelian group
law on~\eqref{eq:197} which is not, in general, the product of the natural
abelian group laws on the factors.  
  \end{remark}

  \begin{remark}[]\label{thm:80}
 Let $\varphi \:\gp'\to\gp$ be the inclusion of the identity arrows
$(\gp')_0=(\gp')_1=\gp_0$.  Then for any \ptzgce\ $\nu $, the
pullback~$\varphi ^*\nu $ is a trivial twisting of~$\gp'$.  In other words,
the twistings in this paper are non-equivariantly trivial, hence simpler than
general twistings of $K$-theory.
  \end{remark}

  \subsection*{Twisted $K$-theory}

The following is an extension of Definition~\ref{thm:61} to the special
groupoids and special twistings (Remark~\ref{thm:80}) we encounter in this
paper.  It is justified in Appendix~\ref{sec:19}.

  \begin{definition}[]\label{thm:71}
  Let $X$~be a nice\footnote{locally contractible and completely regular}
compact topological space with a continuous action of a \emph{finite}
group~$G$.  Let $\gp=X\gpd G$ be the quotient groupoid and $\nu =\ptc$~a
\ptzgce\ of~$\gp$.
  \begin{enumerate}
 \item Let $\sV_G^{\nu}(X)$ denote the monoid of isomorphism classes of
finite rank $\nu$-twisted $\zt$-graded $G$-equivariant vector bundles $W\to
X$.

 \item Let the submonoid~$\Triv_G^{\nu}(X)$ consist of finite rank
$\nu$-twisted $\zt$-graded $G$-equivariant vector bundles $W\to X$ such that
there exists an odd automorphism $P\:W\to W$.

 \item Define the \emph{twisted $K$-theory group} $K^\nu _G(X)=\hK\phi {\tau
,c}_G(X)=\Rep_s^{\ptc}(G)/\Triv_s^{\ptc}(G)$.

  \end{enumerate}

  \end{definition}

\noindent
 Recall that a $G$-equivariant vector bundle over~$X$, twisted or not, is an
ordinary vector bundle over~$X=\gp_0$.

  \begin{remark}[]\label{thm:72}
 In case $G=1$ is the trivial group, then Definition~\ref{thm:71} reduces to
the original definition of the topological $K$-theory group~$K^0(X)$.
Similarly, if $\phi \equiv c\equiv 1$ and $\tau $~is the trivial extension,
then it reduces to the original definition of~$K^0_G(X)$.  See~\cite{A1} for
a nice, detailed exposition.  
  \end{remark}

  \begin{remark}[]\label{thm:74}
 Suppose $X$~is a space with involution~$\sigma \:X\to X$, which we regard as
a space~$X$ with $G=\zt$ action.  Define $\phi \:\gp_1\to\pmo$ on the quotient
groupoid $\gp=X\gpd G$ to be~1 on all identity arrows (as it must be) and
$-1$ on all arrows which encode the $\sigma $~action.  Let $c,\tau $~be
trivial.  Then the twisted $K$-theory group~$\hKp\phi _G(X)\cong KR^0(X)$ is
Atiyah's (untwisted) $KR$-group; see~\cite{A2}.   More generally, if
$G=G''\times \zt$ and $\phi $~is projection onto the second factor, then
$\hKp\phi _G(X)\cong KR_{G''}^0(X)$ is an equivariant $KR$-theory group.
  \end{remark}

  \begin{remark}[]\label{thm:137}
 If $c\equiv 1$, so that $\nu $~is a $\phi $-twisted  extension, then
the twisted $K$-theory is isomorphic to the group completion of the
monoid~$\Vect^\nu _G(X)$ of ungraded finite rank $\nu $-twisted
$G$-equivariant vector bundles $W\to X$.
  \end{remark}

  \begin{remark}[]\label{thm:73}
 $K$-theory is defined for a more general class of groupoids (\emph{local
quotient groupoids}) in~\cite[Appendix]{FHT1}.  The general definition uses
families of Fredholm operators.  It is a \emph{theorem}, proved in
Appendix~\ref{sec:19}, that for the special twistings in
Definition~\ref{thm:85}, namely \ptzgces\ of a global quotient by a finite
group~$G$, each twisted $K$-theory class has a representative which is a
$\zt$-graded vector bundle of finite rank.  It may fail if the group~$G$ is a
compact Lie group.  For example, if $G=\TT$~acts trivially on~$X=\TT$ there
are nontrivial central extensions of infinite order, and only the zero
twisted $K$-theory class has a finite rank representative.
  \end{remark}

Finally, we remark on the role of Clifford algebras in $K$-theory~\cite{ABS}.
We review some basics about Clifford algebras in Appendix~\ref{sec:12}.  If
$R$~is a Clifford algebra\footnote{or, more generally, a $\zt$-graded central
simple algebra~\cite{Wa}} then we can use $R$~as an additional twisting of
representations and of vector bundles in the following sense.  Namely, we ask
that a representation, or more generally the fibers of a vector bundle, be
modules for the algebra~$R$, and that the equivariance maps~\eqref{eq:84}
commute with the $R$-module structure.  This commutation must be understood
in the graded sense, that is, using the Koszul sign rule.  If $R=\Cl n$ for
some~$n\in \ZZ$, this has the effect of shifting the degree in $K$-theory
by~$n$.   Details may be found in many sources, for example~\cite{DK},
\cite[Appendix]{FHT1}, \cite{F3}.

   \section{Gapped systems with a finite dimensional state space
   }\label{sec:11}

In this section we take up systems of free fermions with a finite dimensional
state space.  A similar problem is treated in~\cite{K,AK}. As quantum
mechanical systems they have a time evolution, but we do not assume any
symmetries from space.  In other words, the underlying Galilean
spacetime~$(M,\Gamma )$ is 1-dimensional.  So $\Aut(M,\Gamma )$~is an
extension
  \begin{equation}\label{eq:64}
     1\longrightarrow U\longrightarrow \Aut(M,\Gamma )\longrightarrow
     \pmo\longrightarrow 0 
  \end{equation}
where $U$~is the line of time translations.  The extension is split by
choosing an origin of time, and a splitting maps $-1\in \pmo$ into the
time-reversal which fixes the origin.  Let $(G,\phi ,\tau ,c)$ be an extended
\Qsymmetry class, as in Definition~\ref{thm:84}.  A Galilean symmetry group
includes a homomorphism~\eqref{eq:98} to ~$\Aut(M,\Gamma )$, but as in
Remark~\ref{thm:98} we do not include time translations in~$G$.  Also, the
homomorphism~$t$ to the quotient~$\pmo$ in~\eqref{eq:64} is~$t=\phi c$, so is
already included in the data of the extended \Qsymmetry class.  We make the
following ``finiteness'' hypothesis.

  \begin{hypothesis}[]\label{thm:38}
 The symmetry group~$G$ of a 0-dimensional gapped system is a compact Lie
group. 
  \end{hypothesis}

\noindent
 In fact, it is usually assumed to be a finite group.  Recall from
Definition~\ref{thm:25} that a gapped system~$(\sH,H,\tar)$ consists of a
Hilbert space~$\sH$, a Hamiltonian~$H$, and a $\ptc$-twisted
representation~$\tar$ of~$G$ on~$\sH$.  By Lemma~\ref{thm:27} we may assume
$H^2=1$, after a homotopy, so $H$~is a grading operator on~$\sH$.  The
finiteness condition we work with is the following.

  \begin{hypothesis}[]\label{thm:39}
 In a finite dimensional gapped system (FDGS) the Hilbert space~$\sH$ is
finite dimensional.
   \end{hypothesis}

  \begin{proposition}[]\label{thm:101}
 Assume Hypothesis~\ref{thm:38} and Hypothesis~\ref{thm:39}.  Then the
commutative monoid~$\TPZ\Gptc$ of topological equivalence classes of FDGS
with extended \Qsymmetry class $(G,\phi ,\tau ,c)$ is isomorphic to
$\Rep_s^{\ptc}(G)$,

  \end{proposition}

\noindent
 The monoid $\TPZ\Gptc$ is defined in Definition~\ref{thm:26} and the monoid
$\Rep_s^{\ptc}(G)$ is defined in Definition~\ref{thm:61}(i).

  \begin{proof}
 First, since the space of positive definite hermitian metrics on~$\sH$ is
contractible, and since we mod out by homotopy, we can ignore the metric
on~$\sH$.  Next, finite dimensional representations of compact Lie groups are
discrete---they have no continuous deformations---so two homotopic
representations are in fact isomorphic.   
  \end{proof}

We impose a ``topological triviality'' relation on $\TPZ\Gptc$ to define the
abelian group of reduced topological phases~$\GSZ(G,\phi ,\tau ,c)$; see the
text preceding Lemma~\ref{thm:27}.

  \begin{definition}[]\label{thm:102}
 A FDGS~$(\sH,H,\tar)$ with extended \Qsymmetry class $(G,\phi ,\tau ,c)$ is
\emph{topologically trivial} if there exists an odd automorphism
$P\:\sH\to\sH$ such that $P\tar(g)= c(g)\tar(g)P$ for all $g\in \tG$.
  \end{definition}

  \begin{theorem}[]\label{thm:40}
 There is an isomorphism $\GSZ(G,\phi ,\tau ,c)\cong \hK\phi {\tau ,c}_G $.
  \end{theorem}

\noindent 
 The twisted $K$-theory group in the theorem is given in
Definition~\ref{thm:61}.  The proof is immediate from the definitions.

We can identify the twisted virtual representation ring with a more standard
$K$-theory group if we make an additional strong assumption.  It is not, as
far as we know, justified on physical grounds.  In fact, it is stronger than
the $\psi $-standard assumption in ~\S\ref{sec:4}.  We include it here to
make contact with the literature, and in particular with
Proposition~\ref{thm:24}.  Consider $\psi =(t,c)\:G\to\sC=\pmo\times \pmo$
and let $A\subset\sC$ be its image.  Set $G_0=\ker\psi $.

  \begin{hypothesis}[]\label{thm:41}
 \ 
  \begin{enumerate}
 \item The group~$G$ is a direct product $G\cong A\times G_0$ and under this
isomorphism $\psi $~is projection onto~$A$. 

 \item The restriction~$\tG_0\to G_0$ of $\tG\to G$ splits and we fix a
splitting.
  \end{enumerate} 
Therefore, there is fixed an isomorphism $\tG\cong \tA\times G_0$.
  \end{hypothesis}

There are four nontrivial subgroups of~$\sC$; compare
Proposition~\ref{thm:24}.  Elements~$\bT,\bC,\bS=\bT\bC$ may or may not be
in~$A$, depending on the case.  They have special lifts~$T,C,S$ to the
extension~$\tA$, hence to~$\tG$, as follows immediately from
Proposition~\ref{thm:24}.

  \begin{lemma}[]\label{thm:42}
 \ 
  \begin{enumerate}
 \item If $A$~is the diagonal subgroup of~$\sC$, then $\bS$~has a lift~$S\in
\tA$ with $S^2=1$ and $S$~central. 

 \item If $A=\pmo\times 1\subset \sC$, then $\bT$~has a lift~$T\in \tA$
with $T^2=\pm1$. 

 \item If $A= 1\times\pmo\subset \sC$, then $\bC$~has a lift~$C\in \tA$
with $C^2=\pm1$.  

 \item If $A=\sC$, then $\bT,\bC$~have lifts $T,C\in \tA$ with $TC=CT$,
$T^2=\pm1$, $C^2=\pm1$.
  \end{enumerate}
  \end{lemma}

  \begin{corollary}[]\label{thm:43}
 If Hypothesis~\ref{thm:41} holds, then we have the following table
for~$\hK\phi {\tau ,c}_G$:

  \begin{center}
  \renewcommand{\arraystretch}{1.5}
  \begin{tabular}{|c||c|c|c|c|c|c|c|c|c|c|}
\hline 
$A$&1&\textnormal{diag}&$\pmo\times \oo$&$\sC$&$\oo\times \pmo$&$\sC$&$\pmo\times \oo$&$\sC$&$\oo\times \pmo$&$\sC$ \\
\hline
$T^2$&&&$+1$&$+1$&&$-1$&$-1$&$-1$&&$+1$ \\
\hline
$C^2$&&&&$-1$&$-1$&$-1$&&$+1$&$+1$&$+1$\\
 \hline&&&&&&&&&&\\[-1.6em] \hline
$\hK\phi
{t,c}_G$&$K^0_{G_0}$&$K^{-1}_{G_0}$&$KO^0_{G_0}$&$KO^{-1}_{G_0}$&$KO^{-2}_{G_0}$&$KO^{-3}_{G_0}$&$KO^{-4}_{G_0}$&$KO^{-5}_{G_0}$&$KO^{-6}_{G_0}$&$KO^{-7}_{G_0}$\\ 
 \hline 
  \end{tabular}
  \renewcommand{\arraystretch}{1}
  \end{center}
\vskip12pt
  \end{corollary}

  \begin{proof}
 An action of~$G$ on~$\sH$ of the type that occurs in Theorem~\ref{thm:40} is
equivalent to a linear representation of~$G_0$ on~$\sH$ together with a
\emph{commuting} action of whichever elements~$T,C,S=TC$ are present.
Combining with the interpretation in terms of Clifford algebras developed in
Appendix~\ref{sec:12}, specifically Proposition~\ref{thm:57}, and recalling
from the end of~\S\ref{sec:3} that Clifford modules represent $K$-theory in
shifted degrees, we arrive at the table.
  \end{proof}

   \section{Twistings from group extensions}\label{sec:5}

  \subsection*{Warmup: direct products}

Let $G=G''\times G'$ be a direct product group.\footnote{For our motivational
purposes the reader may restrict to the simplest case in which $G'',G'$~are
finite groups, but with additional care the discussion applies to infinite
discrete groups and Lie groups as well.}  Then an irreducible complex
representation of~$G$ is a tensor product~$R''\otimes R'$ of irreducible
representations~$R''$ of~$G''$ and~$R'$ of~$G'$.  Hence any representation~$E$
of~$G$ decomposes as 
  \begin{equation}\label{eq:40}
     \begin{aligned} E &\cong \bigoplus\limits_{R'',R'}\; \Hom_G(R''\otimes
      R',E)\otimes R''\otimes R' \\ &\cong \bigoplus\limits_{R'}\;
      \Hom_{G'}(R',E)\otimes R' ,\end{aligned} 
  \end{equation}
where $R'',R'$ run over distinguished representatives of the isomorphism
classes of irreducible representations of~$G'',G'$, respectively.  Recall
that `$\Hom_G$'~denotes the vector space of
\emph{intertwiners}---$G$-equivariant maps---between two representations
of~$G$.  Thus $\dim \Hom_{G'}(R',E)$ is the multiplicity of the irreducible
representation~$R'$ in~$E$.  Decomposing under~$G''\subset G$ we deduce that
the multiplicity space in the second line of~\eqref{eq:40} is
  \begin{equation}\label{eq:41}
     \Hom_{G'}(R',E) \cong \bigoplus_{R''}\;\Hom_G(R''\otimes R',E)\otimes R'' .
  \end{equation}
Let $X$~be a set (or space) of irreducible representations of~$G'$.  We
label points of~$X$ by the chosen distinguished irreducible
representations~$R'$.  The multiplicity spaces form a family of complex
vector spaces $\sE\to X$ defined by
  \begin{equation}\label{eq:42}
     \sE_{R'} = \Hom_{G'}(R',E),\qquad R'\in X. 
  \end{equation}
In general $\sE\to X$ is a \emph{sheaf}, not a vector bundle, since the
family of vector spaces may not be locally trivial.   The group~$G''$ acts trivially on~$X$, and from~\eqref{eq:41} we see that
$G''$~acts on each fiber of~$\sE$: in each summand it acts trivially on
$\Hom_G(R''\otimes R',E)$ and irreducibly on~$R''$.  Furthermore, this
construction induces an equivalence between finite dimensional
representations of~$G$ and finitely supported $G''$-equivariant families of
complex vector spaces parametrized by~$X$.  The inverse construction is
straightforward: if $\sE\to X$ has finite support, then
  \begin{equation}\label{eq:43}
     E=\bigoplus\limits_{R'\in X}\sE_{R'}\otimes R' 
  \end{equation}
is a finite dimensional representation of~$G$.  If $G$~is compact, then
$X$~is discrete and this construction induces an isomorphism of $K$-groups
$K_G\cong K_{G''}(X)_\cpt$, where `$\cpt$' denotes compact (hence finite,
since $X$~is discrete) support.

  \begin{example}[]\label{thm:123}
 Consider~$G'=\ZZ$ the free abelian group on one generator.  Its irreducible
complex representations are all one dimensional and parametrized by the
circle group~$X=\TT$ of phases $\lambda =e^{i\theta }$.  The underlying
vector space of each representation is~$\CC$ and in the representation
labeled~$\lambda $ the integer~$n\in \ZZ$ acts on~$\CC$ as multiplication
by~$\lambda ^n$.  Let $G''$~be the trivial group.  For the irreducible
representation~$E$ labeled by~$\lambda $ the vector space~$\sE_\mu ,\;\mu \in
\TT$ is canonically~$\CC$ if~$\mu =\lambda $ and is zero otherwise, so we
obtain a sheaf which is not a vector bundle.  On the other hand, if $E$~is
the vector space of complex-valued ($L^2$) functions on~$\ZZ$, and $\ZZ$ acts
by translation, then this construction reduces to the Fourier transform:
each~$\sE_\lambda $ is canonically~$\CC$, and the construction identifies~$E$
as the space of complex-valued ($L^2$) functions on~$\TT$.
  \end{example}

  \subsection*{Group extensions and twistings}

Let
  \begin{equation}\label{eq:44}
     1\longrightarrow G'\longrightarrow G\xrightarrow{\;\;\pi
     \;\;}G''\longrightarrow 1 
  \end{equation}
be a group extension.  As before let $X$~be the space of isomorphism classes
of irreducible complex representations of~$G'$.  Then there is a \emph{right}
action of~$G''$ on~$X$ as follows.  Let $\rep V\:G'\to\Aut(V)$ be an
irreducible representation of~$G'$ and~$g''\in G''$.  Choose~$g\in G$
with~$\pi (g)=g''$.  Then if $[V]\in X$ denotes the isomorphism class of~$V$,
define $[V]\cdot g''$ to be the isomorphism class of the composition
  \begin{equation}\label{eq:45}
     \rep{V^g}\:G'\xrightarrow{\;\;\alpha (g)\;\;}G'\xrightarrow{\;\;\rep
     V\;\;}\Aut(V), 
  \end{equation}
where $\alpha (g)(g')=gg'g\inv $ for all~$g'\in G'$.  One checks that the
isomorphism class of~\eqref{eq:45} depends only on~$g''\in G''$, not on the
lift~$g\in G$.  Also, note that $(V^{g_1})^{g_2}=V^{g_1g_2}$.

The following theorem is proved in case $G$~is compact in
\cite[Examples~1.12--13]{FHT1} and in more general form in \cite[\S5]{FHT2}.
We give an exposition here.  Our main application is an extension in which
the quotient~$G''$ is a finite group and the kernel~$G'$ is isomorphic to a
lattice~$\ZZ^d$ for some~$d\in \ZZ^{\ge0}$.  Since $\ZZ^d$~is abelian, its
complex irreducible representations are one-dimensional.  Hence in this case
$X$~is diffeomorphic to the $d$-dimensional torus of characters
$\Hom(\ZZ^d,\TT)\simeq\TT^d$.  Recall Definition~\ref{thm:85}.

  \begin{theorem}[]\label{thm:16}
  Let \eqref{eq:44} be an extension of Lie groups and $X$~the space of
isomorphism classes of irreducible representations of~$G'$.  Assume $G''$~is
compact and $G'$~is either compact or a lattice.

  \begin{enumerate}
 \item There is an induced central extension~$\nu $ of the groupoid~$X\gpd
G''$ such that a representation~$E$ of~$G$ induces a $\nu $-twisted
$G''$-equivariant sheaf over~$X$.  (Infinite dimensional representations are
assumed unitary.)

 \item If $G$~is compact, there is an isomorphism $K_G\to K_{G''}^\nu (X)_\cpt$,
where `$\cpt$' denotes `compact support'. 

 \item If $G'$~is abelian, then a splitting of~\eqref{eq:44} trivializes the
central extension~$\nu $. 

  \end{enumerate}
  \end{theorem}

\noindent 
 A trivialization is defined in Definition~\ref{thm:117} (with $\phi =\phi
'$, $c=c'$, and~$\tau $ all trivial).

  \begin{remark}[]\label{thm:130}
 The \ptzgce~$\nu $ depends on some choices (of representative irreducible
representations of~$G'$ in each isomorphism class), but we can track how~$\nu
$ changes under a change of these choices, and so we say $\nu $~is defined up
to canonical isomorphism.  In case $G'$~is abelian, as it is in our
application to gapped topological insulators in~\S\ref{sec:6}, there  are
canonical choices so no ambiguity in the definition of~$\nu $. 
  \end{remark}

The following simple examples with finite groups illustrate the theorem.

  \begin{example}[]\label{thm:20}
 The cyclic group of order~6 is an extension 
  \begin{equation}\label{eq:46}
     1\longrightarrow \zmod3\longrightarrow \zmod6\longrightarrow
     \zt\longrightarrow 1 
  \end{equation}
which is split: $\zmod6\cong \zmod2\times \zmod3$.  Each complex irreducible
representation of the abelian group~$\zmod3$ is one-dimensional, whence the
space~$X$ of isomorphism classes is the space of characters
$\Hom(\zmod3,\TT)$, which consists of 3~points.  Let $\omega =e^{2\pi i/3}$.
Label a character $\lambda \:\zmod3\to\TT$ by~$\lambda (1)$, so that
$X=\{1,\omega ,\omega ^2\}$.  The nonzero element of~$\zt$ acts on~$X$ by the
trivial involution~$\sigma =\id_X$.  The induced central extension~$\nu $
of~$X\gpd(\zt)$ is canonically trivial.
  \end{example}

  \begin{example}[]\label{thm:21}
 The automorphism group~$\Sigma _3$ of a set of 3~elements is a split
extension
  \begin{equation}\label{eq:47}
     1\longrightarrow \zmod3\longrightarrow \Sigma _3\longrightarrow
     \zt\longrightarrow 1 
  \end{equation}
Now $\sigma \:X\to X$ is the nontrivial involution which exchanges $\omega
\leftrightarrow \omega ^2$ and fixes~$1$.  The induced central
extension~$\nu $ is trivializable and $K_{\Sigma _3}$~is isomorphic to
(untwisted)~$K_{\zt}(X)$.
  \end{example}

  \begin{example}[]\label{thm:22}
 The 8-element quaternion group~$Q$ is a nonsplit extension 
  \begin{equation}\label{eq:48}
     1\longrightarrow \zmod2\longrightarrow Q\longrightarrow\zt\times 
     \zt\longrightarrow 1 
  \end{equation}
The quotient~$\zt\times \zt$ acts trivially on the space $X=\{1,-1 \}$ of
characters of~$\zmod2$.  Let $\nu $~be the induced central extension
of~$X\gpd(\zt\times \zt)$. 
The restriction of~$\nu $ to $\{1\}\subset X$ is trivial while the restriction
of~$\nu $ to~$\{-1\}\subset X$ is represented by the nontrivial central
extension 
  \begin{equation}\label{eq:59}
     1\longrightarrow \TT\longrightarrow (Q\times \TT)/\pmo\longrightarrow
     \zt\times \zt\longrightarrow 1. 
  \end{equation}
This is a special case of Proposition~\ref{thm:103} below.  Moreover, the
fact that \eqref{eq:59}~ is nonsplit follows from Lemma~\ref{thm:105} below.
  \end{example}

  \begin{proof}[Proof of Theorem~\ref{thm:16}]
 Assume first that $X$~is discrete, so $G''$~acts on~$X$ through the discrete
quotient~$\pi _0G''$.  (The assumption that $X$~is discrete remains in force
until the last paragraph of the proof.)  Choose a distinguished irreducible
representation $\rep V\:G'\to\Aut(V)$ in each equivalence class ~$[V]\in X$.
We construct a central extension (recall Definition~\ref{thm:85}(i)) of the
groupoid~$X\gpd G$ and then descend it to a central extension of~$X\gpd G''$,
which we define to be~$\nu $.  Recall from~\eqref{eq:45} that for~ $[V]\in X$
and~$g\in G$ we define a new representation~$V^g$ on the vector space~$V$.
Its equivalence class, denoted~$[V]\cdot \pi (g)\in X$, is represented by one
of our chosen representatives $\rep W\:G'\to\Aut(W)$.  By Schur's lemma
  \begin{equation}\label{eq:49}
     \Kl Vg:=\Hom_{G'}(W,V^g) 
  \end{equation}
is a line: the space of intertwiners between isomorphic irreducible complex
representations is one-dimensional.  Composition of intertwiners defines an
isomorphism
  \begin{equation}\label{eq:50} 
  \begin{aligned}
     \Kl V{g_1}\otimes \tiL\mstrut _{[V]\cdot \pi (g_1),g_2}&\longrightarrow \Kl
     V{g_1g_2}, \qquad g_1,g_2\in G\\ 
      f_1\otimes f_2 &\longmapsto f_1\circ f_2
  \end{aligned}
  \end{equation}
and the associative law (groupoid analog of~\eqref{eq:80}) is satisfied since
composition of functions is associative.  Now if~$g'\in G'$ then $[V]\cdot
\pi (g')=[V]$ and there is a canonical nonzero element of~\eqref{eq:49},
namely the map~$\rep V(g')$.  So $\Kl V{g'}$ is canonically trivial
for~$g'\in G'$, and taking~$g_1\in G'$ in~\eqref{eq:50} we see that the
line~$\Kl Vg$ depends up to canonical isomorphism only on~$\pi (g)\in G''$.
This is the descent mentioned above.  More precisely, for~$g''\in G''$
define~$\Ll V{g''}$ as the space of sections~$s$ of the map
  \begin{equation}\label{eq:51}
     \bigcup\limits_{g\in \pi \inv (g'')}\Kl Vg\longrightarrow \pi \inv
     (g'') 
  \end{equation}
such that $\rep V(g')\otimes s(g)$ and~$s(g'g)$ correspond
under~\eqref{eq:50} for all~$g\in \pi \inv (g'')$, $g'\in G'$.  In other
words, $s(g)\in \Kl Vg$ and the set of~$s$ which satisfy the equivariance
condition is a 1-dimensional vector space since $s$~is determined by its
value at any~$g\in \pi \inv (g'')$.  Then \eqref{eq:50}~induces isomorphisms
  \begin{equation}\label{eq:52}
     \Ll V{g_1''}\otimes L\mstrut _{[V]\cdot g_1'',g_2''}\longrightarrow \Ll
     V{g_1''g_2''}, \qquad g_1'',g_2''\in G'',
  \end{equation}
and so a central extension~$\nu $ of~$X\gpd G''$.
 
Let $\rep E\:G\to\Aut(E)$ be a representation and as in~\eqref{eq:42} define
a vector bundle $\sE\to X$ by
  \begin{equation}\label{eq:53}
     \sE_{[V]}=\Hom_{G'}(V,E). 
  \end{equation}
(We are still assuming $X$~is discrete, so any parametrized family of vector
spaces $\sE\to X$ is a vector bundle, possibly with infinite rank at some
points.)  If~$g\in G$ and $[W]=[V]\cdot \pi (g)$ in~$X$, then there is a map
  \begin{equation}\label{eq:54}
     \begin{aligned} \Eb V\otimes \Kl Vg&\longrightarrow \Eb W \\ \varphi
      \otimes f&\longmapsto \rep E(g)\inv \circ \varphi \circ f\end{aligned} 
  \end{equation}
which intertwines the $G'$-actions and satisfies associativity in that the
diagram 
\definecolor{labelkey}{rgb}{1,1,1}
  \begin{equation}\label{eq:55}
     \xymatrix{\Eb V\otimes \Kl V{g_1}\otimes \tiL\mstrut _{[V]\cdot \pi
     (g_1),g_2} \ar[d]_{\eqref{eq:54}}\ar[r]^(.6){\eqref{eq:52}} & \Eb V\otimes
     \Kl V{g_1g_2}\ar[d]^{\eqref{eq:54}}\\ \sE\mstrut 
     _{[V]\cdot \pi (g_1)}\otimes \tiL\mstrut _{[V]\cdot \pi
     (g_1),g_2}\ar[r]^(.6){\eqref{eq:54}} 
     &\sE_{[V]\cdot \pi (g_1g_2)} }\definecolor{labelkey}{rgb}{1,0,0}
  \end{equation}
commutes.  Apply~\eqref{eq:55} to $g_1\in G'$ to conclude that under the
descent described around~\eqref{eq:51} the action~\eqref{eq:54} descends to an
action
\definecolor{labelkey}{rgb}{1,0,0}
  \begin{equation}\label{eq:56}
     \Eb V\otimes \Ll Vg\longrightarrow \Eb W 
  \end{equation}
Therefore, we obtain a $\nu $-twisted vector bundle $\sE\to X\gpd G''$. 
 
Assume now that $G$~is compact, so every irreducible representation~$V$ is
finite dimensional and rigid: the compactness of~$G'\subset G$ implies that
$X$~is discrete.  If $\sE\to X$ is a $\nu $-twisted vector bundle with
compact---hence finite---support, then \eqref{eq:43}~defines a finite
dimensional representation $\rho _E\:G\to\Aut(E)$ on
  \begin{equation}\label{eq:140}
     E = \bigoplus \limits_{[V]\in X}\sE_{[V]}\otimes V 
  \end{equation}
as follows.  Fix~$g\in G$ and $[V]\in X$, set $g''=\pi (g)$, and suppose
$[V]\cdot g''\in X$ is represented by the distinguished representation $\rho
_W\:G'\to\Aut(W)$.  Define~$\rho _E$ to map $\sE_{[V]}\otimes V$ into
$\sE_{[W]}\otimes W$ by the composition
  \begin{equation}\label{eq:87}
     \Eb V\otimes V\xrightarrow{\id\otimes \eta \otimes \id}\Eb V\otimes \Kl
     Vg\otimes \Hom_{G'}(V^g,W)\otimes V\xrightarrow{\alpha \otimes
     \ev} \Eb W\otimes W. 
  \end{equation}
Here, recalling~\eqref{eq:49}, $\eta $~is the canonical isomorphism 
  \begin{equation}\label{eq:88}
     \CC\longrightarrow \Hom_{G'}(W,V^g)\otimes \Hom_{G'}(V^g,W) 
  \end{equation}
defined using the fact that $\Hom_{G'}(W,V^g)$ and~$\Hom_{G'}(V^g,W)$ are
dual lines.  Also, $\ev$~is the natural evaluation
  \begin{equation}\label{eq:89}
     \Hom_{G'}(V^g,W)\otimes V\longrightarrow W
  \end{equation}
and $\alpha $~is the structure map which defined $\sE\to X$ as a $\nu
$-twisted bundle.  The functors $E\mapsto\sE$ and $\sE\mapsto E$ are inverse
equivalences, from which the isomorphism $K\mstrut _G\cong K^\nu _{G''}(X)_c$
follows immediately.  This completes the proof of~(ii).
 
Assume $G'$~is abelian, but not necessarily compact.  Then every irreducible
representation is 1-dimensional and given by a character $G'\to\CC^\times $.
We choose the representative vector space in each isomorphism class of
irreducibles to be the trivial line~$\CC$.  Then $\Kl
Vg=\Hom_{G'}(\CC,\CC)=\CC$ in~\eqref{eq:49} is canonically trivial and the
composition~\eqref{eq:50} is the multiplication $\CC\otimes \CC\to\CC$.  If
$j\:G''\to G$ is a splitting of~\eqref{eq:44}, then we define $\Ll V{g''}=\Kl
V{j(g'')}=\CC$.  The cocycle map~\eqref{eq:52} is again multiplication
$\CC\otimes \CC\to\CC$, so under these identifications $\nu $~is the trivial
central extension.  This proves~(iii).

It remains to prove~(i) when $G'$~is a lattice (finitely generated free
abelian group) and $X$~the Pontrjagin dual abelian group of unitary
characters~$\Hom(G',\TT)$.  Then $X$~is naturally topologized as a smooth
torus.  A unitary character $\rho _V\:G'\to\TT\subset \Aut(\CC)$ acts
on~$V=\CC$ by multiplication.  So we can take the lines~$\Kl Vg$
in~\eqref{eq:49} to be the trivial line~$\CC$, as in the previous paragraph.
However, the descent described in~\eqref{eq:51} may be nontrivial and so $\Ll
V{g''}$ is not naturally trivialized.  If $\rho _E\:G\to\Aut(E)$ is a unitary
representation on a Hilbert space~$E$, then the spectral theorem gives a
self-adjoint projection-valued measure~$\mu _E$ on~$X$.  If $U\subset X$ is
an open set, define $\sS_U=\image \mu _E(U)\subset E$; it is a closed
subspace of~$E$.  The proof of~(i) is completed by the following lemma.

  \begin{lemma}[]\label{thm:83}
 The assignment $U\mapsto\sS_U$ is a sheaf on~$X$. 
  \end{lemma}

  \begin{proof}
 If $U'\subset U$ then $\sS_{U'}\subset \sS_U$.  Define the restriction map
$\sS_U\to\sS_{U'}$ to be orthogonal projection.  This is obviously a
presheaf.  To verify the sheaf property suppose $U_1,U_2\subset X$ and
$e_1\in \sS_{U_1},\,e_2\in \sS_{U_2}$ have equal orthogonal projections~$\be$
in~$\sS_{U_1\cap U_2}$.  Then $e_1+e_2-\be\in \sS_{U_1\cup U_2}$ is the
unique vector which projects orthogonally onto~$e_1\in \sS_{U_1}$ and $e_2\in
\sS_{U_2}$. 
  \end{proof}  \end{proof}

  \subsection*{A nontrivial example}

Assume $G'$~is abelian and take $X$~to be the space of unitary characters
$\lambda \:G'\to\TT$.  Fix~$\lambda \in X$ and let $G''(\lambda )\subset G''$
be the subgroup which stabilizes~$\lambda $.  It may be regarded as a
subgroupoid  
  \begin{equation}\label{eq:105}
     \{\lambda \}\gpd G''(\lambda )\hookrightarrow X\gpd G''. 
  \end{equation}

  \begin{proposition}[]\label{thm:103}
 The restriction of the central extension~$\nu $ in Theorem~\ref{thm:16} to
$\{\lambda \}\gpd G''(\lambda )$ is the central extension
  \begin{equation}\label{eq:106}
     1\longrightarrow \TT\longrightarrow H\longrightarrow G''(\lambda
     )\longrightarrow 1 
  \end{equation}
defined as the associated extension (Definition~\ref{thm:106}) via $\lambda
\:G'\to\TT$ of the pullback of~\eqref{eq:44} by the inclusion $G''(\lambda
)\hookrightarrow G''$:
  \begin{equation}\label{eq:130}
     \xymatrix{ 1\ar[r] & \TT\ar[r] & H\ar[r] & G''(\lambda )\ar[r]&1\\
     1\ar[r] & G'\ar[u]^\lambda \ar[r]\ar@{=}[d] & G(\lambda
     )\ar[u]\ar[r]\ar@{^{(}->}[d] & G''(\lambda
     )\ar@{=}[u]\ar[r]\ar@{^{(}->}[d] & 1\\ 
     1\ar[r]&G'\ar[r] & G\ar[r] & G''\ar[r] & 1} 
  \end{equation}
  \end{proposition}

  \begin{proof}
 As in the proof of Theorem~\ref{thm:16} we take the lines~$\Kl\lambda
g\;(\lambda \in X,\,g\in G)$ to be trivial, and so \eqref{eq:50}~is
multiplication.  Assume $g''\in G''(\lambda )$.  Then the set of equivariant
sections of~\eqref{eq:51} is identified with the set of complex-valued
functions $f\:\pi \inv (g'')\to\CC$ such that $f(g'g) = \lambda (g')f(g)$ for
all $g\in \pi \inv (g'')\,,g'\in G'$.  This is precisely the line at~$g''$
associated to the central extension~\eqref{eq:106}.
  \end{proof}

  \begin{example}[]\label{thm:104}
 We resume Example~\ref{thm:97}, which concerns the nonsplit group
extension~\eqref{eq:102}.  We show that the associated central
extension~$\nu $ of Theorem~\ref{thm:16} is not trivializable.  The
characters of~$\Pi $ form the Pontrjagin dual torus~$X$.  We identify it with
the standard torus~$\TT^2$ by letting $(\lambda _1,\lambda _2)\in \TT^2$ act
as the character
  \begin{equation}\label{eq:107}
     \lambda \:(n ^1,n ^2)\longmapsto \lambda _1^{n ^1}\lambda _2^{n
     ^2},\qquad n ^1,n ^2\in \ZZ. 
  \end{equation}
The actions of the reflections which generate the quotient $\zt\times \zt$
are 
  \begin{equation}\label{eq:108}
     \begin{aligned} g_1''\:(\lambda _1,\lambda _2)&\longmapsto (\lambda
      _1\inv ,\lambda _2) \\ g_2''\:(\lambda _1,\lambda _2)&\longmapsto
      (\lambda _1,\lambda _2\inv ) \\ \end{aligned} 
  \end{equation}
The character $\lambda =(-1,-1)$ is fixed by both~$g_1''$ and~$g_2''$.  By
Proposition~\ref{thm:103} the restriction of~$\nu $ to $\{\lambda
\}\gpd(\zt\times \zt)$ is the associated extension
  \begin{equation}\label{eq:109}
     \xymatrix{1 \ar[r]& \Pi \ar[r] \ar[d]^\lambda  & G \ar[d]
     \ar[r]& \zt\times \zt \ar[r]\ar@{=}[d] & 1\\ 1 \ar[r]
     &\TT 
     \ar[r] &\tilG \ar[r] &\zt\times \zt\ar[r]&1} 
  \end{equation}
In~$G$ the commutator of lifts of $g_1'',g_2''$ is $(0,1)\in \Pi $, as
remarked at the end of Example~\ref{thm:97}.  Therefore, the commutator of
lifts of $g_1'',g_2''$ to~$\tilG$ is $\lambda (0,1)=-1$.  The following lemma
shows that the central extension $1\to\TT\to\tilG\to\zt\times \zt$ is
\emph{not} split.  It follows that $\nu $~ is not trivializable. 
  \end{example}

  \begin{lemma}[]\label{thm:105}
 Let 
  \begin{equation}\label{eq:110}
     1\longrightarrow \TT\longrightarrow \tilG\longrightarrow \zt\times
     \zt\longrightarrow 1 
  \end{equation}
be a central extension, and $g_1,g_2\in \tilG$ lifts of the generators of the
quotient.  Then $g_2\inv g_1\inv g\mstrut _2g\mstrut _1 = \pm1$ and
\eqref{eq:110}~is split if and only if the sign is~$+$. 
  \end{lemma}

\noindent 
 The commutator is independent of the lifts since the extension is central.

  \begin{proof}
 By Lemma~\ref{thm:11}(i) the extension~\eqref{eq:110} is split over each
factor of the quotient, so we can choose lifts~$g_1,g_2$ such that
$g_1^2=g_2^2=1$.  Then if $g_1g_2=\mu g_2g_1$ for some~$\mu \in \TT$, 
  \begin{equation}\label{eq:111}
     g\mstrut _2=g_1^2g\mstrut _2 = \mu g\mstrut _1g\mstrut _2g\mstrut _1 =
     \mu ^2g\mstrut _2g_1^2 = \mu ^2g\mstrut _2, 
  \end{equation}
from which $\mu ^2=\pm1$, as claimed.  If $\mu =+1$ then $g_1,g_2$~generate a
splitting of $\zt\times \zt$ in~$\tilG$.  Conversely, a splitting produces
lifts with $g_1^2=g_2^2=1$ and $g_1g_2=g_2g_1$, so the sign in the commutator
is~$+$. 
  \end{proof}

  \subsection*{A generalization}

For the application in the next section we need to extend
Theorem~\ref{thm:16} to certain extended \Qsymmetry classes based on~$G$.
Namely, we assume given a homomorphism $\phi \:G''\to\pmo$, which we pull
back to~$G$ using $\pi \:G\to G''$ in~\eqref{eq:44}, and a $(\phi \circ \pi
)$-twisted
 extension
  \begin{equation}\label{eq:68}
     1\longrightarrow \TT\longrightarrow \tG\longrightarrow G\longrightarrow
     1 
  \end{equation}
Also, let $c\:G''\to\pmo$ be a homomorphism, which is pulled back to~$G$
using~$\pi $.  Overload the symbols~`$\phi $' and~`$c$' by using them to
denote the composition with~$\pi $ as well.  We are interested in
$\ptc$-twisted representations of~$G$, as in Definition~\ref{thm:84}(ii), so
sketch here the modifications in the previous discussion necessary to account
for~$\phi $, $\tau $, and~$c$.

First, let $1\to\TT\to\tGt\to G'\to 1$ be the restriction of~$\tau $ to~$G'$.
Notice that this is a \emph{central} extension (untwisted).  Let $\Xt$ denote
the space of isomorphism classes of $\tau $-twisted irreducible
representations of~$G'$, i.e., complex linear irreducible representations
of~$\tGt$ on which $\TT\subset \tGt$ acts by scalar multiplication.  We
define an action of~$G''$ on~$\Xt$.  Let $g''\in G''$ and suppose first that
$\phi (g'')=1$.  Let $[V]\in \Xt$ be represented by $\rep V\:\tGt\to\Aut(V)$
and choose $g\in G$ such that $\pi (g)=g''$.  Choose a lift~$\tg\in \tG$
of~$g\in G$ and observe that $\alpha (g)\tg' = \tg\tg'\tg\inv ,\;\tg'\in
\tGt,$ is well-defined independent of the lift~$\tg$, since any two lifts
differ by an element of the center~$\TT$.  Then replace~$G'$ by~$\tGt$
in~\eqref{eq:45} to define $[V]\cdot g''\in \Xt$.  If instead $\phi (g'')=-1$
modify the last step and take $[V]\cdot g''$ to be the isomorphism class of
the complex conjugate representation~$\bVg$ to~\eqref{eq:45}.  

  \begin{theorem}[]\label{thm:51}
 Let \eqref{eq:44} be an extension of Lie groups, $\phi ,c\:G''\to\pmo$
homomorphisms, and \eqref{eq:68}~a $(\phi \circ \pi )$-twisted extension.
Let~$\Xt$ be the space of isomorphism classes of irreducible $\tau $-twisted
representations of~$G'$.  Assume $G''$~is compact and $G'$~is either compact
or a lattice.
  \begin{enumerate}
 \item There is an induced \ptzgce~$\nu $ of the groupoid~$\Xt\gpd G''$ such
that a $(\phi \circ \pi ,\tau ,c \circ \pi )$-twisted representation~$E$
of~$G$ induces a $\nu $-twisted $G''$-equivariant sheaf over~ $\Xt$.
(Infinite dimensional representations are assumed unitary.)

 \item If $G$~is compact, there is an isomorphism $\hKp\phi _G^{\tau,c}
\to K^\nu _{G''}(\Xt)_\cpt$. 

 \item If $(G')^\tau $~is abelian, then a splitting of the group extension
  \begin{equation}\label{eq:123}
     1\longrightarrow (G')^\tau \longrightarrow \tG\xrightarrow{\;\pi \;}
     G''\longrightarrow  1  
  \end{equation}
induces an isomorphism of \ptzgces\ $\nu(\phi ,c) \xrightarrow{\cong }\nu $,
where $\nu (\phi ,c)$~is the \ptzgce\ in Remark~\ref{thm:119}.

  \end{enumerate}

   \end{theorem}

  \begin{proof}
 The proof follows that of Theorem~\ref{thm:16} with a few modifications.
Throughout replace~$G'$ by~$\tGt$ and~$X$ by~$\Xt$.  If~$g\in G$ and $(\phi
\circ \pi )(g)=-1$, then replace~\eqref{eq:49} with 
  \begin{equation}\label{eq:124}
      \Kl Vg:=\Hom_{\tGt}(\overline{W},V^g) ,
  \end{equation}
where $V^g$ is defined as before by~\eqref{eq:45}.   The descent to
lines~$\Ll V{g''}$ follows the argument around~\eqref{eq:51}.  The data of the
\ptzgce~$\nu $ consists of the lines~$\Ll V{g''}$ and the homomorphisms
$(\Xt\gpd G'')_1\to\pmo$ obtained by composing~$\phi ,c$ with the projection
$(\Xt\gpd G'')_1=\Xt\times G''\to G''$.
 
For~(iii) it helps to organize the data in the diagram 
  \begin{equation}\label{eq:125}
     \xymatrix{ &\TT\ar@{=}[r]\ar[d]&\TT\ar[d] \\ 1 \ar[r]&(G')^\tau
     \ar[r]\ar[d] & \tG\ar[d]\ar[dr]^\pi  \\ 1 \ar[r]& G'\ar[r] & G \ar[r] &
     G''\ar[r]\ar@<.5ex>[d]^{\phi }\ar@<-.5ex>[d]_{c } & 1 \\ &&&\pmo } 
  \end{equation}
A splitting of~\eqref{eq:123}, which is a right inverse of~$\pi $, induces a
trivialization of the lines~$\Ll V{g''}$ if $G'$~is abelian, as in the
proof of Theorem~\ref{thm:16}(iii). 
  \end{proof}

   \section{Gapped topological insulators}\label{sec:6}

  \subsection*{Periodic systems of electrons}

We have in mind the following sort of quantum system.  Let $E$~be a Euclidean
space and let $\bC\subset E$ be the sites of a periodic system of
atoms.\footnote{A crystal (Definition~\ref{thm:31}) in spacetime consists of
the trajectory in spacetime of a subset of space, on which the only
restriction is invariance under a full lattice, whereas in this heuristic
introduction we find it more convenient to freeze time and also assume that
$\bC$~is a discrete set of points.}  Assume there is a full lattice~$\Pi $ of
translations of~$E$ which preserves~$\bC$.  A typical Hilbert space in this
situation is~$\sH=L^2(E;W)$, the space of $L^2$~functions on~$E$ with values
in a finite dimensional complex vector space~$W$.  It is the Hilbert space of
a single electron in the lattice. If we incorporate electron spin, then
$W$~is the 2-dimensional representation of~$SU(2)$.  Let $\XL$~be the abelian
group of characters $\lambda \:\Pi \to\TT$, the \emph{Brillouin torus}.
Fourier transform writes a function $f\:E\to W$ in terms of quasi-periodic
functions~$f_\lambda $, which satisfy
  \begin{equation}\label{eq:132}
     f_\lambda (x+\xi ) = \lambda (\xi )f_\lambda (x),\qquad \xi \in \Pi . 
  \end{equation}
In condensed matter physics \eqref{eq:132}~is called the \emph{Bloch wave
condition}, and one usually writes $\lambda (\xi )=e^{ik\cdot \xi }$ where
$k$ (not uniquely fixed by~$\lambda $) is the \emph{Bloch momentum}; see
Proposition~\ref{thm:115} for further discussion.  Note that quasi-periodic
functions are \emph{not} in~$L^2(E;W)$.  If $\lambda $~is the trivial
character~$\lambda\equiv 1$, then $f_\lambda $~is periodic; for
general~$\lambda $, the function~$f_\lambda $ shifts by a phase when the
argument is translated by an element of the lattice.  Periodic functions
on~$E$ are equivalently functions on the torus~$E/\Pi $---it now makes sense
to demand that they be~$L^2$---and quasi-periodic functions which
satisfy~\eqref{eq:132} are sections of a complex vector bundle\footnote{The
line bundle $\sL_\lambda \to E/\Pi $ is constructed in the proof of
Proposition~\ref{thm:115} } $\sL_\lambda\otimes W \to E/\Pi $.
Proposition~\ref{thm:115} is a precise formulation of the Fourier transform,
or Bloch sum, which expresses an $L^2$~function on~$E$ as an $L^2$~section of
a \emph{Hilbert bundle} $\sE\to \XL$ whose fiber at~$\lambda \in \XL$ is the
infinite dimensional Hilbert space $\sE_\lambda =L^2(E/\Pi ;\sL_\lambda \otimes W)$:
  \begin{equation}\label{eq:133}
     \sH=L^2(E;W)\cong L^2(\XL;\sE). 
  \end{equation}
For a noninteracting single electron the Hamiltonian~$H$ is typically a sum
  \begin{equation}\label{eq:134}
     H = \frac{\hbar^2}{2m} \triangle + u,
  \end{equation}
where $m$ is the mass of the electron, $\triangle$~is the Laplace operator on
Euclidean space, and $u = \sum_{\bar c \in \bar C} u_{\bar c}$ is a sum of
local potential energy functions, thought of as localized near each site.  We
assume that $H$~is invariant under~$\Pi $, and so under the
isomorphism~\eqref{eq:133} it corresponds to a family of
Hamiltonians~$\{H_\lambda \}$ parametrized by~$\lambda \in \XL$, which we
prove in Proposition~\ref{thm:150} vary continuously in~$\lambda $.  In fact,
$H_\lambda $~is a Laplace operator on $L^2(E/\Pi ;\sL_\lambda )$---an
elliptic operator on a compact manifold---so has discrete spectrum.  We
assume that the full Hamiltonian~$H$ has a gap in its spectrum.  Often one
fixes a particular \emph{Fermi level} in the gap.  We normalize the Fermi
level to be at zero energy.  In other words, $H$~is invertible with bounded
inverse.  It follows that each~$H_\lambda $ is also invertible, and so there
is a decomposition $\sE=\sEp\oplus \sEm$ of the Hilbert bundle as a sum of
``valence bands''~$\sEm$ and ``conduction bands''~$\sEp$ (see
Proposition~\ref{thm:149}).  This is the moment we restrict to insulators.
We further assume that there is a \emph{finite} set of valence
bands---$\sEm$~has finite rank---and therefore an \emph{infinite} set of
conduction bands---$\sEp$~has infinite rank.  Because of this asymmetry,
there are no Hamiltonian-reversing symmetries in this system
(Lemma~\ref{thm:90}).  More generally, we assume that the symmetry group~$G$
of the system is a Lie group containing~$\Pi $ as a normal subgroup such that
$G''=G/\Pi $ is compact.  Then the quotient~$G''$ acts on~$\XL$---the action
factors through a finite group---and the action lifts to~$\sE$, but the
lifted action is twisted, as explained in ~\S\ref{sec:5}.  There are further
possible \ptzgces\ of the groupoid $\XL\gpd G''$ as $G$~may have a nontrivial
\Qsymmetry type which governs its representation on the Hilbert space~$\sH$.

  \begin{remark}[]\label{thm:143} 
 Our considerations also apply to more realistic Hamiltonians than
~\eqref{eq:134}.  In particular, in three dimensions, we could take $W$ to be
the 2-dimensional spin representation and then include the physically
important spin-orbit term $\frac{1}{2m^2 c^2} \vec S \cdot (\nabla u \times
p)$, where $\vec S = \frac{\hbar}{2} \vec \sigma$ is the spin operator and we
have used unexplained notation instantly recognizable from the condensed
matter literature, see e.g.~\cite{BP}.  More generally, the entire series of
relativistic corrections can be included by taking $W$ to be the
four-dimensional Dirac representation and using the Dirac Hamiltonian $c \vec
\alpha \cdot \vec p + \beta m c^2 + u$, where
  \begin{equation}\label{eq:154}
     \vec \alpha = \begin{pmatrix} 0 & \vec \sigma \\ \vec \sigma & 0 \\
     \end{pmatrix} \qquad \beta =\begin{pmatrix} 1 & 0 \\ 0 & -1 \\
     \end{pmatrix} .
  \end{equation}
  \end{remark}

Another important physical example---a spin system---is an analog with a
discrete field supported on the sites in~$\bC$.  For this purpose we assume
that the $\Pi $-action on~$\bC$ has only finitely many orbits.  Let
$W_{\bc}$~be a finite dimensional complex vector space attached to each
site~$\bc\in \bC$, and assume the action of~$G$ on~$\bC$ is lifted to the
vector bundle $W\to\bC$.  The quantum Hilbert space~$\sH=L^2(\bC;W)$ is the
space of $L^2$~sections of $W\to\bC$.  Now the Fourier transform produces a
\emph{finite rank} hermitian vector bundle $\sE\to\XL$, and Fourier transform
produces an isomorphism $\sH=L^2(\bC;W)\cong L^2(\XL;\sE)$, as in
~\eqref{eq:133}.  So the Hamiltonians~$\{H_\lambda \}$ are self-adjoint
operators on finite dimensional vector spaces~$\sE_\lambda $.
 
We remark that sometimes a finite rank bundle $\sE\to\XL$ is also used to
focus on a subbundle of the infinite rank bundle in the first
example~\eqref{eq:134}.  This is done to model a finite number of bands.  (If
eigenvalues of~$H_\lambda $ do not cross as $\lambda $~varies, then each band
is a line bundle whose fiber at~$\lambda \in \XL$ is an eigenspace
of~$H_\lambda $, but of course the eigenvalues may cross and so $\sE$~is not
generally a sum of line bundles.)
 
With these examples in mind we craft Definition~\ref{thm:126} and
Hypothesis~\ref{thm:56} below.  We consider both finite rank (Type~F) and
infinite rank (Type~I) Hilbert bundles $\sE\to\XL$.

  \subsection*{Formal setup}

Let $(M,\Gamma )$ be a Galilean spacetime with a crystal~$C$
(Definition~\ref{thm:31}) or, more generally, a spin crystal~$C$
(Definition~\ref{thm:32}).  The symmetry group~$G(C)$ contains a normal
subgroup~$U$ of time translations; the quotient~$G(C)/U$ is a spacetime
crystallographic group, which is a group extension~\eqref{eq:6}.  The
kernel~$\Pi $ is a full lattice of spatial translations.  The quotient~$\hP$
is a finite group, the \emph{magnetic point group}, which is an
extension~\eqref{eq:7} of the point group~$P$ of orthogonal spatial
transformations by time-reversal symmetries.  More generally, we take~ $G$ to
be a Galilean symmetry group in the sense of Definition~\ref{thm:94},
modified to mod out by time translation, so equipped with a homomorphism
$\gamma \:G\to\Aut(M,\Gamma )/U$ which is split over the intersection
of~$\gamma (G)$ with the spatial translation subgroup~$V\subset \Aut(M,\Gamma
)$.  Because symmetries must preserve the crystal, we have $\gamma (G)\subset
G(C)/U$.  Assume $\Pi \subset \gamma (G)$, and then the splitting~\eqref{eq:99}
gives an inclusion~$\Pi \subset G$ as a normal subgroup.  Summarizing, the
symmetry group~$G$ fits into the diagram
  \begin{equation}\label{eq:67}
   \xymatrix{
     1\ar[r] & \Pi \ar@{=}[d]\ar[r] & G\ar[r]^\pi \ar[d]^\gamma  &
     G''\ar[r] \ar[d]^{\bar\gamma }& 1  \\ 
     1\ar[r] & \Pi \ar[r] & G(C)/U\ar[r] & \hP\ar[r] & 1
   }
  \end{equation} 
of group extensions.  As the Hamiltonian~$H$ accounts for time
translations~$U$, we do not include~$U$ as a subgroup of~$G$; see
Remark~\ref{thm:98}. 

  \begin{definition}[]\label{thm:126}
 A \emph{band insulator} consists of the following data:
  \begin{enumerate}
 \item a Galilean spacetime~$(M,\Gamma )$ with a spin crystal~$C\subset M$; 

 \item a Galilean symmetry group~$G$ as described in the previous paragraph; 

 \item an extended \Qsymmetry class $(G'',\phi ,\tau ,c)$; and

 \item a gapped system~$(\sH,H,\tar)$ with extended \Qsymmetry type $\Gptc$,
where $\Gptc$~is the pullback of~ $(G'',\phi ,\tau ,c)$ along
$G\xrightarrow{\pi } G''$.

  \end{enumerate}
  \end{definition}

\noindent
 Notice that the extension $\tG\to G$ splits over~$\Pi $, and we regard~$\Pi
$ as a subgroup of~$\tG$.  Let $\Lv=\Hom(\Pi ,\TT)$ be the compact abelian
Lie group of characters of~$\Pi $.  Restrict~$\tar\:\tG\to\QAut(\sH)$ to
$\tP\subset \tG$, and note that $\tar(\tP)$~is an abelian group of operators
on~$\sH$ which commute with the Hamiltonian~$H$, a self-adjoint operator
on~$\sH$.  The spectral theorem simultaneously diagonalizes the operators
in~$\tar(\tP)$.  This is encoded in a projection-valued measure~$\mes$
on~$\tXL$ whose value on a Borel subset~$U\subset \tXL$ is orthogonal
projection~$\mes(U)$ onto---very roughly speaking---the subspace
$\mes(U)(\sH)\subset \sH$ of vectors which transform under a character
of~$\Pi $ which lies in~$U$.  Of course, this description only works for the
discrete part of the spectrum, and as we assume below that the spectrum is
continuous, vectors in~ $\mes(U)(\sH)$ are formally a smearing of
non-existent eigenfunctions.  In case~$\sH=L^2(E;W)$, described in the
introductory subsection above, quasi-periodic functions~\eqref{eq:132} are
not~$L^2$, so must be smeared to obtain vectors in~$\sH$.  The assignment
$U\mapsto \mes(U)(\sH)$ is a \emph{sheaf}~$\SH$ of Hilbert spaces on~$\tXL$;
see Lemma~\ref{thm:83}.

  \begin{remark}[]\label{thm:46}
 We elaborate on the remarks in Example~\ref{thm:123}.  If $\sH$~is an
irreducible representation of~$\tP$---that is, a particular character
$\lambda \:\tP\to\TT$ acting on a one-dimensional Hilbert space~$\sH$---then
$\sS_{\sH}$~is the skyscraper sheaf at~$\lambda \in \tXL$.  If $\sH$~is an
irreducible representation of~$\tG$, then $\sH$~is supported on an orbit of
the $G''$-action on~$\tXL$.  This picture is familiar in the
Wigner-Bargmann-Mackey theory of representations of the Poincar\'e group, for
example, where the sheaf constructed from the translations in Minkowski
spacetime is supported on an orbit of the Lorentz group on the characters of
the translation group (mass shell).  Hypothesis~\ref{thm:56}(ii) below
imposes a regularity which is very far from irreducibility.
  \end{remark}

To define the monoid $\TP\Gptc$ of topological phases, we need to impose
finiteness conditions on the symmetry group~$G$ and on the representation.
Recall that the \emph{Haar measure} on the compact abelian Lie group~$\Lv$ is
the unique translation-invariant measure with prescribed total volume.  

  \begin{hypothesis}[]\label{thm:56}
 \ 
  \begin{enumerate}
 \item $G''$~ is a compact Lie group.  

 \item The spectral measure~$\mes$ satisfies the following regularity
property with respect to Haar measure: there is a $\tG$-equivariant Hilbert
bundle $\sE\to\tXL$ and a $\tG$-equivariant isomorphism
  \begin{equation}\label{eq:150}
     \sH\xrightarrow{\;\;\cong \;\;}L^2(\tXL,\sE) 
  \end{equation}
under which $\Pi $~acts by multiplication on each fiber; then $\sS_{\sH}$~is
the sheaf of sections of the Hilbert bundle $\sE\to\tXL$.

 \item The Hamiltonian~$H$ induces a \emph{continuous} family of
self-adjoint operators~$\{H_\lambda \}$ on the fibers of~$\sE\to\tXL$.

 \item The invertible Hamiltonian~$H$ induces a decomposition $\sE\cong
\sEp\oplus \sEm$ such that $L^2(\tXL,\sEp)$ is the subspace of~$\sH$ on which
$H$~is positive and $L^2(\tXL,\sEm)$~is the subspace of~$\sH$ on which $H$~is
negative.  We assume $\sEm$~has \emph{finite} rank.
 
 \item We make two distinct hypotheses on~$\sEp$:  
  \begin{equation*}\label{eq:135}
     \begin{aligned} \textnormal{Type~F:}&\textnormal{ $\sEp$ has finite
      rank} \\ \textnormal{Type~I:}&\textnormal{ $\sEp$ has infinite rank} \\
      \end{aligned} 
  \end{equation*}
  \end{enumerate}
  \end{hypothesis}

\noindent 
 Statement~(i) requires no comment.  In~(ii) the action of~$\xi \in \Pi $ on
$\psi \in L^2(\tXL,\sE)$ is $(\xi \cdot \psi )(\lambda )=\lambda (\xi )\psi
(\lambda )$ for $\lambda \in \tXL$.  The $\tG$-action on $\sE\to\tXL$
descends to a twisted $G''$-action---better said, $\sE$~is a twisted bundle
for a \ptzgce\ (Definition~\ref{thm:85}(iii)) of the groupoid $\tXL\gpd
G''$---\emph{and the \ptzgce\ is precisely the \ptzgce~$\nu $ of
Theorem~\ref{thm:51}(i)}.  Because $H$~commutes with the action of~$\Pi $, it
preserves the fibers of~$\sE\to\tXL$; the import of~(iii) is the continuity
of the family of operators~$\{H_\lambda \}$ in the sense of
Definition~\ref{thm:112}.  The decomposition $\sE\cong \sEp\oplus \sEm$ then
follows from Proposition~\ref{thm:149}.  The import of~(iv) is the finite
rank condition on~$\sEm$.  Finally, notice that in Type~I we necessarily have
$c\equiv 1$, since an odd symmetry would be an isomorphism between the
infinite rank bundle~$\sEp$ and the finite rank bundle~$\sEm$.

In Appendix~\ref{sec:14} we prove that the periodic system of electrons
discussed in the introductory subsection satisfies these hypotheses.

We can now compute the commutative monoid\footnote{For Type~I the
set~$\TPI\Gptc$ is not a monoid since there is no zero element---$\sEp$~has
infinite rank and so the zero Hilbert space~$\sH$ is not allowed.  We fix
this by formally adjoining a zero element to~$\TPI\Gptc$.}~$\TP\Gptc$ of
topological phases of band insulators.  We will obtain an explicit answer in
Type~F and only a partial answer in Type~I; in both cases we will find a
simple answer after passing to an abelian group~$\GS\Gptc$.  Recall that
Definition~\ref{thm:26} of topological phase relies on the definition of
continuous families of gapped systems.  The discussion in
Appendix~\ref{sec:14}, especially Definition~\ref{thm:113}, provides the
necessary background.

  \subsection*{Topological phases: Type~F}

We show that the isomorphism classes of~$\sEp,\sEm$ determine the topological
phase. 
As in Definition~\ref{thm:71} let $\VXT$ denote the commutative monoid of
isomorphism classes of $\zt$-graded finite rank $\nu $-twisted
$G''$-equivariant vector bundles over~$\tXL$.

  \begin{theorem}[]\label{thm:129}
 For band insulators of Type~F the map 
  \begin{equation}\label{eq:138}
     \begin{aligned} \TPF\Gptc &\longrightarrow \VXT \\
      (\sH,H,\tar)&\longmapsto \sEp\oplus \sEm\end{aligned} 
  \end{equation}
is an isomorphism of commutative monoids.  
  \end{theorem}

  \begin{proof}
 To begin we prove that \eqref{eq:138}~is well-defined.  First, after a
homotopy we may assume that $H$~is the grading operator; see
Lemma~\ref{thm:27}.  In any case $\sEp\oplus \sEm$ is defined using the
positive and negative spectral projections, which only depend on the grading
operator associated to~$H$.  A homotopy of Type~F band insulators is a
Hilbert bundle $\sH\to \Delta ^1\times \tXL$ together with a continuous
family of Hamiltonians~$H_t$ and representations~$\tar_t$, $t\in \Delta ^1$.
(See Definition~\ref{thm:113}.)  The (Bloch) Hamiltonians~$H_{t ,\lambda}$ ,
$(t,\lambda )\in \Delta ^1\times \tXL$ on the fibers of $\sE\to\Delta
^1\times \tXL$ form a continuous family, by Hypothesis~\ref{thm:56}(iii), and
the positive and negative spectral projections are obviously continuous since
$\sE$~has finite rank.  Now because $G''$~and $\tXL$~are compact, isomorphism
classes of $\nu $-twisted finite rank $G''$-equivariant super vector bundles
over~$\tXL$ are rigid---they have no continuous deformations---and so bundles
which are homotopic are in fact isomorphic.  It follows that
\eqref{eq:138}~is well-defined.
 
To prove that \eqref{eq:138}~is surjective, suppose given a $\nu $-twisted
bundle $\sE=\sEp\oplus \sEm\to\tXL$.  Define 
  \begin{equation}\label{eq:139}
     \sH=L^2(\tXL;\sE) = L^2(\tXL;\sEp)\oplus L^2(\tXL;\sEm)
  \end{equation}
and let $H$~be the grading operator.  The representation~$\tar$ is
constructed using the map~\eqref{eq:87} (which simplifies since $\Pi $~is
abelian).   
 
To prove that \eqref{eq:138}~is injective, suppose
$(\sH_i,H_i,\tar_i),\;i=0,1$ are band insulators which map to isomorphic $\nu
$-twisted bundles $\sE_i=\sEp_i\oplus \sEm_i$.  Then an isomorphism $\theta
\:\sE_0\to\sE_1$ induces an isomorphism $\sH_0\to\sH_1$, since
$\sH_i\cong L^2(\tXL,\sE_i)$.  The grading operators correspond, since $\theta
$~is an isomorphism of $\zt$-graded bundles, and so too do the
representations~$\tar_i$, since $\theta $~is an isomorphism of (twisted)
equivariant bundles.
  \end{proof}

To pass to an abelian group we impose that $\sEp\oplus \sEm$ be the
difference $\sEp-\sEm$.  This is convenient and leads to a computable
$K$-theory group.  Although it is used in the condensed matter literature, it
is not clear to us that this subtraction is well-motivated physically.

  \begin{definition}[]\label{thm:132}
 A band insulator ~$(\sH,H,\tar)$ of Type~F is \emph{topologically trivial}
if there exists an odd automorphism~$P$ of the associated $\nu $-twisted
equivariant bundle $\sEp\oplus \sEm$.  Define $\GSF\Gptc$ to be the quotient
of $\TPF\Gptc$ by the submonoid of isomorphism classes of topologically
trivial band insulators.
  \end{definition}

\noindent
 Definition~\ref{thm:71} immediately implies the following. 

  \begin{theorem}[]\label{thm:131}
 For band insulators of Type~F there is an isomorphism
  \begin{equation}\label{eq:141}
  \begin{aligned}
     \GSF\Gptc&\longrightarrow K^\nu _{G''}(\tXL)\\
     (\sH,H,\tar)&\longmapsto \sEp\oplus \sEm
     \end{aligned}
  \end{equation}
  \end{theorem}

  \subsection*{Topological phases: Type~I}

Let $(\sH,H,\tar)$~be a band insulator of Type~I.  As observed after
Hypothesis~\ref{thm:56} because $\sEp$~has infinite rank and $\sEm$~has
finite rank there are no odd symmetries which interchange them, whence
$c\equiv 1$.  Thus $\sEp$~and $\sEm$~are independent, and we will show that
they are separately invariants of~$(\sH,H,\tar)$ up to homotopy.  The
infinite rank bundle $\sEp\to\tXL$ is trivializable as a non-equivariant
bundle, although it may be nontrivial as an equivariant bundle, as we will
see in Example~\ref{thm:138} below.  We do not know a standard commutative
monoid in which to locate the isomorphism class of~$\sEp$.  The group
completion of this monoid is trivial, by an Eilenberg swindle.  For that
reason we only track the finite rank bundle~$\sEm$.

As a preliminary we prove that certain universal $\nu $-twisted Hilbert
bundles over~$\tXL$ exist.  The theorem we need is a small modification
of~\cite[\S3.3]{FHT1}, to which we refer for background and details.  A $\nu
$-twisted $G''$-equivariant Hilbert bundle $\sF\to\tXL$ is \emph{universal}
if for any $\nu $-twisted $G''$-equivariant Hilbert bundle $\sE\to\tXL$ there
exists an embedding $\sE\to\sF$.  By \cite[Lemma A.24]{FHT1} a universal
bundle~$\sF$ is \emph{absorbing} in the sense that $\sF\oplus \sE \cong \sF$
for any $\nu $-twisted $G''$-equivariant Hilbert bundle~$\sE$.  This last
property shows that $\sF$~has infinite rank.  We remark that the universality
property can be made stronger, so that it holds on open subsets as well, but
we do not need this here.

  \begin{lemma}[]\label{thm:134}
 There exists a universal, hence also absorbing, $\nu $-twisted
$G''$-equivariant Hilbert bundle $\sF\to\tXL$.
  \end{lemma}

\noindent 
 We remark that it is not \emph{a priori} obvious that \emph{any} nonzero
$\nu $-twisted $G''$-equivariant Hilbert bundle over~$\tXL$ exists.
Lemma~\ref{thm:134} holds for any \ptzgce~$\nu $.

  \begin{proof}
 Let $\gp=\tXL\gpd G''$.  Write $\nu =(\phi ,\sigma )$, where $\sigma $~is a
$\phi $-twisted extension of~$\gp$.  The homomorphism~$\phi $ on the arrows
of the groupoid $\gp$ determines a double cover $\tgp\to\gp$, and the $\phi
$-twisted extension~$\sigma $ of~$\gp$ pulls back to an untwisted central
extension~$\tnu$ of~$\tgp$.  We describe it explicitly (for general
groupoids).  Namely, $\tgp^0=\gp_0\times \zt$ and $\tgp^1=\gp_1\times \zt$:
an arrow $(x\xrightarrow{\gamma }x')\in \gp_1$ lifts to two arrows
$\bigl((x,\epsilon )\xrightarrow{(\gamma,\epsilon )}(x',\epsilon ') \bigr)$
which satisfy $\epsilon +\phi (\gamma )=\epsilon ' $.  The $\phi $-twisted
extension $\nL\to\gp_1$ lifts to an untwisted central extension
$\tnL\to\tgp^1$ with $\tnL_{(\gamma ,0)}=\nL_\gamma $ and $\tnL_{(\gamma
,1)}=\overline{\nL_\gamma }$.  We leave the reader to check that the
structure maps~\eqref{eq:82} for~$\sigma $ induce structure
maps~\eqref{eq:79} for~$\tilde{\sigma}$.  By \cite[Lemma~3.12]{FHT1} there
exists a universal $\tnu $-twisted Hilbert bundle ~$\sF'\to\tgp$.  Let
$\sigma \:\tgp\to\tgp$ be the involution defined by $(x,\epsilon )\mapsto
(x,\epsilon +1)$ and $(\gamma ,\epsilon )\mapsto(\gamma ,\epsilon +1)$.  Set
$\sF=\sF' \oplus \overline{\sigma ^*\sF'}$, a $\tilde{\sigma }$-twisted
Hilbert bundle over~$\tgp$.  For each $(x,\epsilon )\in \tgp$, adjoin a new
arrow $\bigl((x,\epsilon )\xrightarrow{\delta _{(x,\epsilon )}}(x,\epsilon
+1)\bigr)$.  Let $\gp'$~denote the groupoid~$\tgp$ with the arrow~$\delta
_{(x,\epsilon )}$ adjoined.  Define $\phi '\:(\gp')_1\to\pmo$ to be~0 on
arrows ~$(\gamma ,\epsilon )\in \tgp$ and 1~on arrows~$\delta _{(x,\epsilon
)}$.  Attach the trivial line to~$\delta _{(x,\epsilon )}$; these trivial
lines together with the lines~ $\tnL$ form a $\phi '$-twisted
extension~$\sigma '$ of~$\gp'$.  The groupoid~$\gp'$ with twisting~$(\phi
',\sigma ')$ is locally equivalent to the groupoid~$\gp$ with twisting~$\nu
=(\phi ,\sigma )$.  The Hilbert bundle $\sF\to\tgp$ extends to a $(\phi
',\sigma ')$-twisted Hilbert bundle over~$\gp'$, and an easy argument proves
that it is universal.
  \end{proof}

Let $\Vect^\nu _{G''}(\tXL)$ be the commutative monoid of equivalence classes of
finite rank $\nu $-twisted ungraded $G''$-equivariant vector bundles
over~$\tXL$.

  \begin{theorem}[]\label{thm:133}
 For band insulators of Type~I the map 
  \begin{equation}\label{eq:142}
     \begin{aligned} \TPI\Gpt &\longrightarrow \Vect^\nu _{G''}(\tXL) \\
      (\sH,H,\tar)&\longmapsto \sEm\end{aligned} 
  \end{equation}
is a surjective homomorphism of commutative monoids.  If $G''$~is trivial, then
\eqref{eq:142}~is an isomorphism.
  \end{theorem}

  \begin{proof}
 With one amplification, the argument in the proof of Theorem~\ref{thm:129}
proves that \eqref{eq:142}~is well-defined.  Namely, since $\sEp$~has
infinite rank, we need to use Proposition~\ref{thm:149} to prove that the
positive and negative spectral projections vary continuously.  For the
surjectivity, given $\sEm\to\tXL$ of finite rank, consider $\sF\oplus
\sEm\to\tXL$, where $\sF$~is the universal bundle guaranteed by
Lemma~\ref{thm:134}.  (Here we only use the existence of an infinite rank
bundle, not its universality.)  Set $\sH=L^2(\tXL;\sF)\oplus L^2(\tXL;\sEm)$
and let $H$~be the grading operator.  Construct a representation~$\tar$
of~$\tG$ using~\eqref{eq:87}.  Then~$\sEm$ is the image of~$(\sH,H,\tar)$ under
\eqref{eq:142}.  For the last statement we observe that an infinite rank
(nonequivariant) Hilbert bundle is trivializable; see the remark preceding
Definition~\ref{thm:111}. 
  \end{proof}

Our passage to an abelian group for Type~I is more natural than that for
Type~F.  

  \begin{definition}[]\label{thm:135}
 For band insulators of Type~I we define $\GSI\Gpt$ to be the
group completion of the commutative monoid~$\TPI\Gpt$. 
  \end{definition}

  \begin{theorem}[]\label{thm:136}
 For band insulators of Type~I there is an isomorphism
  \begin{equation}\label{eq:143}
  \begin{aligned}
     \GSI\Gpt&\longrightarrow K^\nu _{G''}(\tXL) \\ 
      (\sH,H,\tar)&\longmapsto \sEm 
  \end{aligned}
  \end{equation}
  \end{theorem}

  \begin{proof}
 The group completion of $\TPI\Gpt$ is the direct sum ~$\sA_1\oplus \sA_2$
of the group completion~$\sA_1$ of the commutative monoid of infinite rank $\nu
$-twisted $G''$-equivariant Hilbert bundles $\sEp\to\tXL$ and the group
completion~$\sA_2$ of the commutative monoid of finite rank $\nu $-twisted
$G''$-equivariant Hilbert bundles $\sEm\to\tXL$.  The latter is the twisted
equivariant $K$-theory group $K^\nu _{G''}(\tXL)$, by Remark~\ref{thm:137}.
We claim that $\sA_1=0$.  Let $\sF\to\tXL$ be a universal $\nu $-twisted
$G''$-equivariant Hilbert bundle.  We use it to pull off an Eilenberg swindle
(Remark~\ref{thm:125}): given any $\sEp\in \sA_1$ the absorption property
of~$\sF$ implies $\sF+\sEp=\sF$ in~$\sA_1$, and since $\sA_1$~is a group we
can cancel ~$\sF$ to obtain $\sEp=0$.
  \end{proof}

  \begin{example}[]\label{thm:138}
 The following simple example shows that the commutative monoid built from
bundles $\sEp\to\tXL$ may be nontrivial if $G''\not= \{1\}$.  Let space be
1-dimensional and suppose $G=\zt\ltimes\ZZ$~is the split extension
$1\to\ZZ\to G\to\zt\to1$, where $\zt$~acts on~$\ZZ$ by~$-1$.  This is a
possible symmetry group of a crystal in one spatial dimension, say the
subset~$\bC=\ZZ\subset \EE^1$ of the Euclidean line.  Suppose $\phi ,\tau
$~are trivial.  Now $\zt$~acts on the circle~$X=\TT$ of characters of~$\ZZ$
by $\lambda \mapsto \bar\lambda $, which has fixed points at~$\lambda =\pm1$.
Any $\zt$-equivariant vector bundle $\sEp\to\TT$ restricts at the two fixed
points to representations of the stabilizer group~$\zt$.  Thus at each point
we obtain ~$n^+,n^-\in \ZZ\cup \{\infty \}$ which tell the dimensions of the
subspaces on which $\zt$~acts trivially and by the sign representation,
respectively.  The monoid of infinite rank bundles is then the set of
quartets $(n^+_1,n^-_1,n^+_{-1},n^-_{-1})\in (\ZZ\cup\{\infty \})^{\times 4}$
such that the sums $n^+_1+n^-_1$ and $n^+_{-1}+n^-_{-1}$ are infinite.  Note
that equivariant line bundles which realize all four ways of ordering two~0s
and two~1s are easily constructed, and then taking finite direct sums we can
realize any such quartet.  The easy Eilenberg swindle which proves that the
group completion of this monoid of quartets is trivial is based
on~\eqref{eq:127}.  On the other hand, realistic physical examples seem to
have $n^+_1=n^-_1=n^+_{-1}=n^-_{-1}=\infty $.  For example, if
$\sH=L^2(\EE^1)$ then the fiber of~$\sE$ at~$\lambda =1$ is the Hilbert
space~$L^2(\EE^1/\ZZ)$ of periodic functions.  The Hamiltonian induces a
self-adjoint operator on periodic functions, and by assumption it commutes
with reflection on the circle~$\EE^1/\ZZ$, which is the action of the
generator of $\zt\subset G$.  Each eigenspace of the reflection action is
infinite dimensional, and since we assume the negative eigenspaces of the
Hamiltonian are finite dimensional, it follows that $n^+_1=n^-_1=\infty $.  A
similar argument applies to anti-periodic functions~$(\lambda =-1)$.
  \end{example}

  \subsection*{Simplifying assumptions; more familiar $K$-theory groups}

As at the end of~\S\ref{sec:11}, Theorem~\ref{thm:131} and
Theorem~\ref{thm:136} simplify if we assume certain splittings.  This is
not a physically valid hypothesis in general, as far as we know, but it does
hold in many examples.  Specifically, we assume Hypothesis~\ref{thm:41},
modified to account for the form~\eqref{eq:67} of~$G$ and, in the case of
Type~I, for the fact that $c\equiv 1$.  Let $\psi =(t,c) \:G''\to
\sC=\pmo\times \pmo$ and $G''_0=\ker\psi $.  Set $G_0=\pi \inv (G''_0)$,
where $\pi \:G\to G''$ is the projection.

  \begin{hypothesis}[]\label{thm:58}
 Let $A\subset \sC$ be the image of~$\psi $.
  \begin{enumerate}
 \item The group~$G''$ is a direct product $G''\cong A\times G''_0$ and under
this isomorphism $\psi $~is projection onto~$A$.  Furthermore,
\eqref{eq:67}~splits over $G''_0\subset G''$ and so $G\cong A\times G_0$.

 \item The restriction~$\tG_0\to G_0$ of $\tG\to G$ splits and we fix a
splitting.

  \end{enumerate}
Therefore, there is fixed an isomorphism $\tG\cong \tA\times G_0$.
  \end{hypothesis}

\noindent
  Elements of $\tA\subset \tG$ commute with~$\Pi$.  Thus if $a\in \tA$
satisfies~$\phi (a)=-1$---that is, $\tar(a)$~acts antilinearly in any $\tau
$-compatible representation~$\tar$---then $\tar(a)$~transforms the
character~$\lambda \in \XL$ to the inverse character~$\lambda \inv \in \XL$.
Let $\sigma \:\XL\to \XL$ denote the involution $\lambda \mapsto\lambda \inv
$. Note the group~$G''_0$ acts on~$\XL$ through the homomorphism $\bar{\gamma
}\:G''_0\to \hP$ to the magnetic point group~$\hP$ in~\eqref{eq:67}, and the
magnetic point group acts on the lattice~$\Pi $ through the
extension~\eqref{eq:6}.

  \begin{corollary}[]\label{thm:59}
  If Hypothesis~\ref{thm:58} holds, then we have the following table for the
twisted equivariant $K$-theory group in Theorem~\ref{thm:131} for
band insulators of Type~F, where $\nn$~is defined in~\eqref{eq:76} below:

 \bigskip
  \begin{center}
  \renewcommand{\arraystretch}{1.5}
 \ \hskip-3.3em
  \begin{tabular}{|c||c|c|c|c|c|}
\hline 
$A$&1&\textnormal{diag}&$\pmo\times \oo$&$\sC$&$\oo\times \pmo$\\
\hline
$T^2$&&&$+1$&$+1$&\\
\hline
$C^2$&&&&$-1$&$-1$\\
 \hline&&&&&\\[-1.6em] \hline
 $K^\nu _{G_0''}(\tXL)$&$K^{\nn}_{G_0''}(\XL)$&$K^{\nn-1}_{G_0''}(\XL)$&
$KR^{\nn}_{G_0''}(\XL)$&$KR^{\nn-1}_{G_0''}(\XL)$&
$KR^{\nn-2}_{G_0''}(\XL)$\\[1ex]
 \hline 
  \end{tabular}
  \renewcommand{\arraystretch}{1}
  \end{center}
\vskip12pt
  \begin{center}
  \renewcommand{\arraystretch}{1.5}
  \begin{tabular}{|c||c|c|c|c|c|c|}
\hline 
$A$&$\sC$&$\pmo\times \oo$&$\sC$&$\oo\times \pmo$&$\sC$ \\
\hline
$T^2$&$-1$&$-1$&$-1$&&$+1$ \\
\hline
$C^2$&$-1$&&$+1$&$+1$&$+1$\\
 \hline&&&&&\\[-1.6em] \hline
 $K^\nu
_{G_0''}(\XL)$&$KR^{\nn-3}_{G_0''}(\XL)$&$KR^{\nn-4}_{G_0''}(\XL)$&$KR^{\nn-5}_{G_0''}(\XL)$&
$KR^{\nn-6}_{G_0''}(\XL)$& $KR^{\nn-7}_{G_0''}(\XL)$\\[1ex]
 \hline 
  \end{tabular}
  \renewcommand{\arraystretch}{1}
  \end{center}
\vskip12pt 
  \end{corollary}

\noindent 
 An element of~$KR^q(\XL)$ is represented by a $\zt$-graded complex vector
bundle $E\to\XL$, a lift~$\tilde\sigma $ of the involution~$\sigma $ to an
antilinear involution of~$E$, and an action of the Clifford algebra~$\Cl q$
on the fibers of~$E$ which (graded) commutes with~$\tilde\sigma $.  In the
equivariant case~$KR^q_{G''_0}(\XL)$ there is an additional linear action
of~$G''_0$ on~$E$ which covers the action on~$\XL$ and commutes
with~$\tilde\sigma $ and the Clifford algebra action.  A refined version of
Corollary~\ref{thm:59} is an equivalence between the category of bundles
$E\to\XL$ of this type and the category of $\nu $-twisted $G''$-equivariant
$\zt$-graded bundles $\sE\to\XL$, where the chart tells the precise
correspondence.  The proof is modeled on that of
Corollary~\ref{thm:43}---including the Morita equivalences therein---but with
real structures replaced by~$\tilde\sigma $.  Note that both
elements~$\bT,\bC$ of~$A$, if present, act on~$\XL$ via the
involution~$\sigma $.  This is the familiar fact in band structure theory
that time reversal and ``particle-hole conjugation'' map $k\mapsto -k$, where
$k$~is the Bloch momentum.

  \begin{proof}
 Given a $\nu $-twisted $G''$-equivariant super vector bundle $\sEp\oplus
\sEm\to\XL$ we need to rewrite it as a $\zt$-graded twisted $G''_0$-bundle
and identify the new \ptzgce.  The \ptzgce~$\nn$ in each entry arises from
the extension
  \begin{equation}\label{eq:76}
     1\longrightarrow \Pi \longrightarrow G_0\longrightarrow
     G''_0\longrightarrow 1 ,
  \end{equation}
as in Theorem~\ref{thm:16}(i), and since $G_0$~commutes with~$\tA$ it simply
adds to the twisting we now derive from the action of~$\tA$.  The first two
entries---the complex case---are straightforward: use $S$~in the second entry
to define the action of a complex Clifford algebra~$\Clc{-1}$.  The three
columns with~$T^2=+1$ are also straightforward: set $\tilde\sigma =T$ and, if
it is present, use $C$~as a single Clifford generator.  In the two columns in
which $C$~is present and $T$~is absent set~$E=\sE\oplus \sigma
^*\overline{\sE}$.  There is an evident antilinear lift~$\tilde\sigma $
of~$\sigma $ to~$E$, and now we let the Clifford generators be~$\ts C,\ts
iC$.  If $T^2=-1$ and $\bC\in A$ then we again set $E=\sE\oplus \sigma
^*\overline{\sE}$ and use $\ts C,\ts iC, iTC$ as Clifford generators.  In the
remaining case in which $T^2=-1$ and $\bC\notin A$ set $E=(\sE\oplus \Pi
\sE)\oplus (\sigma ^*\overline{\sE}\oplus \Pi \sigma ^*\overline{\sE})$,
where $\Pi $~is parity-reversal (see Lemma~\ref{thm:128}).  The map~$\ts$ is
the evident one, and the action of~$\Cl{-4}$ is modeled after~\eqref{eq:73}:
  \begin{equation}\label{eq:164}
  \begin{aligned}
     e_1&=\begin{pmatrix} 0&-1&0&0\\ 1&0&0&0\\0&0&0&-1\\0&0&1&0
     \end{pmatrix},\qquad &&e_2=\begin{pmatrix} 0&i&0&0\\
     i&0&0&0\\0&0&0&i\\0&0&i&0 \end{pmatrix} \\ 
      e_3&=\begin{pmatrix}
     0&0&0&T\\ 0&0&T&0\\0&T&0&0\\T&0&0&0 \end{pmatrix},\qquad
     &&e_4=\begin{pmatrix} 0&0&0&iT\\ 0&0&iT&0\\0&iT&0&0\\iT&0&0&0
     \end{pmatrix} \end{aligned}
  \end{equation}
 
Conversely, for the inverse equivalence we begin with~$E$ and produce~$\sE$.
In the complex case, which comprises the first two columns of the table,
$\sE=E$.  For the columns with $KR$-theory in degrees~$q=0,\pm1$ (recall we
use the Morita equivalence which identifies degree~$-7$ and degree~$+1$) set
$\sE=E$ and let~$T=\ts$; if $q=\pm1$, then set~$C=e_1$.  For degrees~$q=\pm2$
let $\sE$~be the $(+i)$-eigenspace of~$e_1e_2$ and set~$C=\ts e_1$.  For
degrees~$q=\pm3$ let $\sE$~be the $(+i)$-eigenspace of~$\mp e_1e_2$ and
set~$C=\ts e_1$ and $T=-\ts e_2e_3$.  For degree~$q=-4$ let $\sE$~be the
simultaneous $(-1)$-eigenspace of~$e_1e_2e_3e_4$ and $(+i)$-eigenspace
of~$e_1e_2$.   Then $T=e_1e_3$ is even and squares to~$-1$.
  \end{proof}

A subset of these arguments suffices for Type~I.  In that case we work with
the bundle $\sEm\to\XL$.

  \begin{corollary}[]\label{thm:139}
   If Hypothesis~\ref{thm:58} holds, then we have the following table for the
twisted equivariant $K$-theory group in Theorem~\ref{thm:136} for
band insulators of Type~I:

\bigskip 
  \begin{center}
  \renewcommand{\arraystretch}{1.5}
  \begin{tabular}{|c||c|c|c|}
\hline 
$A$&1&$\pmo\times \oo$&$\pmo\times \oo$\\
\hline
$T^2$&$+1$&$+1$&$-1$ \\
 \hline&&&\\[-1.6em] \hline
 $K^{\nu}_{G_0''}(\tXL)$&$K^{\nn}_{G''_0}(\XL)$&$KR^{\nn}_{G''_0}(\XL)$&
$KR^{\nn-4}_{G''_0}(\XL)$\\[1ex]
 \hline 
  \end{tabular}

  \renewcommand{\arraystretch}{1}
  \end{center}
\vskip12pt

  \end{corollary}

   \section{Chern-Simons and Kane-Mele invariants}\label{sec:16}
 
To illustrate that the passage from~$\TP$ to~$\GS$ retains information of
physical interest, we consider two topological invariants with direct
physical significance.  These examples also illustrate our discussion
in~\S\ref{sec:6} in familiar examples from the contemporary condensed matter
literature.   
 
Both invariants involve the same basic data (Definition~\ref{thm:126}), which
we now specify.  The Galilean spacetime ~$\MM^{d+1}=\RR\times \EE^d$ is the
product of time and standard $d$-dimensional Euclidean space.  We do not
explicitly specify a crystal~$C\subset \MM^{d+1}$ as we only need a symmetry
group~$G$, which may be a subgroup of the complete symmetry group of a given
crystal.  The lattice subgroup of~$G$ is the standard $\Pi =\ZZ^d\subset
\RR^d$ acting by translations.  These are systems of Type~I, so $c\equiv 1$.
We assume in addition either a time-reversal, parity-reversal symmetry, or
both.  In the time-reversal case
  \begin{equation}\label{eq:156}
     G = \zt \times \Pi 
  \end{equation}
is a direct product with $\phi =t\:\zt\to\pmo$~nontrivial on the
generator~$\bT$ of~$\zt$.  There is a nontrivial $\phi $-twisted 
extension $\tG\to G$ given as the product of the nontrivial $\phi $-twisted
 extension of~$\zt$ in Lemma~\ref{thm:11}(ii) with the group~$\Pi $: if $T\in
\tG$~projects to~$\bT\in \zt$, then $T^2=-\id$ (because we are describing
spin~$1/2$ electrons).  Let $\Lv$~be the $d$-dimensional torus of characters
$\lambda \:\Pi \to\TT$.  Then $\bT$~acts on~$\Lv$ by the involution $\sigma
\:\lambda \mapsto\bar\lambda $, according to the text preceding
Theorem~\ref{thm:51}.  In the parity-reversing case the group~$G$ is the
semidirect product
  \begin{equation}\label{eq:155}
    G = \zt \ltimes\Pi
  \end{equation}
in which the generator $\bP\in \zt$~acts on~$\Pi=\ZZ^d$ by $\xi \mapsto -\xi
$.  The homomorphism $\phi =t$ is trivial.  The generator~$\bP\in \zt$ acts
on~$\Lv$ by the involution $\sigma \:\lambda \mapsto\bar\lambda $.  The
Hilbert space is a space of vector-valued functions on~$\EE^d$, as in the
introductory exposition in~\S\ref{sec:6}, and under Hypothesis~\ref{thm:56}
we obtain a finite rank complex vector bundle $\sEm\to\Lv$.  The extra twist
discussed in~\S\ref{sec:5} does not enter as in both~\eqref{eq:156}
and~\eqref{eq:155} the projection $G\to\zt$ is split; see
Theorem~\ref{thm:51}(iii).  So the nontrivial $\zt$~action on~$\Lv$ lifts to
a twisted action on~$\sEm$.  In the time-reversing case it is a projective
action (the square of the lift is~$-\id$) and in addition is antilinear; in
the parity-reversing case it is an honest linear $\zt$~action.  If both a
time-reversal and parity-reversal symmetry are present, then
  \begin{equation}\label{eq:168}
     G=(\zt\times \zt)\ltimes\Pi . 
  \end{equation}
In this case we have both a linear and antilinear lift of~$\sigma $; their
product is a quaternionic structure on~$\sEm$.  Let $G''=\zt$ or~$\zt\times
\zt$ be the extended point group in~\eqref{eq:156}, \eqref{eq:155},
or~\eqref{eq:168}.

We now go beyond~Hypothesis~\ref{thm:56} and use the fact that the Hilbert
space consists of $L^2$~functions on~$\EE^d$ and that the bundle~$\sEm\to\XL$
is \emph{smooth}; see the paragraph following Remark~\ref{thm:144}.  Then
Proposition~\ref{thm:141} implies that there is a covariant derivative
operator~$\nEm$ on $\sEm\to\Lv$, obtained by compression.  It is commonly
known as the \emph{Berry connection}.  Furthermore, it is invariant under the
lifted, possibly twisted, action of~$G''$.  This follows from
Remark~\ref{thm:144} and the fact that the Hamiltonian is assumed
$G$-invariant, whence spectral projection $\sE\to\sEm$ commutes with~$G$.
Contrary to the impression one gets from the condensed matter literature, we
do not need the specific Berry covariant derivative~$\nEm$ to construct the
two topological invariants discussed in this section.

  \begin{lemma}[]\label{thm:157}
 The set of $G''$-invariant covariant derivatives on $\sEm\to\XL$ is a
nonempty affine space. 
  \end{lemma}

  \begin{proof}
 The set of all covariant derivatives on $\sEm\to\XL$ is an affine space over
$\Omega ^1(\XL;\End\sEm)$.  Fix a covariant derivative~$\nabla $ and average
over~$G''$ to obtain a $G''$-invariant covariant derivative.  The difference
of any two is a 1-form invariant under~$G''$, so an element of the vector
space $\Omega ^1(\XL;\End\sEm)^{G''}$. 
  \end{proof}

\noindent
 Note that the $G''$-action may be twisted---it acts via the $\phi $-twisted
 extension~$(G'')^\tau $ of~$G''$---and if there are time-reversal symmetries
($\phi \not\equiv 1$) then there are antilinear transformations $\varphi
\:\sEm\to\sEm$.  The covariant derivative $\nabla $~is $\varphi $-invariant
if $\varphi ^*\nabla =\overline\nabla $, where $\overline\nabla $~is the
covariant derivative on the complex conjugate bundle.  In this case $\Omega
^1(\XL;\End\sEm)^{G''}$ is a \emph{real} vector space.

Corollary~\ref{thm:139} gives a simple expression for the twisted equivariant
$K$-theory group.  For the time-reversing case~\eqref{eq:156} we obtain the
$KR$-group 
  \begin{equation}\label{eq:203}
     \GSI\Gpt\cong KR^{-4}(\XL) 
  \end{equation}
For   the  parity-reversing  case~\eqref{eq:155}   we  have   the  untwisted
equivariant  $K$-theory  group  
  \begin{equation}\label{eq:204} 
     \GSI\Gpt\cong K^0_{\zt}(\XL). 
  \end{equation}
If both symmetries are present~\eqref{eq:168} we have 
  \begin{equation}\label{eq:205}
      \GSI\Gpt\cong KR^{-4}_{\zt}(\XL)\cong KO^{-4}_{\zt}(\XL). 
  \end{equation}
Recall  that in  this  case  elements are  represented  by complex  bundles
$\sEm\to\XL$ with two commuting lifts of the involution~$\sigma $, one linear
and  one antilinear,  which accounts  for the  $KR$-theory  description.  The
product  of  the  lifts  is  a  quaternionic  structure  on~$\sEm$.   In  the
$KO$-theory description of  the group we use this  quaternionic structure and
the commuting linear lift of~$\sigma $.

  \subsection*{Computations}

In this subsection we report on the computation of some relevant $K$-theory
groups.  Some of these were done in collaboration with Aaron Royer, and we
hope to write a separate paper in the near future in which we make more
computations and also provide details for the ones quoted here. 
 
Consider the $d$-dimensional torus $\tord=\cir\times \cdots\times \cir$ as an
equivariant space for the group $\ztd=\zt\times \cdots\times \zt$, where each
$\zt$~factor acts on the corresponding $\cir$~factor by reflection.  The
diagonal $\zt\subset \ztd$ is the involution $\sigma \:\lambda
\mapsto\bar\lambda $ which appears in the previous subsection.  Mike Freedman
brought the following theorem to our attention.

  \begin{theorem}[]\label{thm:179}
 The torus $\tord$ is equivariantly stably homotopy equivalent to a wedge of
spheres.
  \end{theorem}

\noindent
 In the cases~$d=2,3$ we have 
     \begin{align} \cir\times \cir &\sims\cir\vee\cir\vee S^2
     \label{eq:198}\\ \cir\times 
      \cir\times \cir &\sims\cir\vee\cir\vee\cir\vee S^2\vee S^2\vee S^2 \vee
      S^3  \label{eq:202}\end{align} 
In the first splitting the symmetry group $\zt\times \zt$ acts on the first
$\cir$~summand via reflection after projection to the first $\zt$~factor, on
the second $\cir$~summand via reflection after projection to the second
$\zt$~factor, and on the $S^2=I\times I\bigm/\partial (I\times I)$ summand by
the reflection action on each $I=[-1/2,1/2]$~factor.  The actions in higher
dimensions are similar, and easiest to see by writing $S^k=I^{\times
k}\bigm/\partial (I^{\times k})$.  The stable splitting in
Theorem~\ref{thm:179} is clearly also equivariant for any subgroup of~$\ztd$.
 
This stable splitting is the key step in the computations we report here.  An
alternative derivation, which provides an independent check, is via the
Kunneth spectral sequence for $RO(G)$-graded equivariant cohomology theories.
We also use the twisted Thom isomorphism theorem, which we illustrate with
a few specific cases.  First, for any compact Lie group~$G$ and $G$-space~$X$
we have
  \begin{equation}\label{eq:199}
  \begin{aligned}[2]
     &KO^{q}_{G\times \zt}(X\times \RR)_{\textnormal{cv}}&&\!\!\!\!\!\cong
     KO^{q}_G(X) \\
     &KO^{q}_{G\times \zt}(X\times \RR^2)_{\textnormal{cv}}&&\!\!\!\!\!\cong
     K^{q-2}_G(X) \\
     &KO^{q}_{G\times \zt}(X\times \RR^3)_{\textnormal{cv}}&&\!\!\!\!\!\cong
     KO^{q-4}_G(X) 
  \end{aligned}
  \end{equation}
where $q\in \ZZ$, the group $\zt$~acts on~$\RR^k$ by $\xi \mapsto -\xi $, and
`cv'~denotes compact support on~$\{x\}\times \RR^k$ for all~$x\in X$.  There
may be additional support conditions in the $X$-direction on both sides
of~\eqref{eq:199}.  The second twisted Thom isomorphism is, for any
space~$X$,
  \begin{equation}\label{eq:200}
     KR^{q }(X\times \RR^k)_{\textnormal{cv}}\cong
     KR^{q+k}(X)\cong KO^{q+k}(X).  
  \end{equation}
Here the involution which defines~$KR$ acts trivially on~$X$ and by $\xi
\mapsto -\xi $ on~$\RR^k$.  The compact support is in the $\RR^k$-direction. 
 
Recall that for any \emph{equivariant} cohomology theory~$h$, the
\emph{reduced cohomology}~$\tilde h(X)$ of a pointed space~$(X,x_0)$ is the
kernel of the restriction map $h(X)\to h(\{x_0\})$ to the basepoint.  In the
equivariant case the basepoint must be a fixed point of the group action.   
 
A typical picture of the torus~$\tord$ is as the quotient of~$I^{\times d}$
by identifying opposite faces of the boundary, and so there is a collapse map 
  \begin{equation}\label{eq:201}
     q\:\tord\longrightarrow S^d 
  \end{equation}
in which every point of the boundary~$\partial (I^{\times d})$ maps to the
basepoint of~$S^d$.  Stably, $q$~corresponds to collapse of all spheres
except~$S^d$ in the stable splitting of Theorem~\ref{thm:179}.

\newpage
  \begin{theorem}[]\label{thm:180}
 \ 

  \begin{enumerate}
 \item $KR^{-4}(\cir\times \cir)\cong \ZZ\times \zt$. 

 \item $KR^{-4}(\cir\times \cir\times \cir)\cong \ZZ\times (\zt)^{\times 
4}$.  

 \item The image of $\widetilde{KR}^{-4}(S^3)\xrightarrow{q^*}\widetilde{KR}
^{-4}(\cir\times \cir\times \cir)$ is cyclic of order two. 

 \item The image of
$\tKO^{-4}_{\zt}(S^3)\xrightarrow{q^*}\tKO_{\zt}^{-4}(\cir\times \cir\times
\cir)$ is infinite cyclic.

 \item The natural map $\tKO_{\zt}^{-4}(S^3)\to \widetilde{KR}^{-4}(S^3)$ is
reduction modulo two. 

  \end{enumerate}
  \end{theorem}

\noindent
 The involutions in~(i)--(v) are all simultaneous reflection on each $\cir$
~factor of the torus; on $S^3\subset \EE^4$ centered at the origin, the
involution is the composition of three orthogonal reflections in hyperplanes
through the origin (with two fixed points).  The latter may also be described
as the one-point compactification of the three-dimensional sign
representation of~$\zt$.  The $\zt$~factor in~(i) is due to the~$S^2$ in the
stable splitting~\eqref{eq:198}; the four $\zt$~factors in~(ii) are due to
the $S^2$'s and $S^3$ in the stable splitting~\eqref{eq:202}.  We will see in
the next subsections that the orbital magnetoelectric polarizability and
Kane-Mele invariant are detected on the subgroup
of~$KR^{-4}\bigl((\cir)^{\times 3} \bigr)$ identified in~(iii).  This is the
so-called ``strong'' subgroup; the three $\zt$~factors due to the
$S^2$~summands in~\eqref{eq:202} are the ``weak'' subgroups.  Assertions~(iv)
and~(v) imply that in systems with parity-reversal and time-reversal, as
in~\eqref{eq:168}, the aforementioned $\zt$-invariants have lifts to integer
invariants.  

Finally, we note that there are three projections 
  \begin{equation}\label{eq:206}
     \cir\times \cir\times \cir\longrightarrow \cir\times \cir, 
  \end{equation}
and it follows from the computations that these induce injections $\ZZ\times
\zt\hookrightarrow \ZZ\times (\zt)^{\times 4}$ on~ $KR^{-4}$.  The three
``weak'' $\zt$'s are in the image.

  \subsection*{Orbital magnetoelectric polarizability}

In this subsection~$d=3$.  We use the notation established in the first
subsection.  Fix an orientation of~$\XL$; the choice will be irrelevant to
the topological invariant~\eqref{eq:158} defined below.  A finite rank
complex vector bundle $\sEm\to\XL$ with covariant derivative~ $\nEm$ has a
\emph{Chern-Simons invariant}~\cite{F2}
  \begin{equation}\label{eq:157}
     \theta (\nEm)\in \RtpZ 
  \end{equation}
In~\cite{QHZ, EMV, ETMV} a topological contribution to the orbital
magnetoelectric polarizability is shown to be the Chern-Simons invariant of
the Berry connection.  It is a topological invariant in the presence of a
time-reversal or parity-reversal symmetry.  

  \begin{lemma}[]\label{thm:146}
 Let $\nabla $ be a $G''$-invariant covariant derivative on $\sEm\to\Lv$.
Then either $\theta (\nEm)=0$ or $\theta (\nEm)=\pi $.  Furthermore, $\theta
(\nabla )$~is independent of ~$\nabla $ and the orientation of~$\XL$, so is a
topological invariant~$\theta (\sEm)$ of the bundle $\sEm\to\XL$.
  \end{lemma}

  \begin{proof}
 In both the parity-reversing and time-reversing cases the action of the
generator of~$\zt$ reverses the orientation of~$\XL$.  In the
parity-reversing case the lifted action preserves~ $\nabla$, and since the
Chern-Simons invariant changes sign when the orientation is reversed we have
$\theta (\nabla)=-\theta (\nabla)$.  The same argument works in the
time-reversal case, but as the lifted action on $\sEm\to\XL$ is antilinear,
we must also observe that the Chern-Simons invariant of the complex conjugate
bundle equals that of the original bundle.\footnote{This follows since
$c_2(\overline{\sEm} )=c_2(\sEm)$ and this Chern-Simons invariant is a
secondary invariant of the second Chern class~$c_2$.}  The second statement
is a consequence of Lemma~\ref{thm:157}: the continuous discrete
function~$\theta $ on the connected space of invariant connections is
constant.
  \end{proof}

It now follows from Theorem~\ref{thm:133} that $\theta $~is an invariant of
the topological phase of the system: 
  \begin{equation}\label{eq:158}
     \theta \:\TPI\Gpt\longrightarrow \{0,\pi \}. 
  \end{equation}

  \begin{proposition}[]\label{thm:147}
 $\theta $~factors through a homomorphism~$\bar\theta $ of abelian groups:
  \begin{equation}\label{eq:159}
     \xymatrix{ 
      \TPI\Gpt\ar[rr]^\theta\ar[dr] &&\{0,\pi \} \\ 
      & \GSI\Gpt\ar@{-->}[ur]_{\bar\theta }}
  \end{equation} 
  \end{proposition}

  \begin{proof}
 This follows directly from Definition~\ref{thm:135} and the universal
property~\eqref{eq:160} of the group completion, once we prove that $\theta
$~is a homomorphism.  But this is immediate: if $E=E_1\oplus E_2$ is a direct
sum then we use a direct sum~$\nabla _1\oplus \nabla _2$ of invariant
covariant derivatives to compute the Chern-Simons invariant, which satisfies
$\theta (\nabla _1\oplus \nabla _2)=\theta (\nabla _1)+\theta (\nabla _2)$.  
  \end{proof}

In case the system admits a parity-reversing symmetry, so the symmetry group
is~\eqref{eq:168}, there is an easy formula for~$\theta (\sEm)$.  As observed
earlier, in this case $\sEm$~has a quaternionic structure.  Let $P\to\XL$ be
the associated principal $\Sp_N$-bundle, where the complex rank
of~$\sEm\to\XL$ is~$2N$, and endow it with a connection invariant under
time-reversal (Lemma~\ref{thm:157}).  Now any $\Sp_N$-bundle over a 3-manifold
is trivializable, so let $s\:\XL\to P$ be a section.  Define $f\:\XL\to \Sp_N$
by $\hs^*s=s\cdot f$, where $\hs\:P\to P$ is the lift of~$\sigma $ defined by
the action of time-reversal on~$\sEm$.   

  \begin{proposition}[]\label{thm:185}
 $\theta (\sEm)= (\deg f)\pi \pmod{2\pi }$. 
  \end{proposition}

\noindent
 Here $\deg f$ is the induced map $f_*\:H_3(\XL)\to H_3(\Sp_N)$; observe
that $H_3(\XL)\cong H_3(\Sp_N)\cong \ZZ$. 

  \begin{proof}
 Fix an orientation of~$\XL$ and a generator of~$H^3(\Sp_N;\ZZ)$.  Let
$\alpha \in \Omega ^3(P)$ be the Chern-Simons form of the connection on~$P$.
We apply~\cite[(1.28)]{F2} to the map~$\hs$ to conclude that 
  \begin{equation}\label{eq:209}
     \sigma ^*s^*\alpha = s^*\alpha +\omega +d\eta , 
  \end{equation}
where $\omega \in \Omega ^3(\XL)$  integrates to the degree of~$f$ and $\eta
\in \Omega ^2(\XL)$.  Integrate~\eqref{eq:209} over~$\XL$ and use Stokes'
theorem and the fact that $\sigma $~is orientation-reversing.  Since $\theta
(\sEm)=2\pi \int_{\XL}s^*\alpha $ is the Chern-Simons invariant, the proposition
follows.  
  \end{proof}

The group~$\GSI\Gpt$ is expressed as a $K$-theory group in
\eqref{eq:203}--\eqref{eq:205}, and the appropriate $K$-theory group is
computed in Theorem~\ref{thm:180}.  Consider the time-reversal
case~\eqref{eq:156} for which the $K$-theory group is computed
in~\eqref{eq:203} and Theorem~\ref{thm:180}(ii) to be 
  \begin{equation}\label{eq:207}
     \GSI\Gpt\cong \ZZ\times (\zt)^{\times 4}. 
  \end{equation}

  \begin{proposition}[]\label{thm:181}
 The invariant~$\bar\theta $ vanishes on the three ``weak'' $\zt$'s in the
image of the projections~\eqref{eq:206} and is the identity map on the
strong~$\zt$ which is the image of the collapse map~\eqref{eq:201}. 
  \end{proposition}

  \begin{proof}
 If the bundle $\sEm\to \cir\times \cir\times \cir$ is pulled back from a
bundle over~$\cir\times \cir$, then we can choose a connection which is also
pulled back.  It follows easily that the Chern-Simons invariant vanishes
(essentially because $\cir\times \cir$ is 2-dimensional and it is computed as
the integral of a 3-form).

Let $S^3=\RR^3\cup\{\infty \}$ denote the $3$-dimensional sphere with the
involution $x\in \RR$ maps to~$-x$.  Fix an equivariant map $q \:\XL\cong
\cir\times \cir\times \cir\to S^3$ as in~\eqref{eq:201}.  Let $U\subset \HH$
denote the unit quaternions; $U$~is diffeomorphic to a 3-sphere.  There
exists a continuous function $f\:S^3\to U$ such that
  \begin{equation}\label{eq:208}
     f(\infty )=1,\qquad f(0)=-1,\qquad
     f(-x)=\overline{f(x)}\textnormal{\quad for
     all $x\in \RR^3\cup\{\infty\}$}, 
  \end{equation}
where the bar denotes conjugation in the quaternions: namely, $f$~is the
identity map $S^3\to S^3$, after identifying~$U\approx S^3$ and lining up the
involutions.  Now let $E\to S^3$ be the trivial bundle with fiber~$\HH$, and
let $\HH$~act by right quaternion multiplication on each fiber.  This is the
quaternionic structure on~$E$.  Define~$\epsilon $ as the trivial lift of the
involution~$x\mapsto -x $ composed with left multiplication by~$f$.  Right
multiplication by~$j\in \HH$ and~$\epsilon $ generate the action of
$G''=\zt\times \zt$ in~\eqref{eq:168}.  It follows that the pullback $q
^*E\to\XL$ has both parity-reversing and time-reversing symmetry.  Now since
the map~$q$ is a diffeomorphism off of a set of measure zero, we can compute
the Chern-Simons invariant on~$S^3$.  Apply\footnote{Although
Proposition~\ref{thm:185} is stated and proved for the particular
3-manifold~$\XL$, the argument applies to any orientable 3-manifold with
involution.} Proposition~\ref{thm:185}.  Since $E$~is trivial, the associated
principal $\Sp_1$-bundle has a canonical trivializing section~$s$, and by
definition $\hs ^*s=s\cdot f$ for the function~$f$ in~\eqref{eq:208}.  Since
$f$~has degree~1, it follows that $\theta (q^*E)$~is nonzero.
  \end{proof}

  \subsection*{The Kane-Mele invariant}

This invariant was introduced in~\cite{KM2} and further studied in many
papers, including~\cite{FKM}.  There are many equivalent definitions.  
 
We begin with some preliminaries on quaternionic vector spaces.  Let $W$~be a
complex vector space and $J\:W\to W$ a complex antilinear operator with
$J^2=\pm\id_W$.  Then $J$~is either a \emph{real} structure~($+$) or a
\emph{quaternionic} structure~($-$).  If $(W_i,J_i)$, $i=1,\dots ,N$, is a
finite set of real or quaternionic vector spaces, then the tensor product
$W_1\otimes \cdots\otimes W_N$ carries an antilinear operator $J_1\otimes
\cdots\otimes J_N$ which is a real or quaternionic structure according to the
parity of the number of quaternionic structures in~$\{J_i\}$.  This applies
to symmetric and antisymmetric tensor products as well.  All tensor products
are taken over~$\CC$.  If $(W,J)$~is quaternionic and finite dimensional,
then $\dim_{\CC}W=2\ell $ is even, and the \emph{determinant line} $\Det
W:={\textstyle\bigwedge} ^{2\ell }W$ inherits a \emph{real} structure $\det
J:={\textstyle\bigwedge} ^{2\ell }J$.  Thus $\Det W$~is a complex
line\footnote{A \emph{line} is a 1-dimensional vector space.} with a
distinguished real line~$\Det_J W$ contained as a real subspace.  We claim
that the real line~$\Det_JW$ has a canonical \emph{orientation}---a choice of
component of $\Det_JW\setminus \{0\}$.  For if $e_1,\dots ,e_\ell $ is a
quaternionic basis of~$(W,J)$, then
  \begin{equation}\label{eq:169}
     (e_1\wedge Je_1)\wedge (e_2\wedge Je_2)\wedge \cdots\wedge (e_\ell \wedge
     Je_\ell) \in \Det_JW \subset \Det W 
  \end{equation}
is nonzero, so orients~$\Det_JW$.  The collection of quaternionic bases is
connected, and \eqref{eq:169}~varies continuously, so the orientation is
well-defined independent of the basis.  Alternatively, if we endow~$W$ with a
hermitian metric~$\langle -,- \rangle$ with respect to which $J$~is
antiunitary, then $J$~determines a 2-form~$\omega _J\in {\textstyle\bigwedge}
^2W^*$ defined by
  \begin{equation}\label{eq:170}
     \omega _J(\xi _1,\xi _2)=\langle J\xi _1,\xi _2 \rangle,\qquad \xi
     _1,\xi _2\in W. 
  \end{equation}
The \emph{pfaffian} 
  \begin{equation}\label{eq:171}
     \pfaff\omega _J=\frac{\omega _J^\ell }{\ell !}=\frac{\omega _J\wedge
     \cdots\wedge \omega _J}{\ell !} 
  \end{equation}
is a nonzero element of~$\Det_JW^*$.  Since the space of hermitian metrics is
connected, the induced orientation is independent of~$\langle -,-  \rangle$. 
 
Assume given the band structure at the beginning of~\S\ref{sec:6} for the
time-reversal case~\eqref{eq:156}.  The complex finite rank vector bundle
$\sEm\to\XL$ has an antilinear lift~$\ts$ of the involution $\sigma
\:\XL\to\XL$ which complex conjugates characters.  The involution~$\sigma $
has $2^d$~fixed points.  Let $F\subset \XL$ denote the fixed point set.  At
each fixed point~$\lambda \in F$ the lift~$\ts$ of~$\sigma $ to the
fiber~$\sEm_\lambda $ is a quaternionic structure.  Let $\Det\sEm\to\XL$ be
the (complex) determinant line bundle, obtained by replacing each fiber
of~$\sEm$ by its determinant line.  Then $\ts$~induces a lift~$\det\ts$
of~$\sigma $ to~$\Det\sEm$, and $(\det\ts)^2=\id_{\Det \sEm}$.  At a fixed
point~$\lambda \in F$ the antilinear map~$\det\ts_\lambda $ is the real
structure discussed in the previous paragraph---its fixed points comprise the
real line $\Det_{\ts}\sEm_\lambda \subset \Det\sEm_\lambda $.

  \begin{lemma}[]\label{thm:159}
 Suppose $L\to\XL$ is a complex line bundle with an antilinear lift~$\alpha $
of~$\sigma $ whose square is~$\id_L$.  Then $L\to\XL$ admits a nowhere
vanishing $\alpha $-invariant section. 
  \end{lemma}

  \begin{proof}
 We first prove that there is a nowhere vanishing section, equivalently that
$L\to\XL$ is non-equivariantly trivializable, equivalently that $c_1(L)\in
H^2(\XL;\ZZ)$ vanishes.  Now any class in $H^2(\XL;\ZZ)$ is determined by its
value on the fundamental classes~$[T]$ of 2-dimensional subtori $T\subset
\XL$.  Since $L\cong \sigma ^*\overline{L}$, we have
  \begin{equation}\label{eq:172}
     \langle c_1(L),[T] \rangle = \langle c_1(\sigma ^*\overline{L}),[T]
     \rangle = \langle c_1(\overline{L}),\sigma _*[T] \rangle = -\langle
     c_1(L),[T] \rangle , 
  \end{equation}
and so the pairing $\langle c_1(L),[T] \rangle\in \ZZ$ vanishes.
 
Let $s'\:\XL\to L$ be a nowhere vanishing section of $L\to\XL$.  Define
$f\:\XL\to\CC^\times $ by $\alpha ^*s'=fs'$, where $\CC^\times \subset \CC$
is the multiplicative group of nonzero complex numbers; then $\sigma ^*\bar
f\cdot f=1$.  It follows that the winding number of~$f$ around any
1-dimensional subtorus~$S\subset \XL$ vanishes, whence $f$~has a square
root---a function $g\:\XL\to\CC^\times $ such that~$g^2=f$.  Then $(\sigma
^*\bar g\cdot g)^2=1$, and since at a fixed point~$\lambda \in F$ we have
$|f(\lambda )|=1$, it follows that $|g(\lambda )|=1$ and so $\sigma ^*\bar
g\cdot g=1$.  The nowhere vanishing section~$s=gs'$ is invariant: $\alpha
^*s=s$.  
  \end{proof}

Apply Lemma~\ref{thm:159} to $\Det\sEm\to\XL$ to conclude that there exists a
nowhere vanishing $\det\ts$-invariant section~$s$.  At a fixed point~$\lambda
\in F$ we have $s(\lambda )\in \Det_{\ts}\sEm_\lambda $, where recall
$\Det_{\ts}\sEm_\lambda \subset \Det\sEm_\lambda $ is a distinguished
oriented real line.  Define $\delta _\lambda (s)=\pm1$ according
as~$s(\lambda )$ is in~($+$) or not in~($-$) the half-line
of~$\Det_{\ts}\sEm_\lambda $ defined by the orientation.

  \begin{lemma}[]\label{thm:160}
 $\prod\limits_{\lambda \in F}\delta _\lambda (s)$ is independent of~$s$ if
$d=\dim \XL$ satisfies~$d\ge2$.
  \end{lemma}

  \begin{proof}
 Any other nowhere vanishing $\det\ts$-invariant section has the form~$hs$
for $h\:\XL\to\CC^\times $ with $\sigma ^*\bar h=h$.  Write $\XL=S'\times
S''$ for $S',S''\subset \XL$ subtori of dimensions~$1,d-1$, respectively.
Then $F=F'\times F''$ is the Cartesian product of the fixed point sets
of~$\sigma $ restricted to~$S',S''$.  Let $m\in \ZZ$~be the winding number
of~$h\res{S'}$.  Now at each fixed point $h$~is a nonzero real number, so
$\sign h=\pm1$ is well-defined.  If $F'=\{1,\lambda '\}$ and $\lambda ''\in
F''$ it follows that $\sign h(1,\lambda '')\cdot \sign h(\lambda ',\lambda
'')=(-1)^m$.  Since~$d\ge2$---equivalently, $F''\not= \emptyset $---we
conclude $ \prod\limits_{\lambda \in F}\sign h(\lambda )=1$, as desired. 
  \end{proof}

  \begin{definition}[]\label{thm:161}
 The \emph{Kane-Mele invariant} of $\sEm\to\XL$ is 
  \begin{equation}\label{eq:173}
     \kappa (\sEm) = \prod\limits_{\lambda \in F}\delta _\lambda (s)\;\in
     \pmo, 
  \end{equation}
where $s$~is any nowhere vanishing $\det\ts$-invariant section of
$\Det\sEm\to\XL$.  
  \end{definition}

\noindent
 Analogous to Proposition~\ref{thm:147}, and with almost the same proof, we
have the following. 

  \begin{proposition}[]\label{thm:162}
 $\kappa $~factors through a homomorphism $\bar\kappa $ of abelian groups; 
  \begin{equation}\label{eq:174}
     \xymatrix{ \TPI\Gpt\ar[rr]^\kappa\ar[dr] &&\pmo \\ &
      \GSI\Gpt\ar@{-->}[ur]_{\bar\kappa }} 
  \end{equation}
  \end{proposition}

  \begin{proof}
 If $E=E_1\oplus E_2$ is a direct sum, then $\Det E\cong \Det E_1\otimes \Det
E_2$ and we use the tensor product of nowhere vanishing invariant sections to
verify that $\kappa (E)=\kappa (E_1)\kappa (E_2)$ is a homomorphism. 
  \end{proof}

  \begin{remark}[]\label{thm:163}
 Recall from~\eqref{eq:203} that $\GS\Gpt\cong KR^{-4}(\XL)$.  We observe
that there is a Kane-Mele invariant defined on each subtorus~$T\subset \XL$
of dimension at least two.
  \end{remark}

In case the system admits a parity-reversing symmetry, so the symmetry group
is~\eqref{eq:168}, then there is a well-known easy formula for~$\kappa
(\sEm)$.  The parity-reversing symmetry also acts on~$\XL$ by the
involution~$\sigma $, and its lift~$\epsilon $ to~$\sEm\to\XL$ is complex
linear, squares to~$\id_{\sEm}$, and commutes with~$\ts$.  The composition
$j=\epsilon \circ \ts$ then covers~$\id_{\XL}$, is complex antilinear,
and~$j^2=-\id_{\sEm}$.  Thus $j$~is a \emph{global} quaternionic structure on
$\sEm\to\XL$.  Now at a fixed point~$\lambda \in F$ there are two
quaternionic structures---$\ts_\lambda $ and~$j_\lambda $---and $\ts_\lambda
=\epsilon _\lambda \circ j_\lambda $, where $\epsilon _\lambda \in
\End\sEm_\lambda $ squares to~$\id_{\sEm_\lambda }$ and commutes
with~$j_\lambda $.  Therefore, there is a quaternionic basis~$e_1,\dots
,e_\ell $ of~$(\sEm_\lambda ,j\mstrut _\lambda )$ such that $\epsilon _\lambda
(e_i)=\pm e_i$ for each~$i$.  Let $\Delta _\lambda =\pm1$ be the product over
the basis of the signs. 

  \begin{proposition}[]\label{thm:164}
 $\kappa (\sEm)=\prod\limits_{\lambda \in F}\Delta _\lambda $. 
  \end{proposition}

  \begin{proof}
 Fix a hermitian metric on $\sEm\to\XL$ which is invariant under~$j$,
$\epsilon $, and hence also~$\ts$.  Then as in~\eqref{eq:171} the pfaffian
$s=\pfaff \omega _j$ is a nowhere vanishing $\det j$-invariant section of
$\sEm\to\XL$.  To compute~$\kappa $ we compare the canonical orientations
of~$\Det_j\sEm=\Det_{\ts}\sEm$ at each~$\lambda \in F$.  That the ratio
is~$\Delta _\lambda $ follows from the formula 
  \begin{equation}\label{eq:175}
     (e_1\wedge \ts_\lambda e_1)\wedge \cdots\wedge (e_\ell \wedge
     \ts_\lambda e_\ell ) =\Delta _\lambda \,(e_1\wedge je_1)\wedge
     \cdots\wedge (e_\ell \wedge je_\ell ), 
  \end{equation}
where as above $e_1,\dots ,e_\ell $ is a quaternionic basis of~$(\sEm_\lambda
,j\mstrut _\lambda )$.   
  \end{proof}

  \begin{proposition}[]\label{thm:182}
 The invariant~$\bar\kappa $ in Proposition~\ref{thm:162} vanishes on the
three ``weak'' $\zt$'s in the image of the projections~\eqref{eq:206} and is
the identity map on the strong~$\zt$ which is the image of the collapse
map~\eqref{eq:201}. 
  \end{proposition}

  \begin{proof}
 If the bundle $\sEm\to \cir\times \cir\times \cir$ is pulled back from a
bundle over~$\cir\times \cir$, then we can choose the section~$s$ to also be
a pullback.  Since pairs of fixed points in the 3-torus are mapped to a
single fixed point of the 2-torus, we see that each sign in~\eqref{eq:173}
occurs with even multiplicity, whence $\kappa (\sEm)=+1$. 
 
Now consider the bundle pulled back from the trivial bundle $E\to S^3$ via
$q\:\cir\times \cir\times \cir\to S^3$, as in the paragraph
containing~\eqref{eq:208}.  We describe~$q$ more explicitly.  Choose an
isomorphism $\XL\cong \RR^3/\ZZ^3$ and map the fundamental domain
  \begin{equation}\label{eq:178}
      R=\{x=(x^1,\dots ,x^3)\in \RR^3:-1/2\le x^i\le 1/2\} 
  \end{equation}
to~$\RR^3\cup\{\infty \}$ by $\pi \:x\mapsto x/\epsilon (x)$, where $\epsilon
(x)$~is the Euclidean distance from~$x$ to~$\partial R$.  Then $q$~maps
$2^3-1$~points of~$F\subset \XL\approx \cir\times \cir\times \cir$ to~$\infty
\in S^3$ and the remaining fixed point to~$0\in S^3$.  Since $q^*E\to \XL$
has both time-reversal and parity-reversal symmetry, we can apply
Proposition~\ref{thm:164}.  It is easy to see $\Delta _\lambda =+1$ at
$2^3-1$~points of~$F$ and $\Delta _\lambda =-1$ at the remaining fixed point,
whence $\kappa (q ^*E)=-1$, as desired.
  \end{proof}

Proposition~\ref{thm:181} and Proposition~\ref{thm:182} together with the
computation Theorem~\ref{thm:180}(ii) and the remark around~\eqref{eq:206}
immediately imply the following result of Wang-Qi-Zhang.

  \begin{corollary}[\cite{WQZ}]\label{thm:183}
 The orbital magnetoelectric polarizability and 3-dimensional Kane-Mele
invariant are equal: $\theta =\kappa $.
  \end{corollary}

\appendix

   \section{Group extensions}\label{sec:13}

Group extensions appear throughout the main body of this paper, and for the
reader's convenience we summarize some basic definitions in this appendix.
We do not comment further on topology, but say once and for all that for
topological groups all group homomorphisms are assumed continuous, and for
Lie groups all group homomorphisms are assumed smooth.

  \begin{definition}[]\label{thm:54}
 A \emph{group extension} is a sequence of group homomorphisms 
  \begin{equation}\label{eq:70}
     1\longrightarrow G'\xrightarrow{\;\;i \;\;} G\xrightarrow{\;\;\pi \;\;}
     G''\longrightarrow      1 
  \end{equation}
such that the composition of any two consecutive arrows is the identity and,
more strongly, the kernel of any homomorphism equals the image of the
preceding homomorphism.  We call~ $G'$ the \emph{kernel} and $G''$~the
\emph{quotient}.
  \end{definition}

\noindent
 The equality of kernel and image is called \emph{exactness}, and it has
consequences at the three ``interior nodes'' of~\eqref{eq:70}.  At~ $G'$ it
implies that $i$~is injective (1:1).  We use the inclusion~$i$ to
identify~$G'$ with its image, which is a subgroup of~$G$.  At~$G$ the
exactness implies that $\pi $~ factors through an injective map of the
quotient group~$G/G'$ into $G''$.  At~$G''$ the exactness implies that $\pi
$~ is surjective (onto), hence the quotient map $G/G'\to G''$ an isomorphism.
We use it to identify~$G''$ as this quotient.

As with any algebraic structure, there is a notion of homomorphism. 

  \begin{definition}[]\label{thm:88}
 Let $G',G''$~be groups and $G^{\tau _1},G^{\tau _2}$ group extensions with
kernel~$G'$ and quotient~$G''$.  A \emph{homomorphism of group extensions~
$G^{\tau _1},G^{\tau _2}$} is a group homomorphism $\varphi \:G^{\tau _1}\to
G^{\tau _2}$ which fits into the commutative diagram
  \begin{equation}\label{eq:19}
     \xymatrix@R-15pt{&&G^{\tau _1}\ar[dd]^{\varphi }\ar[dr]\\ 1 \ar[r]& G'
     \ar[ur]\ar[dr] && G'' \ar[r] & 1\\ && G^{\tau _2}\ar[ur]} 
  \end{equation} 
  \end{definition}

\noindent
 As usual, $\varphi $~is an \emph{isomorphism} if there exists a homomorphism
$\psi \:G^{\tau _2}\to G^{\tau _1}$ of group extensions so that the
compositions $\psi \circ \varphi $~and $\varphi \circ \psi $~are identity
maps, which is simply equivalent to the condition that $\varphi $~be an
isomorphism of groups.. 

There is a notion of a trivialization of a group extension.

  \begin{definition}[]\label{thm:86}
 A \emph{splitting} of the group extension~\eqref{eq:70} is a homomorphism
$j\:G''\to G$ such that $\pi \circ j=\id_{G''}$.    
  \end{definition}

\noindent
 Not every group extension is split: for example, the cyclic group of order~4
is a nonsplit extension of the cyclic group of order~2 by the cyclic group of
order~2:  
  \begin{equation}\label{eq:91}
     1\longrightarrow \zmod2\longrightarrow \zmod4\longrightarrow
     \zt\longrightarrow 1 
  \end{equation}
For example, the group extension~\eqref{eq:12} is split whereas
\eqref{eq:13}~is not split.  A splitting $j\:G''\to G$ of~\eqref{eq:70}
determines a homomorphism
  \begin{equation}\label{eq:92}
     \alpha \:G''\longrightarrow \Aut(G'), 
  \end{equation}
which we also call an action of~$G''$ on~$G'$.  Namely, the action of $g''\in
G''$ is
  \begin{equation}\label{eq:93}
     \alpha (g'')(g')=gg'g\inv , \quad g'\in G',
  \end{equation}
where $g=j(g'')\in G$.  (We remark that if $G'$~is \emph{abelian}, then
~\eqref{eq:92}~is defined without a splitting: let $g$~be any element of~$G$
such that $\pi (g)=g''$.)  There is a converse construction. 

  \begin{definition}[]\label{thm:87}
 Let $G',G''$ be groups and $\alpha \:G''\to\Aut(G')$ a homomorphism.  The
\emph{semidirect product} $G=G''\ltimes _\alpha G'$ is defined to be the
Cartesian product set ~$G''\times G'$ with multiplication defined so that
$G''\times 1$ and $1\times G'$ are subgroups, and
  \begin{equation}\label{eq:22}
     g''\cdot g' = \alpha (g'')(g')\cdot g'',\qquad g''\in G'',\quad g'\in
     G'. 
  \end{equation} 
  \end{definition}

\noindent
 A semidirect product sits in a split extension 
  \begin{equation}\label{eq:94}
     \xymatrix{1\ar[r] &G'\ar[r]^{i}&G\ar@<-.5ex>[r]_\pi
     &G''\ar@<-.5ex>[l]_j\ar[r]&1}  
  \end{equation}
and conversely every split extension is isomorphic to a semidirect product.
We say that a group extension is \emph{trivializable} if it is isomorphic to
a split extension; a choice of splitting is a trivialization.
 
Group extensions can be pulled back and, in very special situations, pushed
out.

  \begin{definition}[]\label{thm:89}
 Suppose \eqref{eq:70}~is a group extension and $\rho'' \:\tilG''\to G''$ a
group homomorphism.  Then there is a \emph{pullback} group extension with
kernel~$G'$ and quotient~$\tilG''$ which fits into the commutative diagram
  \begin{equation}\label{eq:95}
     \xymatrix{1 \ar[r]& G' \ar[r] \ar@{=}[d] & \tilG \ar[d]^\rho
     \ar[r]^{\tilde\pi } & \tilG'' 
     \ar[r]\ar[d]^{\rho''} & 1\\ 1 \ar[r] &G' \ar[r] &G
     \ar[r]^\pi  &G''\ar[r]&1}  
  \end{equation}
It is defined by setting
  \begin{equation}\label{eq:96}
     \tilG=\bigl\{(g,\tg'')\in G\times \tilG'' : \pi (g) = \rho'' (\tg'')
     \bigr\}.  
  \end{equation} 
  \end{definition}

\noindent
 Note that $\tilG$~is a subgroup of the direct product group~$G\times
\tilG''$.  The homomorphisms ~$\rho $ and $\tilde\pi $~are defined by
restricting the projections onto the factors of~$G\times \tilG''$,
respectively.

  \begin{definition}[]\label{thm:106}
 Suppose \eqref{eq:70}~is a group extension and $\rho '\:G'\to\tilG'$ a group
homomorphism with $\tilG'$ abelian.  Assume further that the action of~$G''$
on~$G'$ fixes~$\rho '$.  Then there is an \emph{associated group extension}
with kernel~$\tilG'$ and quotient~$G''$ which fits into the commutative
diagram
  \begin{equation}\label{eq:112}
     \xymatrix{1 \ar[r]& G' \ar[r]^{i} \ar[d]^{\rho'} & G \ar[d]^\rho
     \ar[r] & G'' \ar[r]\ar@{=}[d] & 1\\ 1 \ar[r] &\tilG'
     \ar[r]^{\tilde{i}} &\tilG \ar[r] &G''\ar[r]&1} 
  \end{equation}
It is defined by setting 
  \begin{equation}\label{eq:113}
     \tilG = G\times \tilG'\bigm/\sim,\qquad (g,\tg')\sim\bigl(i(g')g,\rho
     '(g')\tg'\bigr) 
  \end{equation}
for all $g\in G$, $\tg'\in \tilG'$, $g'\in G'$. 
  \end{definition}

\noindent 
 If we view the top row of \eqref{eq:112} as a principal $G'$-bundle, then
this is the associated bundle construction; the hypotheses allow us to
descend the product group structure on~$G\times \tilG'$ to~$\tilG$.

   \section{Clifford algebras and the 10-fold way}\label{sec:12}

The classic reference is~\cite{ABS}.  We begin with a summary of some basics.
Recall that for each integer~$n\in \ZZ$ the standard Clifford algebra~$\Cl n$
is the real associative algebra with unit generated by~$|n|$
elements~$e_1,e_2,\dots ,e_{|n|}$ with $e_ie_j+e_je_i=0$ if~$i\not= j$, and
$e_i^2=\sign(n)$ for all~$i$.  The Clifford algebra~$\Cl0$ is the
field~$\RR$.  The Clifford algebra is naturally $\zt$-graded: the
generators~$e_i$ are odd.  The corresponding \emph{complex} superalgebras are
denoted~$\Clc n$.  For $n$~even the complex Clifford algebra~$\Clc n$ is
isomorphic to the superalgebra of endomorphisms of a $\zt$-graded complex
vector space~$S=S^0\oplus S^1$ with $\dim S^0=\dim S^1$.  As usual, after
choosing a basis we write elements of~$\End S$ as square matrices, which
relative to the decomposition $S=S^0\oplus S^1$ have a $2\times 2$~block
form.  The even elements of~$\End S$ are block diagonal matrices; the odd
elements are block off-diagonal matrices.  Over the reals the Clifford
algebras~$\Cl n$ are matrix superalgebras if $n$~is divisible by~8.  It
happens that as \emph{ungraded} algebras some other Clifford algebras are
ungraded matrix algebras.  For example, as an ungraded algebra $\Cl 2$~is the
algebra of $2\times 2$~real matrices.  This algebra is generated by
  \begin{equation}\label{eq:71}
     e_1 = \begin{pmatrix} 1&0\\0&-1 \end{pmatrix},\qquad 
     e_2 = \begin{pmatrix} 0&1\\1&0
     \end{pmatrix} 
  \end{equation}
which satisfy the Clifford relations.  However, the matrix which represents
$e_1$~is diagonal hence even, not odd.  Thus the isomorphism~\eqref{eq:71}
between~$\Cl2$ and $2\times 2$~matrices does not preserve the $\zt$-grading.
On the other hand, the mixed signature Clifford algebra with two
anticommuting generators~$e_+,e_-$ whose squares are $e_+^2=+1$, $e_-^2=-1$
is isomorphic to the superalgebra of real $2\times 2$ matrices and the
isomorphism preserves the grading:
  \begin{equation}\label{eq:72}
     e_+ = \begin{pmatrix} 0&1\\1&0 \end{pmatrix},\qquad e_- = \begin{pmatrix}
     0&-1\\1&0 \end{pmatrix} 
  \end{equation}
If $S=S'\amalg S''$ is a finite set of generators, written as a disjoint
union, then the Clifford algebra with generating set~$S$ is a graded tensor
product $\Cliff(S)\cong \Cliff(S')\otimes \Cliff(S'')$. 
 
Our interest is in representations of Clifford algebras, or better
\emph{Clifford modules}.  A left Clifford module is a $\zt$-graded vector
space~$V$ with a Clifford algebra acting consistently with the grading.
Clifford modules for a fixed Clifford algebra~$R$ are the objects of a
category~$\Mod R$, and superalgebras~$R_0$ and~$R_1$ are \emph{Morita
equivalent} if $\Mod{R_0}$ and~$\Mod{R_1}$ are equivalent categories.  This
holds if~$R_1\otimes \End(S_1)\cong R_0\otimes \End(S_0)$ for some super
vector spaces~$S_0,S_1$.  Here, as always, we use the graded tensor product
defined with the Koszul sign rule (Remark~\ref{thm:67}).  The equivalence
maps an $R_0$-module $V_0$ to the $R_1$-module $V_0\otimes S$.  In the other
direction an $R_1$-module~$V_1$ maps to the $R_0$-module
$\Hom_{\End(S)}(S,V_1)$.  The Morita class of a complex Clifford
algebra~$\Clc n$ is determined by $n\pmod 2$; the Morita class of a real
Clifford algebra~$\Cl n$ by $n\pmod 8$.  Finally, there is an equivalence
between quaternionic $\zt$-graded vector spaces and $\Cl{-4}$-modules; in
other words, $\Cl {-4}$~is Morita equivalent to the (even) quaternion
algebra~$\HH$.  Let~$W$ be a quaternionic super vector space, i.e., a real
$\zt$-graded vector space with an even action of the
quaternions~$\qi,\qj,\qk$ and a grading operator~$\sigma $ (with $\sigma
^2=1$).  The associated $\Cl{-4}$-module is the real $\zt$-graded vector
space~$V=W\oplus W$ with grading $\left(\begin{smallmatrix} \sigma
&0\\0&-\sigma \end{smallmatrix}\right)$ and Clifford action
  \begin{equation}\label{eq:73}
      e_1  = \begin{pmatrix} 0 & -1 \\ 1 & 0 \\ \end{pmatrix},
     \qquad e_2  = \begin{pmatrix} 0 & \qi \\ \qi & 0 \\ \end{pmatrix},
      \qquad e_3  
     = \begin{pmatrix} 0 & \qj \\ \qj & 0 \\ \end{pmatrix}, \qquad e_4  =
     \begin{pmatrix} 0 & \qk \\ \qk & 0 \\ \end{pmatrix} 
  \end{equation}
Conversely, given a $\Cl{-4}$-module~$V$, the $(-1)$-eigenspace of $\gamma
=e_1e_2e_3e_4$ is a quaternionic vector space with $e_1e_2, e_1e_3, e_1e_4$
acting as unit quaternions~$i,j,k$.

We now consider the universal group~$\sC=\pmo\times \pmo$ which tracks
time-reversal and Hamiltonian-reversal.  It is equipped with homomorphisms
$c,\phi_{\sC} \:\sC\to\pmo$, where $c$~is projection onto the second factor,
and $\phi _{\sC}$~is multiplication $\pmo\times \pmo\to\pmo$.  As
in~\S\ref{sec:4} we define a \emph{\CTt} as a pair~$(A,\tA)$ consisting of a
subgroup~$A$ of~$\sC$ and a $\phi $-twisted extension $1\to \TT\to\tA\to A\to
1$, where $\phi $~is the restriction of~$\fC$ to~$A$.  There are
10~isomorphism classes of \CTt s (Proposition~\ref{thm:24}).  These may be
labeled by the special lifts~$T,C,S$ of the non-identity elements
$\bT=(-1,1)$, $\bC=(1,-1)$, $\bS=(-1,-1)$ of~$\sC$, as in Lemma~\ref{thm:42}.
Recall from Definition~\ref{thm:25}(ii) that a $(\phi ,\tau ,c)$-twisted
representation of~$A$ is a $\zt$-graded complex vector space~$W$ with an
action of~$\tA$ for which $\phi $~tracks linearity/antilinearity and
$c$~tracks evenness/oddness.

  \begin{proposition}[]\label{thm:57}
 Each \CTt~$(A,\tA)$ is associated to a real or complex Clifford
algebra~$R$, as indicated in the following table, such that the category of
$(\phi ,\tau ,c)$-twisted representations of~$A$ is equivalent to the
category of real or complex $R$-modules.
  \bigskip
  \begin{center}
  \renewcommand{\arraystretch}{1.5}
  \begin{tabular}{|c||c|c|c|c|c|c|c|c|c|c|}
\hline 
$A$&1&\textnormal{diag}&$\pmo\times \oo$&$\sC$&$\oo\times \pmo$&$\sC$&$\pmo\times \oo$&$\sC$&$\oo\times \pmo$&$\sC$ \\
\hline
$T^2$&&&$+1$&$+1$&&$-1$&$-1$&$-1$&&$+1$ \\
\hline
$C^2$&&&&$-1$&$-1$&$-1$&&$+1$&$+1$&$+1$\\
 \hline&&&&&&&&&&\\[-1.6em] \hline
$R$&$\Clc0$&$\Clc{-1}$&$\Cl0$&$\Cl{-1}$&$\Cl{-2}$&$\Cl{-3}$&$\Cl{-4}$&$\Cl{-5}$&$\Cl{-6}$&$\Cl{-7}$\\  
 \hline 
  \end{tabular}
  \renewcommand{\arraystretch}{1}
  \end{center}
\medskip
  \end{proposition}

\noindent 
 The Clifford algebra is only determined up to Morita equivalence, which we
use to identify $\Cl{-5}$, $\Cl{-6}$, $\Cl{-7}$ with $\Cl3$, $\Cl2$, $\Cl1$,
respectively.  A closely related discussion and table appear
in~\cite[\S5]{FKi1} and in~\cite{AK}.

  \begin{proof}
 We first indicate the equivalence which maps a $(\phi,\tau,c)$-twisted
representation~$W$ to a Clifford module.  Recall that the vector space
underlying a $\ptc$-twisted representation of~$A$ is \emph{complex}.  Also,
the representation of~$\tA$ is determined by the action of the special
lifts~$T,C,S$, not all of which may be present and some of which may be
redundant.  The argument proceeds on a case-by-case basis.
 
For $A=1$ the representation is a $\zt$-graded complex vector space, which is
a $\Clc0$-module.  For $A=\textnormal{diag}$ there is in addition an odd
operator~$S$ with $S^2=1$, and $iS$~generates the complex Clifford
algebra~$\Clc{-1}$.  (One could equally well use~$S$ to generate the Morita
equivalent~$\Clc{+1}$.)
 
If $\bC\in A$ and~$\bT\notin A$, then $C$~and~$iC$ are odd and generate
a Clifford algebra on the real super vector space~$W_{\RR}$ underlying~$W$.
If $C^2=(iC)^2=-1$ this is~$\Cl{-2}$, whence the fifth column of the table.
If $C^2=(iC)^2=+1$ the Clifford algebra is~$\Cl{+2}$, ergo the penultimate
column.

If $\bT\in A$, then the lift~$T\in \tA$ acts as an even antilinear
transformation with square~$\pm1$ on the complex super vector space~$W$.  If
~$T^2=+1$ it is a real structure on~$W$, and if~$T^2=-1$ it is a quaternionic
structure.  This, together with the equivalence above between quaternionic
super vector spaces and $\Cl{-4}$-modules, gives the columns
with~$R=\Cl0,\Cl{-4}$.  In the real case $T^2=+1$, if in addition~$\bC\in A$,
then $C$~is odd and commutes with~$T$, so if $C^2=-1$ it generates a real
Clifford algebra~$\Cl{-1}$, and this gives the equivalence in the fourth
column of the table.  If $C^2=+1$, then $C$~generates the Clifford
algebra~$\Cl{+1}$, and this gives the last column.  Finally, if $T^2=-1$ and
$\bC\in A$, then $C,iC,iTC$ are odd endomorphisms of~$W_{\RR}$ and generate a
$\Cl{\pm3}$-action, according to~$C^2=\pm1$.  This completes the table.

The inverse equivalence takes a Clifford module~ $V$ to a
$(\phi,\tau,c)$-twisted representation~$W$ of the corresponding CT~group~$A$
as follows.  By Morita equivalence we may assume that $V$~ is a Clifford
module for $\Cl n^{\CC}$ with $n=0,1$ or for $\Cl n$ with $\vert n \vert \leq
4$.  In the complex cases we take~$W=V$; for $n=1$ the action of~ $e_1$
generates the representation of $A={\rm diag}$.  If $V$ is a real $\Cl
n$-module with $n = 0, \pm 1$, then let $W=V\otimes \CC$ be the
complexification of~$V$ and $T= \id\mstrut _V\otimes\; \kappa $, where
$\kappa $~is complex conjugation; for $n=\pm 1$ set $C=e_1 \otimes \kappa
$. For $V$ a real $\Cl{\pm 2}$-module, set $W=V$ with complex structure~ $e_1
e_2$ and~ $C=e_1$.  For $V$ a real $\Cl{\pm 3}$-module, set $W=V$ with
complex structure~ $\mp e_1 e_2$, but now let~$C=e_1$ and~$T=-e_2 e_3$.
Finally, motivated by~\eqref{eq:73}, if $V$ is a real $\Cl {-4}$-module, let
$W$~be the $(-1)$-eigenspace of $e_1 e_2 e_3 e_4 $ with complex structure
~$e_1 e_2$.  Then $T=e_1 e_3$~ is even, preserves~$W$, and squares to $-1$.
  \end{proof}

   \section{Dyson's 10-fold way}\label{sec:9}

Dyson~\cite{Dy} studied \emph{irreducible} corepresentations
(Definition~\ref{thm:152} below) of a group~$G$ equipped with a homomorphism
$\phi \:G\to\pmo$ and classified them into 10~types.  We re-derive his
classification by studying commutants, as did Dyson, with the difference that
we take into account a $\zt$-grading and use the Koszul sign rule.  The
argument then reduces to a $\zt$-graded version of Schur's lemma.  We make
contact with Wall's graded Brauer groups~\cite{Wa}, though he too ignored the
Koszul sign rule in his definition of the center.  We follow instead the
treatment by Deligne~\cite[\S3]{De}.  The arguments apply to vector spaces
over very general fields, but the reader should keep in mind the real
numbers, which is the case of interest in this paper.  All algebras are
unital and associative.
 
Recall that an associative unital algebra~$D $ over a field~$k$ is a
\emph{division algebra} if every nonzero element of~$D $ is invertible.

  \begin{theorem}[Schur]\label{thm:151}
 Let $k$~be a field, $G$~a group, and $\rho \:G\to\Aut_k(V)$ an irreducible
representation of~$G$ on a vector space~$V$ over~$k$.  Then the commutant
  \begin{equation}\label{eq:165}
     D =\{a\in \End_k(V): a\rho (g)=\rho (g)a\textnormal{ for all $g\in
     G$}\}
  \end{equation}
is a division algebra.
  \end{theorem}

  \begin{proof}
 Let~$a\in D $ and consider $\ker(a)\subset V$.  Since $a$~commutes with
each~$\rho (g)$ it follows that $\ker(a)$~is $G$-invariant.  But $\rho $~is
irreducible, whence $\ker(a)=V$ or $\ker(a)=0$, i.e., $a=0$ or $a$~is
injective.  Similarly, the image of~$a$ is either~$0$ or~$V$, so $a=0$ or
$a$~is surjective.  Therefore, $a=0$ or $a$~is invertible.
  \end{proof}

According to a classical theorem of Frobenius there are three division
algebras over~ $k=\RR$ up to isomorphism, namely the field of real numbers~$D
=\RR$, the field of complex numbers~$D =\CC$, and the noncommutative algebra
of quaternions~$D =\HH$.  Thus there is a corresponding trichotomy of
irreducible real representations.
 
From now on suppose $1\not= -1$ in the field~$k$, i.e., $k$~does not have
characteristic~2.   

  \begin{definition}[]\label{thm:152}
 Let $G$~be a group and $\phi \:G\to\pmo$ a homomorphism.  Let $V$~be a
vector space over~$k$ and $I\in \End_k(V)$ an automorphism with $I^2\in
k^\times $ a nonzero scalar transformation.
  \begin{enumerate}
 \item Let $\AVI$~be the $\zt$-graded group of automorphisms of~$V$ which
either commute or anticommute with~$I$.

 \item A \emph{corepresentation} of~$G$ on~$V$ is a homomorphism $\rho
\:G\to\Aut_k(V,I)$ of $\zt$-graded groups, which then satisfies
  \begin{equation}\label{eq:166}
     \rho (g)I = \phi (g)I\rho (g)\qquad \textnormal{for all $g\in G$}. 
  \end{equation} 

 \item The corepresentation~$\rho $ is \emph{irreducible} if there is no
proper subspace of~$V$ which is invariant under~$I$ and ~$\rho (g)$ for
all~$g\in G$.
  \end{enumerate}
  \end{definition}

\noindent
 The pair $(G,\phi )$~is called a \emph{$\zt$-graded group}.  The
$\zt$-grading on~$\AVI$ tracks whether an automorphism commutes or
anticommutes with~$I$.  The algebra~$\End_k(V)$ is $\zt$-graded by the
involution $a\mapsto IaI\inv $, which squares to the identity map: even
elements commute with~$I$ and odd elements anticommute.  If $k=\RR$ and
$I$~is a complex structure on~$V$, then Definition~\ref{thm:152} reduces to
Definition~\ref{thm:15} (with $\tau $~trivial).  In this case even elements
of~$\End_{\RR}(V)$ are complex linear and odd elements are complex
antilinear.  Note that each~$\rho (g)$ commutes or anticommutes with~$I$.
The sign in~\eqref{eq:166} is then the usual Koszul sign.  In~(iii) recall
that a \emph{proper} subspace of~$V$ is neither~0 nor~$V$.
 
  \begin{definition}[]\label{thm:153}
 A super algebra~$D $ over~$k$ is a \emph{super division algebra} if every
nonzero homogeneous element is invertible. 
  \end{definition}

\noindent 
 It is not true that every non-homogeneous element is invertible.

  \begin{theorem}[super Schur]\label{thm:154}
 Let $(G,\phi )$~be a $\zt$-graded group and $\rho \:G\to\Aut_k(V,I)$ an
irreducible corepresentation on~$(V,I)$.  Then the graded commutant
  \begin{equation}\label{eq:167}
  \begin{aligned}
     D =\{a^0+a^1\in \End_k(V): \quad a^0\rho (g)&=\phantom{\phi
     (g)}\rho (g)a^0\\  
a^1\rho (g)&=\phi (g)\rho (g)a^1 \textnormal{\quad for all $g\in G$}\}  
  \end{aligned}
  \end{equation}
is a super division algebra over~$k$. 
  \end{theorem}

\noindent 
 In~\eqref{eq:167} the general element of~$\End_k(V)$ is written as a
sum~$a^0+a^1$ of even and odd elements: $a^0$~commutes with~$I$ and
$a^1$~anticommmutes with~$I$.

  \begin{proof}
 As in the classical case---Theorem~\ref{thm:151}---the kernel and image of a
homogeneous element~$a\in D $ are $I$-invariant and $G$-invariant subspaces
of~$V$, so by irreducibility are either~$V$ or~$0$.
  \end{proof}

We can classify a corepresentation by the isomorphism class of its graded
commutant.  For~$k=\RR$ there are 10~possibilities~\cite{Wa}, and they are
exhibited in the first line of the table below.  Note that 9~of them are
Clifford algebras---the real and complex Clifford algebras of smallest
dimension---but the 10th~($\HH$) is not.  Also, the entries~$\CC,\RR,\HH$ are
purely even, and hence are the usual (ungraded) division algebras over the
reals.

  \begin{remark}[]\label{thm:155}
 The set of isomorphism classes of (super) division algebras over~$k$ may be
identified with the set of (super) simple algebras over~$k$ up to Morita
equivalence (see Appendix~\ref{sec:12}).  If we fix the (graded) center then
they form an abelian group---the \emph{Brauer group}---under tensor product.
For~$k=\RR$ the (graded) center is either~$\RR$ or~$\CC$.  In the graded case
there are 8~isomorphism classes of super division algebras with center~$\RR$
and 2~isomorphism classes with center~$\CC$.  The Brauer groups are cyclic,
and each Morita equivalence class of super central simple algebras is
represented by a Clifford algebra. 
  \end{remark}

  \begin{example}[]\label{thm:156}
 Let $G=\zt$ with the nontrivial grading~$\phi $.  Let $V=\RR^2$ with its
usual complex structure~$I$, which we use to identify~$\RR^2=\CC$.  Define a
$\phi $-twisted representation~$\rho $ by letting $\rho (g)z=\bar z$, where
$g\in G$ is the nonzero element and $z\in \CC$.  Let $D =D ^0\oplus D
^1$ be the graded commutant.  Then $D ^0$~consists of real scalar
endomorphisms $z\mapsto rz$ ($r\in \RR$) and $D ^1$~consists of
endomorphisms $z\mapsto si\bar z$ ($s\in \RR$).  So $D \cong \Cl1$.  There
are many more examples in~\cite[\S V]{Dy}.
  \end{example}

The \CTt s studied in Proposition~\ref{thm:24} and Proposition~\ref{thm:57}
correspond to super division algebras over~$\RR$ as follows.  A \CTt\ is a
pair~$(A,\tA)$ consisting of a subgroup $A\subset \pmo\times \pmo$ and a
twisted extension $\TT\to\tA\to A$.  Recall that $A$~may or may not contain
the generators~$\bT,\bC$ of~$\pmo\times \pmo$, which have special
lifts~$T,C\in \tA$ which satisfy~$T^2=\pm1,\,C^2=\pm1,\, TC=CT$.  Given a
\CTt~$(A,\tA)$, the construction beginning in~\eqref{eq:78} gives a
hermitian line bundle $\tL\to A$ such that $\tD=\oplus _{a\in A}\tL_a$ is a
super algebra which contains~$\tA$.  Note that the multiplication is not
complex linear if the homomorphism $\phi \:A\to\pmo$ is nontrivial.  Also,
the homomorphism $c\:A\to\pmo$ is used to grade~$\tD$.  Define a super
algebra~$D$ as follows.  If $\bT\notin A$ set~$D=\tD$; it is a \emph{complex}
division algebra.  If $\bT\in A$ and~$T^2=-1$ let $D=\tD_{\RR}$ be the
underlying real super algebra.  If $\bT\in A$ and~$T^2=+1$, then conjugation
by~$T$ is a real structure on~$\tD$, and we let~$D$ be the real points of the
subalgebra of~$\tD$ generated by~$C$, if $\bC\in A$; otherwise $D$~consists
only of real scalars.
 
We summarize some of the various 10-fold ways in a compact table:

\bigskip
  \begin{center}
  \renewcommand{\arraystretch}{1.5}
  \begin{tabular}{|c||c|c|c|c|c|c|c|c|c|c|}
\hline 
sDivAlg&$\CC$&$\Clc1$&$\RR$&$\Cl{-1}$&$\Cl{-2}$&$\Cl{-3}$&$\HH$&$\Cl3$&$\Cl2$&$\Cl1$
\\ 
 \hline
CT&$1$& diag& $+\,0$& $+-$& $0\,-$& $--$&$-\,0$& $-+$&
$0\,+$& $++$ \\   
\hline
Dyson&$\CC\CC_1$&$\CC\CC_2$&$\RR\CC$&$\RR\HH$&$\CC\HH$&$\HH\HH$&$\HH\CC$&
$\HH\RR$&  $\CC\RR$ & $\RR\RR$\\ 
 \hline 
  \end{tabular}
  \renewcommand{\arraystretch}{1}
  \end{center}
\vskip12pt

\noindent
 The first row enumerates the 10~super division algebras which arise from
Theorem~\ref{thm:154} with~$k=\RR$.  As mentioned above, all but~$\HH$ are
Clifford algebras.  Note that $\RR,\CC,\HH$ are purely even and are the
classical division algebras over~$\RR$.  The second row enumerates the 10
\CTt s in a compact notation: the first two entries indicate the subgroup
$A\subset \pmo\times \pmo$ and the last~8 tell the status of~$T,C\in \tA$:
`0'~if not present, `+'~if it squares to~$+1$, and `$-$'~if it squares
to~$-1$.  The columns line up with those of the table in
Proposition~\ref{thm:57}, which gives the same data in a more elaborate
form.  The final row of the present table gives Dyson's labels for the
10~types.  For the reader's convenience we present a quick resum\'e of the
first half of~\cite{Dy} where these labels are defined. 
 
Dyson's analysis is based on the classical Theorem~\ref{thm:151} and the
Wedderburn theorem, which asserts that a finite dimensional simple algebra
over a field~$k$ is a matrix algebra over a division algebra.  For~$k=\RR$,
by the Frobenius theorem the algebra must be a matrix algebra over one of the
division algebras~$\RR$, $\CC$, or~$\HH$, and we call this division algebra
its \emph{Wedderburn type}.  Let $(G,\phi )$~be a $\zt$-graded group as in
Definition~\ref{thm:152}.  Dyson assumes that the homomorphism~$\phi $ is
surjective, so the kernel~$G_0\subset G$ of~$\phi $ is an index~2 subgroup.
Let $V$~be an irreducible co-representation, so a complex vector space on
which elements of $G_0$~act complex linearly and elements of $G\setminus G_0$
act complex antilinearly.  Let\footnote{In this paragraph `$A$'~and `$D$' are
as used in~\cite{Dy}, not as used in the previous text.} $A$~be the
subalgebra of~$\End_{\RR}(V)$ generated by~$\rho (g_0),\,g_0\in G_0$ and
$D$~the subalgebra of~$\End_{\RR}(V)$ generated by~$\rho (g),\,g\in G$ and
the complex structure~$I$.  If the co-representation is irreducible, then
$D$~is a simple algebra, so by the Wedderburn theorem is a matrix algebra
over~$\RR$, $\CC$, or~$\HH$.  The algebra~$A$ is not simple, since the real
representation of~$G_0$ on~$V_{\RR}$ may be reducible, but Dyson proves that
the irreducible summands are all of the same Wedderburn type.  He further
shows that the Wedderburn types of~$D$ and~$A$ are independent in that all
9~possibilities can occur.  The Dyson type is the conjunction of the
Wedderburn types of~$D$ and~$A$.  Thus, for example, for Dyson type~$\HH\CC$
the algebra $D$~is quaternionic and the algebra $A$~is complex.  The
type~$\CC\CC$ splits into two subtypes according as the complex
representation of~$G_0$ on~$V$ splits into the sum of two non-isomorphic or
isomorphic irreducible representations.  In his analysis Dyson studies the
ungraded commutator algebras of~$A$ and~$D$, denoted~$X$ and~$Z$, and they
have the same Wedderburn types as~$A$ and~$D$, respectively.  The graded
commutant of the algebra generated by $\rho(g)$, $g\in G$, which is the super
division algebra we associate with the representation~$\rho $, lies between
Dyson's algebras $Z$ and $X$ and uniquely characterizes the Dyson type.  We
highly recommend the beautiful paper~\cite{Dy} to the reader.

   \section{Continuous families of quantum systems}\label{sec:14}

For a finite dimensional complex vector space~$V$ there is only one
reasonable topology on the algebra~$\End V$ of linear maps~$V\to V$.  The
invertible operators~$\Aut V\subset \End V$ inherit a topology in which
inversion and multiplication are continuous: $\Aut V$~is a topological group.
The situation is quite different when $V$~is an infinite dimensional
topological vector space, as then there are several possible topologies on
the vector space of continuous linear operators.  In this appendix we discuss
the compact-open topology for $V$~a Hilbert space.  This is the crucial
ingredient in our definition of a continuous family of quantum systems, and
so too in the definition of homotopy invariants of quantum systems.  The main
references for our discussion are~\cite[Appendices]{AS} and~\cite{RS}; we
follow~\cite{AS} particularly closely in the next subsection.
 
  \subsection*{Topologies on mapping spaces}

Quite generally, if $Y$~is a topological space and $Z$~a metric space, then
there are several natural topologies on the set~$\Map(Y,Z)$ of continuous
maps $f\:Y\to Z$.  In the \emph{topology of pointwise convergence}, or
equivalently \emph{point-open topology}, a sequence $f_n\to f$ iff $f_n(y)\to
f(y)$ for all~$y\in Y$.  In the finer\footnote{A topology~$\tau $ on a
set~$X$ is a collection of subsets of~$X$.  A topology~$\tau '$ is finer, or
stronger, than a topology~$\tau $ if $\tau \subset \tau '$.} \emph{topology
of uniform convergence on compact sets}, or equivalently \emph{compact-open
topology},\footnote{The compact-open topology is defined for $Z$~a general
topological space, not necessarily a metric space.} a sequence $f_n\to f$ iff
for each compact subset~$C\subset Y$, the convergence $f_n(y)\to f(y)$ is
uniform for~$y\in C$.  If $X$~is a topological space, and $F\:X\times Y\to Z$
a continuous map, then $x\mapsto (y\mapsto F(x,y))$ is a continuous map
from~$X$ to $\Map(Y,Z)$ with the compact-open topology.
 
If $Y=Z=\sH$ is an {infinite dimensional} separable complex Hilbert space,
and $\End(\sH)\subset \Map(\sH,\sH)$ the set of continuous, hence bounded,
linear maps $\sH\to\sH$, then the topology of pointwise convergence is called
the \emph{strong operator topology}.  The compact-open topology is finer, but
not so different: as shown in~\cite[Appendix~1]{AS} the compact sets in the
two topologies coincide.  It follows that if $X$~is a \emph{compactly
generated}\footnote{A topological space~$X$ is compactly generated if a
subset~$A\subset X$ is closed if and only if $A\cap C$~is closed for all
compact subsets~$C\subset X$.  If $X$~is compactly generated, then one can
test continuity of maps with domain~$X$ on compact subsets of~$X$.
Topological spaces whose topology can be defined by a metric are called
\emph{metrizable}, and they are compactly generated.  See~\cite[\S6.1]{DK}
for a summary of Steenrod's classic paper~\cite{St} on compactly generated
spaces.}  topological space the continuous maps $X\to\End(\sH)$ are the same
in both topologies.  We write $\End(\sH)\co$ for the space of bounded linear
operators with the compact-open topology.  There is a finer \emph{norm
topology} on~$\End(\sH)$ which is often used; then $f_n\to f$ iff $f_n(x)\to
f(x)$ uniformly for $x$~in the unit sphere of~$\sH$.  In other words, $f_n\to
f$ in the metric topology of the operator norm: $\|f-f_n\|\to 0$.
 
We choose to work with the compact-open topology on~$\End(\sH)$ for several
reasons:
  \begin{enumerate}
 \item If $H$~is an \emph{unbounded} self-adjoint operator on~$\sH$, then the
one-parameter group $t\mapsto e^{-itH/\hbar}$ is continuous in the
compact-open topology but not in the norm topology.

 \item A fiber bundle of Hilbert spaces (Definition~\ref{thm:110}) has
transition functions which are continuous in the compact-open topology.

 \item The projection-valued measure associated to a self-adjoint operator is
countably additive in the compact-open topology~\cite[\S\S
VII.3,VIII.3]{RS}. 

 \item While the norm topology is too strong for many purposes, the
\emph{weak operator topology} is too weak to control the
spectrum.\footnote{For example, if $\{e_n\}_{n=1}^{\infty}$ is an orthonormal
basis of~$\sH$, and $f_n$~orthogonal projection onto the line spanned
by~$e_1+e_n$, then $\{f_n\}$~converges weakly to $1/2$~times orthogonal
projection onto the span of~$e_1$.  This exhibits an operator with
spectrum~$\{0,1/2\}$ as a weak limit of operators with spectrum~$\{0,1\}$.}
  \end{enumerate}

\noindent 
 We warn that the compact-open topology is not sufficiently strong for
problems involving index theory: the space of Fredholm operators is
contractible in the compact-open topology~\cite[Theorem~A2.1]{AS}.   
 
There are other difficulties with the compact-open
topology~\cite[Appendix~1]{AS}, assuming $\sH$~is infinite dimensional, which
stem from the fact that $f\mapsto f\inv $ is \emph{not} a continuous map
$\Aut(\sH)\to\Aut(\sH)$ in the induced topology on the subset
$\Aut(\sH)\subset \End(\sH)$ of invertible operators.

  \begin{definition}[]\label{thm:107}
 The topology on the set~$\Aut(\sH)$ induced from 
  \begin{equation}\label{eq:114}
     \begin{aligned} \Aut(\sH)&\longrightarrow \End(\sH)\co\times
      \End(\sH)\co \\ f\quad &\longmapsto \quad (f, f\inv )\end{aligned}  
  \end{equation}
is called the \emph{compact-open topology} and is denoted~$\Aut(\sH)\co$.
The compact-open topology on the unitary group~$U(\sH)$ is the topology
induced by the inclusion $U(\sH)\subset \Aut(\sH)\co$. 
  \end{definition}

\noindent
 Even though inversion $\Aut(\sH)\co\to\Aut(\sH)\co$ is continuous, by design
using Definition~\ref{thm:107}, multiplication $\Aut(\sH)\co\times
\Aut(\sH)\co\to\Aut(\sH)\co$ is \emph{not} continuous, if $\sH$~is infinite
dimensional, whence $\Aut(\sH)\co$~is not a topological group.  But if $X$~is
compactly generated, then the set of continuous maps
$\Map\bigl(X,\Aut(\sH)\co \bigr)$ is a topological group in the compact-open
topology.  The spaces $\Aut(\sH)\co$~ and $U(\sH)\co$~are
contractible~\cite[Theorem~A2.1]{AS}.

  \begin{proposition}[]\label{thm:108}
 The subgroup $\QAut(\sH)\subset \Aut(\sH)\co$ has two contractible
components in the compact-open topology.
  \end{proposition}

  \begin{proof}
 Each $S\subset \QAut(\sH)$ is real linear and $S(i\psi )=\pm i\psi $ for
all~$\psi \in \sH$ and one choice of sign.  The sign is unchanged under
limits $S_n\to S$ in the point-open, hence compact-open, topology. 
  \end{proof}

  \subsection*{Hilbert bundles and continuous families of quantum systems}

We begin by defining a representation of a group as a continuous
\emph{action} on a vector space.

  \begin{definition}[]\label{thm:109}
 Let $G$~be a topological group and $\sH$~a Hilbert space.  A
\emph{representation} of~$G$ on~$\sH$ is a continuous linear action $\alpha
\:G\times \sH\to\sH$.  The representation is \emph{unitary} if the linear
operator~$\alpha (g,-)$ is unitary for all~$g\in G$.   
  \end{definition}

\noindent
 The (left) action property is 
  \begin{equation}\label{eq:115}
     \alpha \bigl(g_1,\alpha (g_2,\psi ) \bigr)= \alpha (g_1g_2,\psi ),\qquad
     g_1,g_2\in G,\quad \psi \in \sH. 
  \end{equation}
The continuity of~$\alpha $ implies that the induced map $\rho
\:G\to\Aut(\sH)\co$ is continuous.  Conversely, a continuous map $\rho
\:G\to\Aut(\sH)\co$ is a representation if the topology of $G$~is compactly
generated.\footnote{We remind that metrizable spaces are compactly generated,
so this holds for Lie groups~$G$.}  We have already remarked that
one-parameter groups ($G=\RR$) of unitary operators generated by an unbounded
self-adjoint operator are continuous in the sense of
Definition~\ref{thm:109}.  

The compact-open topology is also natural for fiber bundles.

  \begin{definition}[]\label{thm:110}
 Let $X,\sE$~be topological spaces and $\pi \:\sE\to X$ a continuous map.
Suppose each fiber $\sE_x=\pi \inv (x)$, $x\in X$, is endowed with the
structure of a complex Hilbert space.  Then $\pi $~is a \emph{Hilbert bundle}
if it is locally trivial, i.e., if about each~$x\in X$ there exists an open
set~$U\subset X$, a Hilbert space~$\sH$, and a homeomorphism~$\varphi $ which
fits into the commutative diagram 
  \begin{equation}\label{eq:116}
     \xymatrix{\pi \inv (U)\ar[rr]^{\varphi }\ar[dr]_{\pi }&&U\times
     \sH\ar[dl]^{\pi _1}\\ &U} 
  \end{equation}
and preserves the Hilbert space structure on each fiber. 
  \end{definition}

\noindent
 Here $\pi _1$~is projection onto the first factor.  The
homeomorphism~$\varphi $ is a \emph{local trivialization}, and the transition
function between two local trivializations is a continuous map
into~$U(\sH)\co$.  Conversely, if $X$~is compactly generated a fiber
bundle can be constructed from coherent continuous transition functions
into~$U(\sH)\co$.  A Hilbert bundle with infinite dimensional fibers is
trivializable, but may not come with a natural trivialization (e.g.,
Proposition~\ref{thm:115}). 

  \begin{definition}[]\label{thm:111}
 Let $G$~be a topological group and $X$~a topological space.  Then a
\emph{continuous family of unitary representations of~$G$ parametrized
by~$X$} is a Hilbert bundle $\pi \:\sE\to X$ and a continuous action $\alpha
\:G\times \sE\to\sE$ which preserves fibers: the diagram
  \begin{equation}\label{eq:117}
     \xymatrix{G\times \sE\ar[rr]^{\alpha  }\ar[dr]_{\pi \circ \pi _2}&&\sE
     \ar[dl]^{\pi }\\ &X} 
  \end{equation}
commutes. 
  \end{definition}

\noindent
 If $\sE=X\times \sH\to X$ is the trivial bundle with fiber a fixed Hilbert
space~$\sH$, then $\alpha $~defines a continuous map $G\times X\to
U(\sH)\co$.  Conversely, if $G$~and $X$~are compactly generated, then a
continuous map $G\times X\to U(\sH)\co$ which is a homomorphism for
each~$x\in X$ determines a continuous family of representations.
 
A specialization of Definition~\ref{thm:111} applies to a family of
Hamiltonians parametrized by~$X$.  We use the equivalence between
self-adjoint operators and one-parameter groups of unitary operators given by
Stone's theorem~\cite[\S VIII.4]{RS}.

  \begin{definition}[]\label{thm:112}
 A family of (possibly unbounded) self-adjoint operators~$\{H_x\}$, $x\in X$,
acting on a Hilbert bundle $\sE\to X$ is \emph{continuous} iff the associated
family $t\mapsto e^{-itH_x/\hbar}$ of unitary one-parameter groups is a
continuous family of unitary representations of~$\RR$ on $\sE\to X$.
  \end{definition}

Recall that $\lambda \in \CC$ is in the \emph{resolvent set} of a
self-adjoint operator~$H$ if $\lambda -H\:\domain(H)\to\sH$ is bijective with
bounded inverse $R_\lambda (H)$, the \emph{resolvent}.  Then $\lambda \mapsto
R_\lambda (H)$ is a holomorphic function whose domain is the resolvent
set---the complement of the spectrum---so includes $\CC\setminus \RR$ since
$H$~is self-adjoint, and whose codomain is the space of bounded self-adjoint
operators.  A family~$\{H_x\}_{x\in X}$ of self-adjoint operators on a fixed
Hilbert space~$\sH$ has a \emph{common core}~$D\subset \sH$ if $D\subset \sH$
is dense, $D\subset \domain(H_x)$ for all~$x\in X$, and the closure of the
restriction of~$H_x$ to~$D$ is~$H_x$.

  \begin{proposition}[{\cite[\S VIII.7]{RS}}]\label{thm:148}
 Let $X$~be a compactly generated topological space, $\sE\to X$ a Hilbert
bundle, and $\{H_x\}_{x\in X}$ a family of self-adjoint operators acting on
the fibers of~$\sE\to X$. 
  \begin{enumerate}
 \item $\{H_x\}$~ is a continuous family iff the associated resolvent
function $(\lambda ,x)\mapsto R_\lambda (H_x)$ is continuous in the
compact-open topology for~$\lambda \in \CC\setminus \RR,\;x\in X$.

 \item Suppose $\sE\to X$ is the trivial bundle with fiber~$\sH$ and the
operators~$H_x,\;x\in X$ have a common core~$D$.  Then if for each convergent
sequence $x_n\to x$ in~$X$ and~$\psi \in D$ we have $H_{x_n}\psi \to
H_x\psi $, then $\{H_x\}$~is a continuous family.

  \end{enumerate}
  \end{proposition}

\noindent
 We remind that we can substitute the strong topology for the compact-open
topology.

 Next we give a precise version of Definition~\ref{thm:26}.  Recall from
Remark~\ref{thm:98} that an extended \Qsymmetry class $(G,\phi ,\tau ,c)$ has a
canonically associated \Qsymmetry class~$(\tilG,\tp,\tilde\tau)$ which includes
time translations.

  \begin{definition}[]\label{thm:113}
 Let $X$~be a topological space.  A \emph{continuous family of gapped systems
with extended \Qsymmetry class~$(G,\phi ,\tau ,c)$} is a continuous family of
$(\tp,\tilde\tau )$-twisted representations of~$\tilG$ parametrized by~$X$
such that for all~$x\in X$ the Hamiltonian~$H_x$ does not have~0 in its
spectrum. 
  \end{definition}

\noindent
 In other words, there is a Hilbert bundle $\pi \:\sE\to X$ and a continuous
action $\tilG^\tau \times \sE\to\sE$ which on each fiber of~$\pi $ is a
$(\tp,\tilde\tau )$-twisted representation of~$\tilG$.  The Hamiltonian~$H_x$
is the self-adjoint generator of the one-parameter group $\RR\subset
\tilG^\tau \to U(\sE_x)\co$.

In~\S\ref{sec:6} we use the positive and negative spectral projections
associated to Hamiltonians which do not contain~0 in their spectrum.

  \begin{proposition}[]\label{thm:149}
 Let $X$~be a compactly generated topological space, $\sE\to X$ a Hilbert
bundle, and $\{H_x\}$~a continuous family of self-adjoint operators acting
on~$\sE$.  Assume 0~is not in the spectrum of~$H_x$ for all~$x\in X$.  Then
  \begin{enumerate}
 \item the positive spectral projections~$\{\chi ^+(H_x)\}$ and negative
spectral projections~$\{\chi ^-(H_x)\}$ form continuous families of bounded
self-adjoint operators; and

 \item if the negative spectral projections have finite rank, then the images
of~$\chi ^+$ and $\chi ^-$ are sub-Hilbert bundles~$\sEp$ and~$\sEm$. 

  \end{enumerate}

  \end{proposition}
 
\noindent 
 A corollary of the proof of~(ii) is that if $\{E_x\subset \sE_x\}_{x\in X}$
is the image of a family of finite rank spectral projections~$\chi ^{[a,b]}$
with neither~$a$ nor~$b$ in the spectrum of each~$H_x$, then $E\to X$ is a
locally trivial vector bundle.

  \begin{proof}
 The statements are local and by the local triviality of~$\sE\to X$ we may
assume $H_x$~acts on a constant Hilbert space~$\sH$.  Let $\Gamma\:\RR\to\CC$
be a smooth embedding such that $\Gamma (0)=0$ and $\Gamma (-t)=t+i$, $\Gamma
(t)=t-i$ for~$t>1$.  Then if $H$~is self-adjoint and 0~is not in the
spectrum, the spectral calculus implies
  \begin{equation}\label{eq:161}
     \chi ^+(H) = \int_{\Gamma }d\lambda \,R_\lambda (H),
  \end{equation}
where the integral is defined as~$\int_{\Gamma }d\lambda
=\lim\limits_{M\to\infty }\int_{-M}^Mdt$.  Let $H_n\to H$ be a convergent
sequence in the sense of~Definition~\ref{thm:112}, so that by
Proposition~\ref{thm:148} we have $R_{\,\Gamma (t)}(H_n)\to R_{\,\Gamma
(t)}(H)$ in the strong topology for all~$t\not= 0$.  For any vector~$\psi $
the norm~$\|R_{\,\Gamma (t)}(H)\psi \|$~is uniformly bounded, since it is
bounded by~$\|\psi \|$ for~$|t|>1$.  The Lebesgue dominated convergence
theorem now implies $\chi ^+(H_n)\to \chi ^+(H)$ in the strong topology.
This proves~(i).
 
For~(ii) we must prove that the finite dimensional vector
spaces~$\{\sEm_x=\chi ^-(H_x)(\sH)\}$ form a locally trivial family of
subspaces of~$\sH$.  Fix~$x_0\in X$ and let $\pi _x\:\sEm_{x_0}\to\sEm_x$ be
orthogonal projection, i.e., the restriction of~$\chi ^{-}(H_x)$
to~$\sEm_{x_0}$.  We claim that for $x$~in an open neighborhood of~$x_0$ the
map~$\pi _x$ is an isomorphism, hence a local trivialization after applying
Gram-Schmidt to make it an isometry.  The claim follows since by strong
continuity the matrix of inner products $\left(\begin{matrix} \langle \chi
^-(H_x)e_i,\chi ^-(H_x)e_j \rangle \end{matrix}\right)_{i,j=1\dots N}$ is
invertible for $x$~sufficiently close to~$x_0$, where $\{e_i\}_{i=1}^N$ is an
orthonormal basis of~$\sEm_{x_0}$.
  \end{proof}

  \subsection*{Fourier transform (Bloch sums) and the Berry connection} 
 
Next, we prove that Hypothesis~\ref{thm:56}(ii) is satisfied in the standard
situation described at the beginning of~\S\ref{sec:6}.  We need the following
preliminary.

  \begin{lemma}[]\label{thm:114}
 Let $X,Y$ be locally compact Hausdorff topological spaces with Borel
measures, and assume that $Y$~is compact.  Let $\sL\to X\times Y$ be a
hermitian line bundle.  Then the Hilbert spaces
$\sE_x=L^2\bigl(Y;\sL\res{\{x\}\times Y} \bigr)$, $x\in X$, are the fibers of
a Hilbert bundle $\sE\to X$ and there exists an isomorphism of Hilbert spaces
$L^2(X\times Y;\sL)\cong L^2(X;\sE)$.
  \end{lemma}

  \begin{proof}
 About each~$x\in X$ is an open neighborhood $U\subset X$ together with a
continuous isomorphism $\phi \:\sL\res{U\times Y}\to U\times
\sL\res{\{x\}\times Y}$.  Then $\phi $~induces a local trivialization
$\varphi \:\sE\res U\to U\times \sE_x$ which we use to topologize~$\sE$.  To
prove the resulting transition functions are continuous, use the fact that if
$h\:U\times Y\to\CC^\times $ is continuous and $\sK\to Y$ a hermitian line
bundle, then 
  \begin{equation}\label{eq:118}
     x\mapsto \textnormal{multiplication by $h(x,-)$ on $L^2(Y;\sK)$}
  \end{equation}
is continuous in the compact-open topology.
  \end{proof}

 Let $E$~be a Euclidean space of dimension~$d$ and $\Pi $~a full lattice of
translations of~$E$.  The Pontrjagin dual~$\XL$ is the torus of unitary
characters $\lambda \:\Pi \to\TT$.  The affine torus~$E/\Pi $ is a torsor
over the dual torus to~$\XL$.  The lattice~$\Pi $ acts on the Hilbert
space~$L^2(E;\CC)$ of $L^2$~functions with respect to Lebesgue measure, so
determines a sheaf~ $\sS$ of Hilbert spaces, by Lemma~\ref{thm:83}.

  \begin{proposition}[]\label{thm:115}
 $\sS$~is the sheaf of sections of a Hilbert bundle $\sE\to \XL$.
  \end{proposition}

  \begin{proof}
 Let $\sL\to \XL\times E/\Pi $ be the ``Poincar\'e line bundle'' whose
sections are functions $\hf\:\XL\times E\to\CC$ which satisfy the
quasi-periodicity condition
  \begin{equation}\label{eq:119}
     \hf(\lambda ,x+\xi ) = \lambda (\xi )\hf(\lambda ,x),\qquad \xi \in \Pi
     . 
  \end{equation}
Fourier transform defines an isomorphism of Hilbert spaces 
  \begin{equation}\label{eq:120}
     L^2(E;\CC)\longrightarrow L^2(\XL\times E/\Pi ;\sL) 
  \end{equation}
which takes $f\in L^2(E;\CC)$ to 
  \begin{equation}\label{eq:121}
     \hf(\lambda ,x) = \sum\limits_{\xi '\in \Pi }\lambda (\xi ')\inv f(x+\xi
     '); 
  \end{equation}
the inverse maps $\hf\in L^2(\XL\times E/\Pi ;\sL)$ to 
  \begin{equation}\label{eq:122}
     f(x) = \int_{\XL}d\lambda \,\hf(\lambda ,x), 
  \end{equation}
where $d\lambda $~is Haar measure with unit total volume.
Lemma~\ref{thm:114} identifies the codomain of~ \eqref{eq:120} with the space
of $L^2$~sections of the Hilbert bundle $\sE\to \XL$ whose fiber at~$\lambda
\in \XL$ is
  \begin{equation}\label{eq:131}
      \sE_{\lambda }=L^2\bigl(E/\Pi ;\sL_\lambda \bigr) ,
  \end{equation}
where $\sL_\lambda =\sL\res{\{\lambda \}\times E/\Pi } $.
  \end{proof}

  \begin{remark}[]\label{thm:116}
 In the condensed matter physics literature \eqref{eq:119}~is known as the
\emph{Bloch wave condition} and \eqref{eq:121} is known as a \emph{Bloch
sum}. 
  \end{remark}

The Hilbert bundle $\sE\to\XL$ is in fact smooth and can be built from
locally constant transition functions.  Let $V$~denote the vector space of
translations of~$E$, so $\Pi \subset V$, and let $\Pi ^*\subset V^*$ be the
dual lattice of linear functionals $k\:V\to\RR$ such that $\langle k,\xi
\rangle\in \ZZ$ for all~$\xi \in \Pi$.  Here the bracket is the pairing
between~$V^*$ and~$V$.  We identify $H_1(\XL;\ZZ)\cong \Pi ^*$.

  \begin{proposition}[]\label{thm:141}
 For each choice of origin~$x_0\in E$ there is a flat covariant
derivative~$\nabla ^{\sE}$ on $\sE\to\XL$ whose holonomy around a loop with
homology class $k\in \Pi ^*$ is multiplication by the periodic function
  \begin{equation}\label{eq:151}
     h_{\pi ^*}(x)=e^{2\pi i\langle k,x-x_0  \rangle},\qquad x\in E.
  \end{equation} 
  \end{proposition}

  \begin{proof}
 Let $\hf_0\:E\to\CC$ be quasi-periodic for fixed~$\lambda _0$, as
in~\eqref{eq:119}, and let $\{\lambda _t\:0\le t\le 1\}$, be a smooth path
in~$\XL$.  Identify $\XL=V^*/\Pi ^*$ and choose a lift~$k_0\in V^*$
of~$\lambda _0$; then there is a unique smooth lift~$\{k_t\}\subset V^*$ of
the path~$\{\lambda _t\}\subset \XL$.  Define the parallel transport
of~$\hf_0$ along~$\{\lambda _t\}$ to be the path of functions
  \begin{equation}\label{eq:152}
     \hf_t(x) = e^{2\pi i\langle k_t-k_0,\,x-x_0\rangle}\hf_0(x),\qquad x\in
     E,\quad 0\le t\le1.
  \end{equation}
It is independent of the choice of~$k_0$ and defines the parallel transport
of a flat covariant derivative with the stated holonomy.
  \end{proof}

  \begin{remark}[]\label{thm:145}
 The dependence on the origin~$x_0$ is easily worked out from~\eqref{eq:151}
and~\eqref{eq:152} and is the familiar indeterminacy of the covariant
derivative on the Poincar\'e line bundle: we can tensor with a flat covariant
derivative pulled back from~$\XL$. 
  \end{remark}

  \begin{remark}[]\label{thm:144}
 If $\Pi $~is extended to a larger symmetry group~$G$, as in \eqref{eq:67},
then the quotient $G''$~acts on $\XL$ and with a possible central extension
the groupoid $\XL\gpd G''$ lifts to~$\sE\to\XL$.  We claim that this lifted
action commutes with the covariant derivative~$\nabla ^{\sE}$.  This follows
simply by observing that it preserves the parallel transport~\eqref{eq:152},
since the action of~$G''$ preserves the pairing $\langle -,-
\rangle\:V^*\otimes V\to\RR$.
  \end{remark}

In band theory one encounters smooth finite rank subbundles~$F\subset \sE$.
Let $F\xrightarrow{i}\sE\xrightarrow{p}F$ denote the inclusion and orthogonal
projection.  Then the compression $p\circ \nabla ^{\sE}\circ i$ is a
covariant derivative on~$F\to\XL$ called the \emph{Berry connection}.  If the
``Bloch Hamiltonians'' on the fibers of $\sE\to\XL$ are \emph{elliptic}
differential operators, and if $F\to\XL$ is obtained by finite rank spectral
projection, then $F\to\XL$ is a smooth vector bundle.  (That it is a
\emph{continuous} vector bundle follows from the proof of
Proposition~\ref{thm:149}.)  The proof, which we omit, uses elliptic
regularity.

Finally, we prove that differential operators such
as~\eqref{eq:134}---elliptic or not---lead to continuous families of Bloch
Hamiltonians.  Let $W$~be a finite dimensional hermitian vector space, $M$~a
smooth manifold, and $\{H_m\}_{m\in M}$ a smooth family of self-adjoint
differential operators of fixed (finite) order~$k\in \ZZ^{\ge0}$ acting on
functions $E\to W$.  Explicitly, let $x^1,\dots ,x^n$ be coordinates on~$E$.
Then we can write
  \begin{equation}\label{eq:162}
     H_m = \sum\limits_{q=0}^K\,(\sqrt{-1})^q\,a\mstrut _{i_1\cdots
     i_q}(m;x^1,\dots 
     ,x^n)\,\frac{\partial ^q}{\partial x^{i_1}\cdots \partial
     x^{i_q}},\qquad m\in M, 
  \end{equation}
where $a\mstrut _{i_1\cdots i_q}\:M\times E\to\End(W)$ is smooth and each
$a\mstrut _{i_1\cdots i_q}(m;x)$~is hermitian.  There is an implicit sum over
the symmetric multi-index~$i_1\cdots i_q$.  Assume $H_m$~is invariant under
translation by~$\Pi $, i.e., $a\mstrut _{i_1\cdots i_q}(m;x+\xi )= a\mstrut
_{i_1\cdots i_q}(m;x)$ for all $m\in M,\;x\in E,\;\xi \in \Pi $.  Then for
each~$\lambda \in \XL$ there is an induced self-adjoint differential
operator~$H_{m,\lambda }$ on $W$-valued quasi-periodic
functions~\eqref{eq:119}, which are sections of the vector bundle
$\sL_\lambda \otimes W\to E/\Pi $.  Let $\sE\to M\times \XL$ be the Hilbert
bundle with fiber $\sE_{m,\lambda }=L^2(E/\Pi ;\sL_\lambda \otimes W)$.

  \begin{proposition}[]\label{thm:150}
 The family of self-adjoint operators $\{H_{m,\lambda }\}$ is continuous. 
  \end{proposition}

  \begin{proof}
 Relative to the local trivialization~\eqref{eq:152} we view~$H_{m,\lambda }$
as acting on 
  \begin{equation}\label{eq:163}
     D=\{e^{2\pi i\langle k_t-k_0,x-x_0 \rangle}\hf_0(x):\hf_0\:E\to W
     \textnormal{ is smooth with quasi-periodicity $\lambda _0$} \}.
  \end{equation}
Now apply Proposition~\ref{thm:148}(ii).
  \end{proof}

   \section{Twisted $K$-theory on orbifolds}\label{sec:19}

In this appendix we prove the theorem alluded to in Remark~\ref{thm:73}.  It
ensures that Definition~\ref{thm:71} reproduces standard twisted $K$-theory
for global quotients~$X\gpd G$, where $X$~is a nice compact space and $G$~a
finite group.  

Untwisted $K$-theory, equivariant or not, is defined in terms of finite rank
vector bundles~\cite{A1}.  It was realized early on (Atiyah-J\"anich) that
Fredholm operators also provide a model for $K$-theory, a crucial observation
for the Atiyah-Singer index theorem, for example.  By contrast, twisted
$K$-theory classes cannot in general be represented by finite rank bundles.
For example, if the twisting has infinite order in the abelian group of
twistings, then only the zero class has a finite rank representative.  Hence
models for twisted $K$-theory are built using Fredholm
operators~\cite{FHT1,AS}.  For the application to gapped topological
insulators~(\S\ref{sec:6}), specifically for Theorem~\ref{thm:131} and
Theorem~\ref{thm:136}, it is crucial that every twisted equivariant
$K$-theory class over the Brillouin torus have a finite rank representative.
This is what we sketch here.  We freely use \cite[Appendix]{FHT1}.
 
The starting point is \cite[Proposition~A.37]{FHT1}, which asserts that if
$\gp$~is a compact groupoid which is \emph{locally} of the form~$Y\gpd H$ for
$Y$~a compact space and $H$~a compact Lie group, then $\gp$~admits a
universal Hilbert bundle which is a sum of finite rank bundles if and only if
the groupoid~$\gp$ is locally equivalent to a \emph{global} quotient~$Y\gpd
H$.  The standard argument~\cite[Appendix]{A1} that a family of (untwisted)
Fredholm operators over a compact space is represented by a finite rank
$\zt$-graded vector bundle can be carried out on the groupoid~$\gp$ to prove
that if there is a universal Hilbert bundle which is a sum of finite rank
bundles, then every $K$-theory class on~$\gp$ is represented by a finite rank
$\zt$-graded vector bundle.  For twisted $K$-theory we use
\cite[Lemma~3.11]{FHT1} to reduce to the untwisted case on the centrally
extended groupoid.  Below we comment on the extra $\zt$-twists. 
 
The only new point, then, is the following.  Let $X$~be a nice topological
space with a continuous action of a finite group~$G$.  Set~$\gp=X\gpd G$.
Suppose $\tL\to \gp_1$ is a central extension, as in
Definition~\ref{thm:85}(i).  The elements of unit norm $\gp_1^\tau \subset
\tL$ are the morphisms of a groupoid~$\gp^\tau $ with $\gp^\tau _0=\gp_0=X$. 

  \begin{lemma}[]\label{thm:188}
 The groupoid~$\gp^\tau $ is locally equivalent to a global quotient~$Y\gpd
H$ of a compact space~$Y$ by a compact Lie group~$H$. 
  \end{lemma}

  \begin{proof}
 The action map $p\:\gp_1=X\times G\to X$ has finite fibers, so the
pushforward $p_*\tL\to X$ is a rank~$N$ hermitian vector bundle, where $N$~is
the cardinality of~$G$.  The fiber at~$x\in X$ is
  \begin{equation}\label{eq:215}
     \oplus \tL_{(p\xrightarrow{g} x)}
  \end{equation}
where the sum is over the finite set of arrows $(z\xrightarrow{g}x)$ with
target~$x$.  The bundle~$p_*\tL$ is equivariant for~$\gp_1^\tau $: the
$\TT$-torsor of unit norm elements in $\tL_{(x\xrightarrow{f}y)}$ acts via
the multiplication 
  \begin{equation}\label{eq:216}
     \tL_{(x\xrightarrow{f}y)}\otimes \tL
     _{(z\xrightarrow{g}x)}\longrightarrow \tL_{(z\xrightarrow{fg}y)} 
  \end{equation}
in the central extension. 
 
Let $P\to X$ be the $U_N$-bundle of frames associated to~$p_*\tL\to X$.  It is
also $\gp^\tau$-equivariant, and furthermore there is a unique arrow
in~$\gp^\tau$ between any two points of~$P$: elements in the central
extension~$\gp^\tau _1$ which cover non-identity elements in the stabilizer
$G_x\subset G$ at~$x\in X$ permute the lines~\eqref{eq:2} nontrivially.  So
the groupoid with objects~$P$ and arrows~$\gp^\tau$ is equivalent to a
space~$Y$, and the groupoid with objects~$X$ and arrows~$\gp^\tau$ is locally
equivalent to the global quotient $Y\gpd U_N$.
  \end{proof}

Lemma~\ref{thm:188} leads to a universal twisted Hilbert bundle over~$\gp$
which is a sum of finite rank bundles.  The argument given covers central
extensions of~$\gp$.  For \ptzgces, as in Definition~\ref{thm:85}(iii), we
first apply the argument of Lemma~\ref{thm:134} to account for a
nontrivial homomorphism~$\phi \:\gp_1\to\pmo$.  A similar argument applies to
nontrivial $c\:\gp_1\to\pmo$: replace the complex conjugate Hilbert bundle by
the parity-reversed Hilbert bundle.

   \section{Diamonds and dust}\label{sec:15}

In this appendix we give an illustration of the canonical \ptzgce\ of
Theorem~\ref{thm:16} for the three-dimensional diamond structure, a common
structure in materials science found in many materials, for example the
carbon column of the periodic table: diamond, silicon, germanium, and grey tin.
We will describe the pullback of the \ptzgce\ to the subvarieties of the
Brillouin torus $X_\Pi$ with nontrivial stabilizer group.  This appendix is
included partially to assist some readers in translating between the language
in the body of the text and more standard terminology in condensed matter
physics.

In the diamond structure the lattice $\Pi\subset V= \RR^3$ is the
face-centered-cubic (fcc) lattice.  This may also be viewed as the lattice
$D_3 \cong A_3$, the root lattice of the Lie algebra $so(6) \cong
su(4)$. Concretely,
\begin{equation}\label{eq:D3-latt}
D_3 = \{ (x_1, x_2 ,x_3) \in \ZZ^3 ~\vert ~ x_1 + x_2 + x_3 \equiv  0  \pmod 2
\}\subset \RR^3. 
\end{equation}
The ``reciprocal lattice'' lattice is, up to scale, (see below)
the  weight lattice $D_3^*$. It is also known as the body-centered-cubic
(bcc) lattice. Concretely it is
\begin{equation}\label{eq:bcc-latt}
D_3^* = D_3 \cup (D_3 + s) \cup (D_3 + v) \cup (D_3 + s'),
\end{equation}
where
\begin{equation}
\begin{split}
s & =   (\half, \half, \half)  \\
 v  & =   (0,0,1)  \\
 s'  & =   (\half, \half,- \half)  \\
\end{split}
\end{equation}
are spinor, vector, and conjugate spinor weights. The diamond crystal is
\begin{equation}
\bar C= D_3 \cup (D_3+s) \subset \EE^3 .
\end{equation}
Even if we choose a lattice point as an origin (thus identifying the affine
space $\EE^3$, which contains the crystal $\bar C$, with the vector space
$V\cong \RR^3$, which contains the lattice $\Pi$), $\bar C$ is not a lattice
since $2s \cong v$ in the discriminant group $D_3^*/D_3 \cong \ZZ/4\ZZ$. It
is useful to regard $2\bar C$ as the set of points $(x_1,x_2,x_3)\in \ZZ^3$
with $x_i\equiv x_j  \pmod 2$ and $\sum_i x_i \equiv  0,1  \pmod 4$.

In nature the fcc lattice $D_3$ comes scaled by $a/2$ where the fcc lattice
is constructed from a cubic lattice of side length $a$ and we add lattice
points at the centers of the sides of the cubes.  In condensed matter the
``reciprocal lattice'' which is in the (Fourier) dual space $V^*$ is then
scaled by $2/a$. The condensed matter convention is to take $ \sum n^j k_j
\in 2 \pi \ZZ$ where $(n^1,n^2,n^3) \in \Pi$ is in the space lattice and $k=
(k_1,k_2,k_3)$ is in the reciprocal lattice.  Therefore, in comparing to the
condensed matter literature we should scale $D_3^*$ by $\frac{2}{a}\times
2\pi$.  For simplicity we will work with $a=2$, and we define our reciprocal
vectors without the factor of $2\pi$.
 Thus, for us the Brillouin torus $X_\Pi$ is $\RR^3/D_3^*$. A standard
fundamental domain is the Wigner-Seitz domain whose closure is given by
points in $\RR^3$ closest to the lattice point $0$ (in the lattice
$D_3^*$). In mathematics this is known as the Voronoi cell. The Wigner-Seitz
cell for the fcc lattice is illustrated in Figure~1.
  \begin{figure}[h]
  \centering
  \includegraphics[scale=.3]{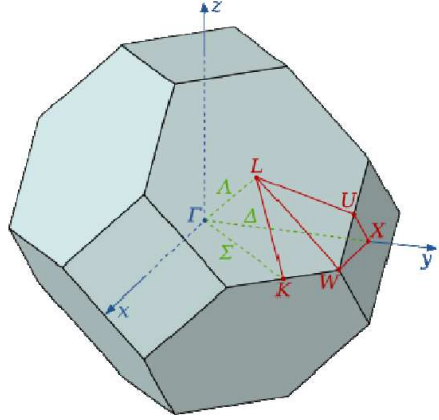}
  \caption{The Wigner-Seitz cell for the fcc lattice.}\label{fig:1}
  \end{figure}
The notation shown here is standard notation going back to the classic paper
analyzing the irreducible representations of the space group in~\cite{BSW}.

Another useful way to think about the Brillouin torus is as the quotient of
points $(k_1, k_2, k_3)\in \RR^3$ by the equivalence relation $k_i \sim
k_i+1$ \emph{and} $k \sim k + s$.  Thus, the Brillouin torus is a quotient of
the standard torus~$\RR^3/\ZZ^3$ by a free $\ZZ/2\ZZ$-action.  In this
example, $\phi=t=c=1$.  The extension~\eqref{eq:6} for the crystallographic
group of the fcc lattice splits, but that for the diamond structure is not
split, i.e., it is non-symmorphic.  In both cases we have the sequence:
\begin{equation}\label{eq:diamond-sequence}
1 \to D_3 \to G(\bar C)  \to \zt \times O_h \to 1
\end{equation}
 where $G(\bar C) := G(C)/U$, the extended point group is $\zt \times O_h$
and the point group is $P=O_h$, the cubic group.  The latter can be regarded
as $S_4 \times \ZZ/2\ZZ$, where $S_4$ is the Weyl group of $so(6)=su(4)$ and
the $\ZZ/2\ZZ$ acts on $\RR^3$ by inversion $x\to -x$. For our purposes a
more useful characterization is $O_h \cong (\ZZ/2\ZZ)^3 \ltimes S_3$ where
$S_3$ acts on $(x_1,x_2,x_3)$ by permutation and the $(\ZZ/2\ZZ)^3$ subgroup
is generated by $\epsilon_1,\epsilon_2,\epsilon_3$, where $\epsilon_i$
 flips the sign of $x_i$, leaving
the other coordinates fixed.

It is now convenient to introduce the Seitz notation for elements of the
Euclidean group (the group $Euc(E)$ of~\S\ref{sec:1}): Again we choose an
origin to identify $\EE^d \cong \RR^d$ and then define $\{ R\vert v \}$ by
\begin{equation}
\{ R\vert v \}x := Rx + v,
\end{equation}
where $Rx$ means $R^i_{~j}x^j$, $R\in O(d)$.

To see that the space group of the diamond structure is non-symmorphic
(i.e. \eqref{eq:diamond-sequence} does not split) we note that 24 elements of $O_h$ preserve
$(\sum_i 2x_i) \mod 4$, and hence can be lifted $R \to \{ R\vert 0 \}$ in the space group
of the diamond crystal, but the remaining 24 elements shift $(\sum_i 2 x_i) \to [(\sum_i 2 x_i ) -1 ]\mod 4$ and hence
require a shift by an element of $s + D_3$ in order to preserve the crystal $\bar C$. We choose to lift such elements 
$R$ to the  screw-displacement $ \{ R \vert s \}$.

Let us now describe the groupoid $X_\Pi\gpd O_h$. The action of the point group $O_h$
on the Brillouin torus is defined by
\begin{equation}
(\{ R\vert v \}\cdot \lambda)(\xi ) := \lambda\left(\{ R\vert v \} \{ 1 \vert \xi  \} \{ R\vert v \}^{-1} \right)= \lambda(R \xi )
\end{equation}
where $\lambda \in X_\Pi$ is a character of $\Pi$ and $\xi \in \Pi $. It is expressed (nonuniquely!) in terms of 
 more standard $k$-vectors $k=(k_1,k_2,k_3)\in V^*\cong \RR^3$ by
$\lambda(\xi ) = \exp[ 2\pi  i \sum_{j=1}^3 k_j \xi ^j ] $. The $O_h$ action on $X_\Pi$ has nontrivial
stabilizer groups at special points, circles, and (two-dimensional) tori. These subvarieties with enhanced
symmetry form orbits under the $O_h$ action. The nontrivial submanifolds, together with their
stabilizer groups may be summarized by the following table:

\def\ve{\varepsilon}

\bigskip
\begin{center}
\begin{tabular}{|c|c| }
\hline Typical character $k$ &   Stabilizer group $G''(\lambda)$   \\
\hline $(y,y,y') $ &   $\ZZ/2\ZZ$   \\
\hline $(y,y',0) $ &  $\ZZ/2\ZZ$  \\
\hline $(y,0,0)$ &   $D_4$   \\
\hline $(y,0,1/2) $ &   $\ZZ/2\ZZ \times \ZZ/2\ZZ$ \\
\hline $(y,y,0) $ &   $\ZZ/2\ZZ\times\ZZ/2\ZZ$  \\
\hline $(y,y,y)$ &    $S_3$ \\
\hline $(1/4,1/4,1/4)$ &   $\ZZ/2\ZZ \times S_3$ \\
\hline $(1/2,0,0)$ &   $\ZZ/2\ZZ \times D_4$  \\
\hline $(0,0,0) $ &   $O_h$ \\
 \hline \end{tabular}\end{center}
\bigskip

\noindent
 The notation here is as follows: $y$ denotes a generic real number and we
identify $y \sim y+1$.  A prime indicates that the number $y'$ is different
from $y$ modulo $1$. For each example of fixed point loci shown there are
others obtained by the action of the point group on $X_\Pi$. For example, the
2-torus obtained by the closure of points of type $(y,y,y')$ can be mapped to
a 2-torus of points $(\bar y, y' ,y)$. Similarly, the point $(1/2,0,0)$ sits
in an orbit of three points, the other two being represented by $(0,1/2,0)$
and $(0,0,1/2)$. These three points are known as $X$-points in the notation
of~\cite{BSW}.  The other special points are the four $L$-points: $(1/4, 1/4,
1/4)$ and an orbit of $(-1/4, 1/4,1/4)$ under the point group, together with
the trivial character $k=0$.  The trivial character is usually denoted as the
$\Gamma$-point in the condensed matter literature.

We can detect if the canonical \ptzgce\ is nontrivial by pulling
back to these special loci and determining whether the central
extension of~\eqref{eq:106} is trivializable or not. A case-by-case
analysis reveals that the nontrivial central extensions occur only on the
three circles with generic points $(y,0,1/2)$, $(0,y,1/2)$, and
$(0,1/2,y)$. For example, the point $(y,0,1/2)$, with $y$
generic has stabilizer group $\ZZ/2\ZZ \times \ZZ/2\ZZ$ 
generated by $\epsilon_2$ and $\epsilon_3$.
Their lifts to $G(\bar C)$ are $\{ \epsilon_2 \vert s \}$ and
$\{ \epsilon_3 \vert s \}$, whose group commutator is
\begin{equation}
[\{ \epsilon_2 \vert s \}, \{ \epsilon_3 \vert s \}] = \{ 1 \vert \xi   \}
\end{equation}
 with $\xi  = (0,-1,1)$, and for this lattice vector $\lambda(\xi )=-1$.  So,
by the same reasoning as in Example~\ref{thm:104} the extension $H(\lambda)$
of $G''(\lambda)$ is nontrivial. At the enhanced symmetry points $X$, which
lie on these three circles the stabilizer group is $\zt \times D_4$ and the
central extension remains nontrivial. For example, for the point $(1/2,0,0)$
the commutator function (that is, the commutator of the lift of elements to
$H(\lambda)$), which can be shown to characterize the central extension up to
isomorphism, is nontrivial for the pairs $s(\epsilon_1, \epsilon_2) =
s(\epsilon_1, \epsilon_3)=-1$, (and trivial on other pairs).

There are also examples of materials in nature with nontrivial time-reversal
symmetries leading to nontrivial extensions of the Brillouin groupoid, such
as those described in Theorem~\ref{thm:51}.  An example is probably provided
by manganese dioxide in its rutile structure.\footnote{We thank K. Rabe for
this suggestion.}  The space group (number 136) is nonsymmorphic, with a
point group isomorphic to $\zt \times D_4$.  Each manganese atom is
surrounded by a regular octahedron of oxygen atoms.  There are unpaired
$d$-electrons on the manganese atoms which provide local magnetic moments and
these moments can be staggered so that time reversal must be accompanied by a
nontrivial element of the space group. Both sequences~\eqref{eq:6}
and~\eqref{eq:7} are nonsplit.  We leave a detailed analysis of the \ptzgce\
and equivariant $K$-theory for another occasion.

\makeatletter
\let\sectionname\@empty
\makeatother

 \bigskip\bigskip
\bibliographystyle{hyperamsalpha} 
\addcontentsline{toc}{section}{References}  

\newcommand{\etalchar}[1]{$^{#1}$}
\providecommand{\bysame}{\leavevmode\hbox to3em{\hrulefill}\thinspace}
\providecommand{\MR}{\relax\ifhmode\unskip\space\fi MR }
\providecommand{\MRhref}[2]{%
  \href{http://www.ams.org/mathscinet-getitem?mr=#1}{#2}
}
\providecommand{\href}[2]{#2}

  \end{document}